\documentclass{article}
\usepackage{geometry}
\usepackage{hyperref}
\usepackage{lmodern}
\usepackage{animate}
\usepackage{tikz}
\usepackage{tikz-cd}
\usetikzlibrary{shapes.misc, positioning}
\usetikzlibrary{arrows,shapes}
\usetikzlibrary{babel}
\usepackage{mathrsfs,amsmath}
\usepackage{amsthm}
\usepackage{physics}
\usepackage{amsfonts}
\usepackage{amssymb}
\usepackage{graphicx}
\usepackage{wrapfig}
\usepackage{cancel}
\usepackage{braket}
\usepackage[all]{xy}
\usepackage{xcolor}
\usepackage{adjustbox}
\usepackage{subcaption}
\usepackage{capt-of}
\usepackage{tabu}
\makeatletter
\def\fdsy@scale{1.8}
\newcommand\fdsy@mweight@normal{Book}
\newcommand\fdsy@mweight@small{Book}
\newcommand\fdsy@bweight@normal{Medium}
\newcommand\fdsy@bweight@small{Medium}
\DeclareFontFamily{U}{FdSymbolC}{}
\DeclareFontShape{U}{FdSymbolC}{m}{n}{
    <-7.1> s * [\fdsy@scale] FdSymbolC-\fdsy@mweight@small
    <7.1-> s * [\fdsy@scale] FdSymbolC-\fdsy@mweight@normal
}{}
\DeclareFontShape{U}{FdSymbolC}{b}{n}{
    <-7.1> s * [\fdsy@scale] FdSymbolC-\fdsy@bweight@small
    <7.1-> s * [\fdsy@scale] FdSymbolC-\fdsy@bweight@normal
}{}
\makeatother
\DeclareFontFamily{U}{FdSymbolA}{}
\DeclareFontShape{U}{FdSymbolA}{m}{n}{<->FdSymbolA-Book}{}
\DeclareSymbolFont{extrasymbols}{U}{FdSymbolA}{m}{n}
\DeclareMathSymbol{\vardiamondsuit}{\mathord}{extrasymbols}{182}
\DeclareMathSymbol{\varheartsuit}{\mathord}{extrasymbols}{184}

\DeclareMathOperator{\sgn}{sgn}
\definecolor{color1}{RGB}{12,18,155}
\definecolor{color2}{RGB}{83,89,155}
\definecolor{color3}{RGB}{75,80,85}
\newtheorem{theorem}{Theorem}[section]
\newtheorem{corollary}{Corollary}[theorem]
\newtheorem{lemma}[theorem]{Lemma}
\newtheorem{remark}{Remark}

\newcommand\scalemath[2]{\scalebox{#1}{\mbox{\ensuremath{\displaystyle #2}}}}

\begin{document}

\title{\textbf{A note on Haag duality}}

\author{Alan Garbarz$^\diamondsuit{}^\perp$, Gabriel Palau$^\diamondsuit$}
\date{}
\maketitle

\vspace{.25cm}

\begin{minipage}{.9\textwidth}\small \it 
	\begin{center}
    $^\diamondsuit$ Departamento de F\'isica-FCEN-UBA ,
	Ciudad Universitaria, pabell\'on 1, 1428, Buenos Aires, Argentina.
     \end{center}
\end{minipage}

\vspace{.25cm}

\begin{minipage}{.9\textwidth}\small \it 
	\begin{center}
    $^\perp$ IFIBA-CONICET,
	Ciudad Universitaria, pabell\'on 1, 1428, Buenos Aires, Argentina.
     \end{center}
\end{minipage}
\vspace{.5cm}

\maketitle

\begin{abstract}
    Haag duality is a remarkable property in QFT stating that the commutant of the algebra of observables localized in some region of spacetime is exactly the algebra associated to the causally disconnected region. It is a strong condition on the local structure and has direct consequences on entanglement measures. It was first shown to hold for a free scalar field and causal diamonds by Araki in 1964 and later by many authors in different ways. In particular, Eckmann and Osterwalder (EO) used Tomita-Takesaki modular theory to give a direct proof. However, it is not straightforward to relate this proof to the works of Araki, since they rely on two forms of the canonical commutation relations (CCR), called Segal and Weyl formulations, while EO work as starting point assumes that duality holds in the so-called ``first quantization'' in the Weyl formulation. It is our purpose to first introduce the works of Araki in a more easy-to-read but still rigorous and self-contained fashion, and show how Haag duality is stated in the Segal and Weyl formulations and in both first and second quantizations (and their immediate combination). This permits to understand the setting of the EO proof of Haag duality. There is nothing essentially new in this manuscript, with the exception of what we consider a simplification of EO proof that uses the adjoint $S^*$ of the Tomita operator $S$ instead of introducing several auxiliary operators. We hope this note will be useful for those seeking to understand where Haag duality comes from in a free scalar QFT.
    
    \begin{flushleft}
\hrulefill\\
\footnotesize
{E-mails:  alan@df.uba.ar, gpalau@df.uba.ar}
\end{flushleft}
\end{abstract}

\pagebreak

\tableofcontents

\pagebreak

\section{Introduction}

We are mainly interested in the property of Haag duality in QFT, which is related to the local  structure of the theory and superselection sectors \cite{haag1996local}, and this has important consequences on entanglement measures in QFT \cite{Arias:2018tmw,Kawahigashi:1999jz}. It is known to hold for a free scalar field on causal diamonds ( see \cite{araki1964neumann,dell'antonio1968, hislop1986simple, eckmann1973application}), free fermions (in a twisted version, see \cite{dell1968structure}) and the elecromagnetic field (for simply connected regions \cite{Leyland:1978iv}). It also holds in conformal field theories (see \cite{brunetti1993modular}) and in a form relative to globally-hyperbolic submanifolds of Minkowski \cite{camassa2007relative}.
It has been famously shown to hold for Rindler wedges in \cite{Bisognano:1975ih,Bisognano:1976za}.

In order to state the duality, it is first adequate to introduce the general philosophy of Algebraic Quantum Field Theory (AQFT) in a nutshell: to each open region of spacetime, an algebra of operators is assigned. This assignment is called \textit{a net} $N$. If we have a region $\mathcal{O}\subseteq \mathbb{M}^{4}$, then we call $N(\mathcal{O})$ the algebra of bounded operators localized in $\mathcal{O}$ . Causality dictates that given a region $\mathcal{O}$, the operators in $N(\mathcal{O'})$, i.e. those localized in the causal complement\footnote{$\mathcal{O}'$ is the region causally disconnected from $\mathcal{O}$. See \eqref{prima} for a precise definition. } $\mathcal{O}'$, commute with those in $N(\mathcal{O})$. In other words, $N(\mathcal{O'}) \subset N(\mathcal{O})'$, where prime in the algebra means commutant\footnote{Given a subset $\mathcal{M}\subseteq B(H)$ of all linear bounded operators acting on a Hilbert space $H$, its commutant is the set $\mathcal{M}'=\{x'\in B(H): x' x=x x' \text{ for all } x\in \mathcal{M}\}$.}.    
Haag duality is the property that the commutant of $N(\mathcal{O})$ coincides with the algebra of operators localized in $\mathcal{O'}$, namely 
\[N(\mathcal{O})'=N(\mathcal{O'})\qquad\qquad \text{(Haag duality)}\]

To the best of our knowledge, the first proof of Haag duality for the free scalar field was given by Araki. In \cite{araki1963lattice} two possible forms of the canonical commutation relations (CCR) algebra are introduced. They are called  Weyl  and   Segal formulations. In \cite{araki1963lattice} several properties of these CCR algebras are studied, such as the existence of a cyclic vector in the vacuum representation. However, here we are going to motivate and explain their formulations and how they relate to each other, and only discuss the properties that we need to arrive at the different presentations of Haag duality and its proof. 

Let us give more detail on what the reader will find in this manuscript. 
First, in Section 2, we begin with the very basics: the free scalar field theory and its space of solutions. The one-particle space of square-integrable functions on the positive mass hyperboloid is immediately introduced and allows to formulate the ubiquitous CCR $[a(f),a(g)^{\ast}]=(f,g)$, which can be represented by the creation and annihilation operators acting on the Fock space of such one-particle space. Then we introduce Segal's formulation of CCR, where the one-particle space is made of  real spacetime functions, and gives rise via  a representation to unitaries\footnote{In many places in the literature these unitaries are referred as Weyl unitaries instead. We follow the nomenclature of Araki which reserves the name Weyl unitaries for other unitary operators. } $W_{F}(f)$ that act on the same Fock space as before and obey an exponentiated form of the CCR. From the Segal one-particle space of real spacetime functions, it can be obtained a vector space of initial conditions. This motivates the definition of the Weyl CCR algebra and a representation by unitaries $U_{F}(f)$ and $V_{F}(g)$ which also act on the same Fock space.  Such unitaries are generated by the value of the field at some fixed time (say $t=0$) $\varphi(0,\vec{x})$ and its normal future-directed derivative $\pi(0,\vec{x})=\frac{\partial\varphi(t,\vec{x})}{\partial t}|_{t=0}$. 

In Section 3 we begin to dig deeper and present the mathematical construction of the one-particle symplectic vector spaces for the Segal and Weyl  formulations. In particular we spend some time clarifying the relation between the spacetime functions in the Segal construction and the pairs of  initial conditions of the Weyl construction, which are functions on a spacelike hypersurface.   

In Section 4 we present the so-called \textit{first quantization map}. This is a map $\mathsf{S}$ that assigns to open spacetime regions $\mathcal{O}$ a closed real subspace $\mathsf{S}(\mathcal{O})$ of a one-particle space: 
\[\mathsf{S}: \text{Open spacetime regions} \rightarrow \text{Closed real subspace of 1-particle space} \] 
We define these maps for the Segal and Weyl contexts and explain how Araki proves the duality property for them, which roughly reads $\mathsf{S}(\mathcal{O})'=\mathsf{S}(\mathcal{O'})$. We will explain, for each formulation, what $'$ refers to when applied to a closed vector subspace.    

In Section 5 we introduce the so-called \textit{second quantization map}. This map assigns to the closed vector spaces $\mathsf{S}$ of the first quantization a von Neumann algebra $R(\mathsf{S})$ of bounded operators on the corresponding Fock space:
\[R:\text{Closed real subspace of 1-particle space}\rightarrow \text{von Neumann algebra} \] 
Here we leave the guide of Araki's works in order to make direct contact with the starting point of the article by Eckmann and Osterwalder \cite{eckmann1973application}. There they proved Haag duality for the second quantization map with the help of the power of Tomita's modular theory. This is a direct way to obtain the duality, although we believe some details in that reference deserve better explanation or could be simplified. So in section 5 we follow \cite{eckmann1973application} but  at some instances propose different arguments that seem more clear to us, and specially we do not introduce several auxiliary operators. Instead, once the modular operator $S$ is introduced and studied, we can do similar computations for its adjoint $S^*$ that allow to reach the final steps of the proof more directly. Finally, in Section 5 we include a description of a counterexample of the duality constructed originally by Araki in \cite{araki1964neumann}. 

The reader interested solely in the proof of Haag duality can start at section 4, or even at section 5 if only interested in the utilization of Tomita's modular theory.  Many computations and technical clarifications are relegated to appendices. We follow most of the times the notation and conventions of Araki in \cite{araki1964neumann}.

\section{CCR representations in a nutshell}

Starting from the equation of motion for the real scalar field, there are a number of ways of present the CCR. Araki in \cite{araki1963lattice} and \cite{araki1964neumann} uses the usual Fock space, and at the same time both Segal and Weyl formulations as well. It will be important for us to explain them in detail, including the corresponding vacuum representation, and in this Section we give an overview of their constructions.

The real scalar field of mass $m>0$ satisfies the Klein-Gordon equation 
\begin{equation}\label{kleing}
	(\square + m^{2})\phi(x)=0.
\end{equation}
Fourier transforming this equation (see \eqref{fouriertransform} below) 
\begin{equation*}
	(-p^2 + m^{2})\mathscr{F}(\phi)(p)=0,
\end{equation*}
shows $\mathscr{F}(\phi)(p)$ is supported on  the union of  hyperboloids $\{ p\in\mathbb{R}^{4}| p^{2}=m^{2}\}$ but if $\phi$ is real $\mathscr{F}(\phi)(-p)=\mathscr{F}(\phi)(p)^{\ast}$. Then $\phi$ is completely defined by the value of $\mathscr{F}(\phi)$ over the positive mass hyperboloid $H_m:=\left\{p^2=m^2 |\quad p^{0}>0\right\}$. This motivates the following lines in which we follow \cite{palmer1978symplectic}. If $\eta: H_{m}\cup-H_{m}\to \mathbb{C}$ satisfies $\eta(-p)=\eta(p)^{\ast}$, we can obtain the following weak solution\footnote{ A weak solution of \eqref{kleing} is a distribution  $\Phi$ (tempered for us) such that $\Phi((\square + m^{2})f)=0$ for all $f\in \mathcal{S}(\mathbb{R}^{4})$.} of \eqref{kleing}
\begin{equation}\label{efe}
 	F_{\eta}(x)=\frac{-i}{(2\pi)^{\frac{3}{2}}}\int_{\mathbb{R}^{4}} e^{-ipx}\delta(p^2-m^2)\sgn(p^{0})\eta(p) d^{4}p.
\end{equation} 
If in addition $\eta=E(h)$, with $E(h)=\mathscr{F}(h)|_{H_{m}}$ and $h\in \mathcal{S}(\mathbb{R}^4,\mathbb{R})$ (real Schwartz function), then
\begin{equation}\label{suave}
    F_{E(h)}(x)=\int_{\mathbb{R}^{4}} \Delta(x-y)h(y)dy=\frac{-i}{(2\pi)^{\frac{3}{2}}}\int_{\mathbb{R}^{4}} e^{-ipx}\delta(p^2-m^2)\sgn(p^{0})\mathscr{F}(h)(p) d^{4}p,
\end{equation}
where $\Delta$ is the causal propagator\footnote{It can be defined in several equivalent ways. One is the one given by the second equality in \eqref{suave}. Another one is as $2\text{Re}\Delta^{(+)}$, see below \eqref{prodintespaciales}.}. Note that if the Fourier transform of $h$ vanishes when restricted to the hyperboloid, we have a trivial solution. In addition, if $h$ is a compactly supported smooth function the solution is a smooth function, see Theorem 1 in \cite{dimock1980algebras}.

The image under the map $F$ above of 
\begin{equation*}
    \{ h:H_{m}\cup-H_{m}\longmapsto \mathbb{C} |\quad  h(-p)=h(p)^{\ast}\text{ and }\int_{\mathbb{R}^{4}} |h(p)|^2\delta(p^2-m^2)d^{4}p<\infty \}\simeq L^{2}\left(\mathbb{R}^{3},\frac{d^3p}{\omega_{p}}\right)
\end{equation*} 
(where $\simeq$ is given by taking $p^{0}=\omega_p:=\sqrt{m^2+\vec{p}^2}$)  defines a (complex) Hilbert space of weak solutions of the Klein-Gordon equation. In what follows we introduce in detail different (real and complex) Hilbert spaces  associated to the space of (weak) Klein-Gordon fields.

\subsection{The complex one-particle Hilbert space \texorpdfstring{$L$}{}}\label{oneparticlespacesss}
Let us consider the smooth and real spacetime Schwartz functions\footnote{They carry a representation $U(a,\Lambda)$  of ISO$(3,1)$ such that, $
    [U(a,\Lambda)f](x) =f(\Lambda^{-1}(x-a)), \quad f \in \mathcal{S}(\mathbb{R}^4,\mathbb{R})$. Equivalently $\label{unirrep}
    [U(a,\Lambda)\mathscr{F}(f)](p)=e^{ip\cdot a} \mathscr{F}(f)(\Lambda^{-1}p)$. With the complex inner product of \eqref{ele} the representations become unitary.\label{footnoteirrep}} $\mathcal{S}(\mathbb{R}^4,\mathbb{R})$. Let us also introduce their Fourier transform,
\begin{equation}
    \mathscr{F}(f)(p):=(2\pi)^{-3/2}\int_{\mathbb{R}^{4}} dx f(x)e^{i p\cdot x} 
    \label{fouriertransform}
\end{equation}
We can define the following  symmetric semi-definite positive bi-linear form on 
 $\mathcal{S}(\mathbb{R}^{4},\mathbb{R})$ with values in $\mathbb{C}$,
\begin{equation}\label{prodintespaciales}
	(h_{1},h_{2})_{\mathcal{S}}=2i\int_{\mathbb{R}^{4}\times\mathbb{R}^{4}} h_{1}(x)\Delta^{(+)}(x-y)h_{2}(y)dx dy,
\end{equation}
where $\Delta^{(+)}(x)=\frac{-i}{(2\pi)^{3}}\int e^{-ipx} \delta(p^2-m^2)\theta(p^0) d^4p$ and $p\cdot x=p^0x^0-\vec{p}\cdot\vec{x}$. This is nothing but the Lorentz-invariant product on the mass hyperboloid $H_m$. This can be seen by noticing
\begin{align}\label{prodintmomento}
	(h_{1},h_{2})_{\mathcal{S}} & =2\int_{\mathbb{R}^{4}} \mathscr{F}(h_{1})^{\ast}(p)\mathscr{F}(h_{2})(p) \delta(p^2-m^2)\theta(p^0) d^4p \\
	        & =\int_{\mathbb{R}^{3}} \mathscr{F}(h_{1})^{\ast}(\omega(\vec{p}),\vec{p})\mathscr{F}(h_{2})(\omega(\vec{p}),\vec{p}) \frac{d^3p}{\omega(\vec{p})},
\end{align}
where $\omega(\vec{p})=\sqrt{m^2+\vec{p}^2}=:\omega_p$. There is a set of functions such that $(h,h)_\mathcal{S}=0$, namely those with Fourier transform that vanishes on $H_m$. We denote this set by $\mathcal{S}_{0}$. In order to obtain the complex one-particle Hilbert space, we first of all quotient by $\mathcal{S}_0$. In $\mathcal{S}(\mathbb{R}^4,\mathbb{R})/\mathcal{S}_{0}$ it is possible to define a multiplication by $i$. This can be achieved by imposing  the condition $h'=ih$ if $\mathscr{F}(h')|_{H_m}=i \mathscr{F}(h)|_{H_m}$, where $h$ and $h'$ are real Schwartz functions on $\mathbb{R}^4$. We relegate the independence of this definition with  respect to the chosen representative element to the subsection \ref{good definition}.  Note that although it may seem that multiplying by $i$ the Fourier transform of a real function necessarily makes the new function complex-valued (when back-transformed to coordinate space), the fact that this condition is imposed on $H_m$ resolves this puzzle. We explain how this works with an example in Appendix \ref{ejemplo}.

A multiplication by $i$ as just  defined turns $(,)_{\mathcal{S}}$ into a Hermitian inner product in $\mathcal{S}(\mathbb{R}^4,\mathbb{R})/\mathcal{S}_{0}$. By completing $\mathcal{S}(\mathbb{R}^4,\mathbb{R})/\mathcal{S}_{0}$ with respect to this complex inner product we arrive to a complex Hilbert space denoted by $L$,
\begin{equation}\label{ele}
    L:=\overline{\mathcal{S}(\mathbb{R}^4,\mathbb{R})/\mathcal{S}_{0}}
\end{equation}
This is understood to be equipped with the multiplication by $i$ defined above, and we denote the complex inner product by $(,)_{L}$.

The precise way to relate $L$ with the space of square-integrable functions on the mass hyperboloid is as follows. The map $E:\mathcal{S}(\mathbb{R}^4,\mathbb{R})\to L_{2}\left(\mathbb{R}^{3},\frac{d^3p}{\omega_{p}}\right)$ given by $E(h)=\mathscr{F}(h)|_{H_{m}}$ has dense range (see \cite{reedsimonvol2} Chapter X) and Ker$(E)=\mathcal{S}_{0}$, then $E$ passes to the quotient as an injective map. So we have 
\begin{equation*}
   L=\overline{\mathcal{S}(\mathbb{R}^4,\mathbb{R})/\mathcal{S}_{0}}\simeq L^{2}\left(\mathbb{R}^{3},\frac{d^3p}{\omega_{p}}\right),    
\end{equation*}
where also the complex structures are compatible, due to the definition of $i$ in $L$. As a conclusion, given $[h]\in \mathcal{S}(\mathbb{R}^4,\mathbb{R})/\mathcal{S}_{0}$, or thought as an element of $L^2\left(\mathbb{R}^{3},\frac{d^3p}{\omega_{p}}\right)$, we have, through the map $F_{E(h)}(x)$ in \eqref{suave}, a weak solution to the Klein-Gordon equation.

\subsection{Test spaces \texorpdfstring{$H$}{} and \texorpdfstring{$K$}{}}

Before defining the CCR-algebras we present the real Hilbert test spaces on which they are modeled. Both real test spaces we construct here can be related to the complex one-particle Hilbert space $(L,(,)_{L})$. The first test space $H$ is the one of real spacetime functions which at the end will help to construct smeared fields over spacetime regions. The second one, $K$, will be shown in Section 3 to be intimately related to the space of conjugate momenta at $t=0$ (and $\beta K$ is related to the other initial condition, the field at $t=0$).

\subsubsection{The real one-particle Hilbert space \texorpdfstring{$H$}{}}\label{segalspace}

Starting from the complex Hilbert space\
\footnote{What follows is a general construction of a real Hilbert space with complex structure $\beta$ from an arbitrary complex Hilbert space. But in order to make contact from the start with what we have been discussing previously, we have in mind the complex Hilbert space $L$.} $L$ viewed as densely generated by Schwartz real functions modulo $\mathcal{S}_0$, a real Hilbert  space $H$ together with a complex structure $\beta$ can be defined\footnote{In \cite{araki1964neumann} first a real Hilbert space $\hat{L}$ is defined just as we did for $H$ and the complex structure is given by $i$, borrowed from $L$. Then $H$ and $\beta$ are defined as we did and in the end $(\hat{L},i)$ is isomorphic to $(H,\beta)$. We are not going to make this distinction.}, namely an operator acting on $H$ such that:
\begin{equation}\label{betadef}
	\beta^{\ast}=-\beta \hspace{2cm} \beta^{2}=-1.
\end{equation}
First of all we consider $L$ as a real vector space and the real inner product given by $\Re{(,)_{L}}$. It is important that $L$ is complete with respect to it. Then we call $H=L$ as real vector spaces and we use in $H$ the inner product given by $(,)_{H}=\Re{(,)_{L}}$.
Note that given $h\in L$ in the real Hilbert space $H$ the elements $h$ and $ih$ are linearly independent while they are clearly linearly dependent as elements of $L$ (seen as a complex space).  The real inner product is 
\begin{equation}\label{ddd}
    (h_1,h_2)_H:=\Re{(h_1,h_2)_L} 
\end{equation} 
and $\beta$ is defined by
\begin{equation}
    (h_1,\beta h_2)_H:=-\Im{(h_1,h_2)_L}
    \label{betadefH}
\end{equation}
We must convince ourselves that  $\beta^2=-1$. This is possible by noticing that $\beta$ is the operator in $H$ that maps $h$ to $ih$, as $(h_1,\beta h_2)_H=-\Im{(h_1,h_2)_L}=\Re{i(h_1,h_2)_{L}}=\Re{(h_1,i h_2)_{L}}=(h_1,ih_2)_{H}$, for any $h_1$ and $h_2$. Therefore, $\beta$ can be identified with multiplication by the scalar $i$ in $L$. Then $\beta^2 h=\beta i h=i^2 h=-h$. We can also show that $\beta$ is compatible with the inner product, namely that the adjoint $\beta^*$ is equal to $-\beta$: $(-\beta h_1,h_2)_H=-(\beta h_1,h_2)_H=-( h_2,\beta h_1)_H=\Im{(h_2,h_1)_L}=\Im{\overline{(h_1,h_2)_L}}=-\Im{(h_1,h_2)_L}=(h_1,\beta h_2)_H$. In short, $\beta^{-1}=-\beta=\beta^*$.  
The role of the $\beta$ operator is to implement a distinction between the ``position and velocity'' initial conditions which are behind the idea of the Weyl formulation, as will be explained in more detail in the following sections.

The imaginary part of the inner product in $L$ gives a symplectic structure $\sigma$ in $H$
\begin{equation}
    \sigma(h_1,h_2):=\text{Im}(h_1,h_2)_L=-(h_1,\beta h_2).
        \label{sigma}
\end{equation}
In short when we talk about $H$ we will be talking about a real Hilbert space with inner product denoted by $(,)_{H}$ with an operator $\beta$ satisfying \eqref{betadef}.

It is also true that given $(H,\beta)$ a real Hilbert space with $\beta$ an operator satisfying \eqref{betadef}, it is possible to define a complex Hilbert space $L$. As a set  $L=H$. Guided by \eqref{ddd} and \eqref{betadefH} we define 
\begin{equation}\label{siete}
	(h_{1},h_{2})_{L}=(h_{1},h_{2})_{H}-i(h_{1},\beta h_{2})_{H}.
\end{equation} 
The defining properties of $\beta$ are used to prove that the equation above  defines an inner product in $L$. But in order to show this we need first to define the action $\mathbb{C}\curvearrowright L$, which is $(a+bi)h:= ah+b\beta h$. 

Now we want to see \eqref{siete} is sesquilinear. It pulls scalars out because the inner product in $H$ is $\mathbb{R}$-bilinear. Let us see what happens with multiplication by $i$ inside the inner product of $L$, 
\begin{align}
	(ih_{1},h_{2})_{L}&=(\beta h_{1},h_{2})_{H}-i(\beta h_{1},\beta h_{2})_{H}\\
	&=(h_{1},-\beta h_{2})_{H}-i(h_{1},h_{2})_{H} \\
	&=(-i)\left[(h_{1},h_{2})_{H}-i(h_{1},\beta h_{2})_{H}  \right] =(-i)(h_{1},h_{2})_{L},	
\end{align}
where \eqref{betadef} and $\mathbb{R}$-bilinearity were used. Then the behavior of the multiplication by $i$ is the one we expect. Similarly one can see that definition \eqref{siete} is $\mathbb{C}$-linear in the second argument.  
 
\subsubsection{The real Hilbert space \texorpdfstring{$K$}{} of initial conditions} \label{ka}

The real Hilbert space $K$ is a subspace $K\subseteq H$ that satisfies $H=K\oplus_{\mathbb{R}} \beta K $ (and also $L=K+iK$ is its complexification). This property does not give us a unique $K$, for details on a general way to construct such a $K$ the reader can see \cite{araki1963lattice}. 

Instead of that here we make an explicit connection between Segal's real Hilbert space $H$ and the space $K$ that we are going to use when we talk about Weyl's CCR-algebra. We start by defining $K$:
 \begin{equation*}
     K:=\overline{\left\{ [h] \in H |\quad  h\in \mathcal{S}(\mathbb{R}^4,\mathbb{R}) \,\, \text{is even in }x^0\right\}}^{\| \|_{H}}
 \end{equation*}
We will show that $\beta K$ consists of the (completion of the) subspace of odd functions in $x^0$, so that $H=K+\beta K$. Even more, we will prove that $K\perp_{H}\beta K$, and thus $H=K\oplus_{\mathbb{R}}\beta K$, a fact that will be crucial to establish the validity of Haag duality. As subspaces of $H$, both $K$ and $\beta K$ inherit the 
inner product of $H$.

Let us start by recalling some elementary facts. Any $f\in \mathcal{S}(\mathbb{R}^4,\mathbb{R})$ can be written in a unique way as a sum of its even and odd parts in $x^0$, $f=f_{+}+f_{-}$, with 
\begin{equation*}
	f_{+}(x)=\frac{1}{2}\big( f(x^0,\vec{x})+f(-x^0,\vec{x}) \big), \hspace{1cm}	f_{-}(x)=\frac{1}{2}\big( f(x^0,\vec{x})-f(-x^0,\vec{x}) \big), 
\end{equation*} 
In terms of the Fourier transforms this means 
\begin{equation}\label{parimparmomento}
\mathscr{F}(f_{+})(p)=\frac{1}{2}\big( \mathscr{F}(f)(p^0,\vec{p})+\mathscr{F}(f)(p^0,-\vec{p})^{\ast} \big), \hspace{.8cm} \mathscr{F}(f_{-})(p)=\frac{1}{2}\big( \mathscr{F}(f)(p^0,\vec{p})-\mathscr{F}(f)(p^0,-\vec{p})^{\ast} \big). 
\end{equation} 

On the other hand, the operator $\beta$ is better understood in momentum space. Let us recall its definition
\begin{align*}
	\mathscr{F}(\beta f)(p)&=i\mathscr{F}(f)(p) \text{ for }p\in H_{m}  \text{ or } \\
    \mathscr{F}(\beta f)(p)&=i\mathscr{F}(f)(p)\eta(p^0), 
\end{align*}
where $\eta(p^{0})$ is an infinitely differentiable odd real function such that $\eta(p^0)=1$ if $p^0\geq m$. Note that setting $p\in H_{m}$ in the second line we recover the definition of $\beta$. We  will come back to this in section \ref{beta}.
When we consider an even function $f$ in $x^0$, namely a function in $K\cap \mathcal{S}(\mathbb{R}^4,\mathbb{R})$, and taking into account \eqref{parimparmomento} we have
\begin{align*}
	\mathscr{F}(\beta f)(p)&=i\mathscr{F}(f)(p)\eta(p^0)=i\eta(p^0)\frac{1}{2}\big( \mathscr{F}(f)(p^0,\vec{p})+\mathscr{F}(f)(p^0,-\vec{p})^{\ast} \big) \\
	 &=\frac{1}{2}\big( i\eta(p^0)\mathscr{F}(f)(p^0,\vec{p})-(i\eta(p^0)\mathscr{F}(f)(p^0,-\vec{p}))^{\ast} \big) \\
	 &=\frac{1}{2}\big( \mathscr{F}(\beta f)(p^0,\vec{p})-\mathscr{F}(\beta f)(p^0,-\vec{p})^* \big),
\end{align*}
where it was used that $\eta(p^0)$ is real.

A general element $[f]\in K$ is the limit of a sequence $\{f_n\}_{n\in \mathbb{N}}\subseteq\mathcal{S}(\mathbb{R}^4,\mathbb{R})$ of even functions in $x^{0}$, then by boundedness of $\beta$, $\beta [f]=\lim_{n}\beta f_{n}$. But as we saw $\{\beta f_{n}\}_{n\in \mathbb{N}}$ is a sequence of odd functions in $x^{0}$, then we have that  $\beta K\subseteq \overline{\{ [f]\in H |\quad f\in \mathcal{S}(\mathbb{R}^4,\mathbb{R})\text{ odd in } x^0 \}}^{\| \|_{H}}$.

In a similar  way one can show that  $ \{ [f]\in H | \quad f\in \mathcal{S}(\mathbb{R}^4,\mathbb{R}) \text{ odd in } x^0 \} \subseteq \beta K $. Let us assume that $f\in \mathcal{S}(\mathbb{R}^4,\mathbb{R})$ is odd in $x^0$,
\begin{align*}
	\mathscr{F}(f)(p)&=\frac{1}{2}\big( \mathscr{F}(f)(p^0,\vec{p})-\mathscr{F}(f)(p^0,-\vec{p})^{\ast} \big) \\
	&=-\frac{1}{2}\big( \mathscr{F}(\beta^2 f)(p^0,\vec{p})-\mathscr{F}(\beta^2 f)(p^0,-\vec{p})^{\ast} \big) \\
	 &=-\frac{i}{2}\eta(p^0)\big( \mathscr{F}(\beta f)(p^0,\vec{p})+\mathscr{F}(\beta f)(p^0,-\vec{p})^{\ast} \big) \\
	 &=i \eta(p^0) \mathscr{F}((-\beta f)_+)(p)= \mathscr{F}(\beta(-\beta f)_+)(p)
\end{align*}
which shows that $f\in \beta K$. Let us recall tha $\beta$ is a closed operator because its inverse is bounded, then $\beta K$ is a closed subspace of $H$ so we still have an inclusion if we take closure of $ \{ [f]\in H | \quad f\in \mathcal{S}(\mathbb{R}^4,\mathbb{R}) \text{ odd in } x^0 \} \subseteq \beta K $ in order to have $ \overline{\{ [f]\in H | \quad f\in \mathcal{S}(\mathbb{R}^4,\mathbb{R}) \text{ odd in } x^0 \}}^{\| \|_{H}} \subseteq \beta K $. Then  
\begin{align*}
K&=\overline{\left\{ [h] \in H |\quad  h\in \mathcal{S}(\mathbb{R}^4,\mathbb{R}) \,\, \text{is even in }x^0\right\}}^{\| \|_{H}} \\
    \beta K &= \overline{\{ [f]\in H | \quad f\in \mathcal{S}(\mathbb{R}^4,\mathbb{R}) \text{ odd in } x^0 \}}^{\| \|_{H}}
\end{align*}
From now on we will use the subspaces  $K$ and $\beta K$ to refer to even and odd functions of $x^0$ together with its limit points.

It is now time to show that $K\perp_{H} \beta K$, a very simple but important fact that will be extensively used later on. In order to do this let us first of all rewrite the inner product of $H$ in a more convenient way 
\begin{align}
	(h_{1},h_{2})_{H}&=\Re{(h_{1},h_{2})_{\mathcal{S}}}=\Im{i(h_{1},h_{2})_{\mathcal{S}}} \nonumber\\
	&=-\int_{\mathbb{R}^{4}\times \mathbb{R}^{4}} h_{1}(x) \Delta_{1}(x-y) h_{2}(y) dx dy, 
	\label{Hinner}
\end{align}
where we have defined $\Delta_{1}=2\Im{\Delta^{(+)}}$. Now, $\Delta_{1}$ can be expressed as follows,
\begin{align*}
	\Delta_{1}(z)&=2\Im{\Delta^{(+)}(z)}=-\frac{2}{(2\pi)^3}\int_{\mathbb{R}^{4}} \cos(pz)\delta(p^2-m^2)\theta(p^0) d^4 p \\
	   &=-\frac{1}{(2\pi)^3}\int_{\mathbb{R}^{3}} \frac{\cos(\omega_{p}z^0-\vec{p}\cdot \vec{z})}{\omega_p}d^3 p.
\end{align*}
Then the inner product between $h_{1}$ and $h_{2}$ is
\begin{align*}
	(h_{1},h_{2})_{H}&=-\int_{\mathbb{R}^4}\int_{\mathbb{R}^4} h_{1}(x)\Delta_{1}(x-y)h_{2}(y) dx dy \\
	&= \frac{1}{(2\pi)^3}\int_{\mathbb{R}^3}  \bigg[\int_{\mathbb{R}^4} \int_{\mathbb{R}^4} h_{1}(x)\cos(\omega_{p}(x^0-y^0)-\vec{p}\cdot (\vec{x}-\vec{y}))h_{2}(y) dx dy \bigg]  \frac{d^3 p}{\omega_{p}},\\
	&= \frac{1}{(2\pi)^3}\int_{\mathbb{R}^3}  \bigg[\int_{\mathbb{R}^4} \int_{\mathbb{R}^4} h_{1}(x)\cos(\omega_{p}(x^0-y^0))\cos(\vec{p}\cdot (\vec{x}-\vec{y}))h_{2}(y) dx dy \bigg]  \frac{d^3 p}{\omega_{p}},
\end{align*}
which, by separating the factor $\cos(\omega_{p}(x^0-y^0))$ in two terms, can be seen to be zero since either the integral in $x^0$ or the one in $y^0$ vanishes if $h_1$ is an even function in $x^0$ and $h_2$ an odd function in $x^0$. The assertion $(h_{1},h_{2})_{H}=0$ for $h_{1}\in K$ and $h_{2}\in \beta K$ follows by continuity of the inner product.

\begin{figure}[!h]
\begin{center}
\begin{adjustbox}{width=.8\textwidth}
\begin{tikzpicture}
	\filldraw[color=black, fill=white,  thick](2,-7) rectangle (12,-5);
	\filldraw[color=black, fill=white,  thick](0,-.5) rectangle (5,2);
	\filldraw[color=black, fill=white,  thick](5,-4) rectangle (9,-1.5);
	\filldraw[color=black, fill=white,  thick](9,-.5) rectangle (14,2);
	

	\path 
	    (7,-6.5) node(eom)  { $(\square+m^2)\Phi=0$}
	   (7,-5.5) node(sol) {\textbf{Space of real solutions of the Klein-Gordon equation}}
	      (2.5,1.5) node(weyl) {\textbf{Weyl}}
	      (2.5,.5) node(w) {$K=\overline{\left\{\text{even functions of } x^0  \right\}}$}
	      (2.5,0) node(w2) {$\subset H=K \oplus \beta K$}
	      (11.5,1.5) node(segal) {\textbf{Segal}}
	      (11.5,.5) node(s) {$H=\overline{\mathcal{S}(\mathbb{R}^4,\mathbb{R})/\mathcal{S}_0}$}
	      (11.5,0)  node(s2) {with c.s. $\beta$}
	      (7,-2) node(fock) {\textbf{$H_m$-space}}
	      (7,-3) node(f) {$ L^2(\mathbb{R}^3,\frac{d^3p}{\omega_p})$} 
	      (11.4,-1.6) node(E1) {$E$ after $\beta\equiv i$ }
	      (2.4,-1.6) node(E2) {$E$ on $K+iK$ }
	      (7,1) node(inc) {inclusion}
	      (7.2,-4.5) node(F) {$F$};
	
	\draw[thick, right hook ->]  (5,.75) -- (9,.75);
	\draw[thick, ->]  (11.5,-.5) -- (9,-2.5);
	\draw[thick, ->]  (2.5,-.5) -- (5,-2.5);
	\draw[thick, ->]  (7,-4) -- (7,-5);
\end{tikzpicture}
\end{adjustbox}
\end{center}
\caption{A schematic diagram showing that associated to the space of real (weak) solutions of the equation of motion, we can construct different Hilbert spaces. They will give different versions of the so-called first quantization. The 1-particle space on $H_m$ is a complex space, while the Segal and Weyl vector spaces are real, and Segal's has a complex structure (c.s.) $\beta$. In the definition of $H$, $\mathcal{S}_0  $ are the Schwartz functions whose Fourier transform vanishes on the mass hyperboloid. The map $E$ is given by $E(h)=\mathscr{F}(h)|_{H_m}$ which acts on the complex Hilbert space $L$, so we first need to set $\beta\equiv i$ as explained around \eqref{siete}. The map $F$ is defined in \eqref{efe}.}
\label{firstquantization}
\end{figure}
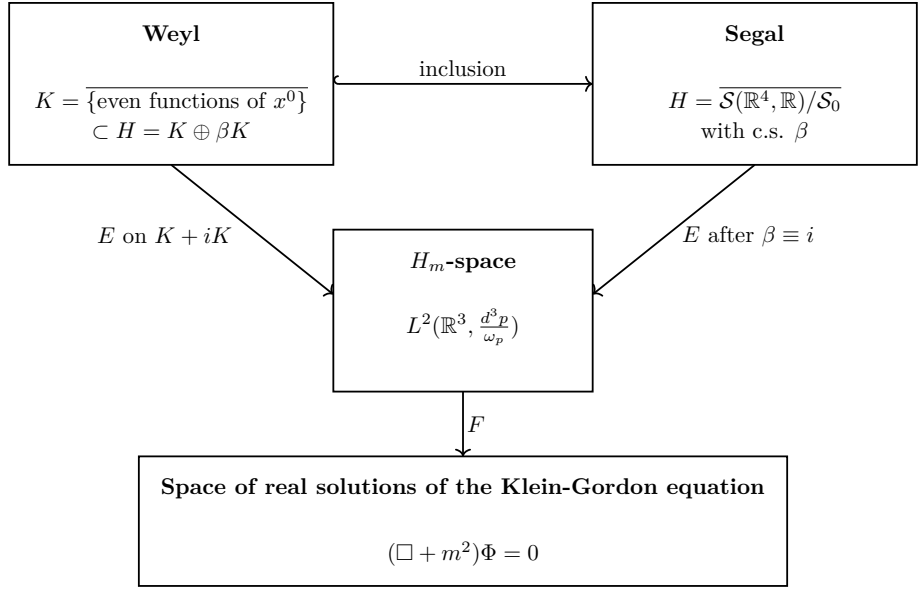

\subsection{Segal and Weyl CCR algebras}

The CCR-algebra in Segal form uses the real vector space $(H,\beta)$. It is the ${\ast}$-algebra generated by abstract elements $\{W(h): \text{ for all }h\in H\}$, with the involution given by $$W(f)^{\ast}=W(-f)$$ and satisfying the commutation relation 
\begin{equation}\label{ccrsegal}
    W(f)W(g)=e^{\frac{i}{2}(f,\beta g )_{H}}W(f+g)=W(f+g)e^{-\frac{i}{2}\sigma(f,g)},
\end{equation}
for all $f,g\in H$, where the identity element of the algebra is $W(0)=1_{S}$. Note that at this point the $W(f)$ are just abstract elements of Segal's CCR-algebra and are not acting on any additional Hibert space. We will construct a representation in the next section, but before that let us define similarly Weyl's CCR-algebra over $K$. 

The CCR-algebra in Weyl form, over $K$, is the ${\ast}$-algebra generated by abstract elements $\{U(f), V(g): \text{ for all }f,g\in K\}$, with the involution given by 
\begin{equation*}
    U(f)^{\ast}=U(-f) \hspace{1cm} V(g)^{\ast}=V(-g)
\end{equation*} and satisfying $U(0)=V(0)=1_{W} $ together with the commutation relation
\begin{equation}\label{weylrep}
    \begin{split}
       U(f_{1}+f_{2})&=U(f_{1})U(f_{2}) \\
	V(g_{1}+g_{2})&=V(g_{1})V(g_{2}) \\
	U(f_{1})V(g_{1})U(f_{2})V(g_{2})&=U(f_{1}+f_{2})V(g_{1}+g_{2})e^{i(f_{2},g_{1})_{K}},
    \end{split}
\end{equation}
with $1_{W}$ again the identity element of the algebra.

An important fact for us is that these two $\ast$-algebras are isomorphic at that level, with the isomorphism given by 
\begin{equation}
   \kappa: W(h) \longmapsto e^{\frac{i}{2}(f,g)_{K}}U(f)V(g),
   \label{isoWeylSegal}
\end{equation}
where $f$ and $g$, in $K$, are given by the relation $h=f+\beta g$, and \eqref{isoWeylSegal} is extended by $\mathbb{C}$-linearity. We will show how the commutation relations imply that $\kappa$ respects the product of the form $W(h_{1})W(h_{2})$, and for general products one must use $\mathbb{C}$-linearity.
So let us suppose that $h_{1}, h_{2}\in H$, then there are $f_1, g_1,f_2,g_2\in K$ such that $h_{1}=f_1+\beta g_1$ and $h_{2}=f_2+\beta g_2$
and obviously $h_{1}+h_{2}=(f_1+f_2)+\beta (g_1+g_2)$. Then using both commutation relations \eqref{ccrsegal} and \eqref{weylrep}, together with the fact that $K\perp_{H}\beta K$ and unitarity of $\beta$ one we can check by a straightforward computation that $\kappa(W(h_{1})W(h_{2}))=\kappa(W(h_{1}))\kappa(W(h_{2}))$. Moreover we can compute what $\kappa$ does with adjoints, $\kappa(W(h)^{\ast})=\kappa(W(h))^{\ast}$.
This shows that $\kappa$ is an isomorphism of $\ast$-algebras. The fact that $\kappa$ maps the identity element $1_{S}$ to $1_{W}$ is straightforward. It is possible to add more structure (a norm using the spectral radius, see \cite{haag1996local}) on these two algebras and show they are also isomorphic at this level, that is isomorphic as $C^{\ast}$-algebras, but we will not do such thing.

In order to give a representation of these two algebras as bounded operators acting on a complex Hilbert space, we first must construct such space and that is what we do in the following subsection. 

\subsection{The Fock space over \texorpdfstring{$L$}{}}\label{Fockrep}

From $L$ we construct the Fock space, denoted $\mathfrak{H}_{T}(L)$ following Araki. It is defined by  $\mathfrak{H}_{T}(L)=\overline{\bigoplus_{n=0}^{\infty} L^{\odot n}}$, where $\odot$ denotes the totally symmetric tensor product (i.e. $L^{\odot n}=$ Sym $L^{\otimes n}$). Given $h\in L$, we can build operators on $\mathfrak{H}_{T}(L)$ acting by
\begin{align}\label{creayaniq}
	a(h)(h_{1}\odot \cdots \odot h_{n})& =\frac{1}{\sqrt{n}}\sum_{j=1}^{n} (h,h_{j})_{L} h_{1}\odot \cdots \hat{h}_{j} \cdots \odot h_{n} \\
	a^{\ast}(h)(h_{1}\odot \cdots \odot h_{n})& =\sqrt{n+1} h\odot h_{1}\odot \cdots \odot h_{n}, 
\end{align}
taking the inner product in $L$ anti-linear in the first argument and where $\hat{h}_j$ means ``skipping the function $h_j$''. Note that over the union of the finite-occupation-number states they satisfy the well-known version of the CCR (see Figure \ref{secondquantization} below), 
\begin{equation}\label{conmut}
	[a(f),a^{\ast}(g)](\alpha)=(f,g)_{L}\alpha,
\end{equation}
for all $\alpha \in \cup_{N=0}^{\infty}\bigoplus_{n=0}^{N} L^{\odot n}$ (see appendix \ref{ap1}). Observe that $\cup_{N=0}^{\infty}\bigoplus_{n=0}^{N} L^{\odot n}\subseteq \mathfrak{H}_{T}(L)$ is a dense inclusion. 

Of course, $a(h)$ is the so-called annihilation operator and $a^{\ast}(h)$ the creation operator. Using these operators and the Fock space $\mathfrak{H}_{T}(L)$ we can proceed to construct the representations of the CCR-algebras in Segal and Weyl form. We will come back to the creation and anniliation operators in section \ref{secondq}. 

Note that what we have described  is the vacuum Fock space since the unitary time translation $U(\Delta t \delta^\mu_0)$ operator acts as $e^{ip^0 \Delta t}$ and we know the Fourier transform of the test functions is supported on the $\mathcal{H}_m$ hyperboloid, so the generator $P^0$ has positive spectrum.

\begin{figure}[!h]
\begin{center}
\begin{adjustbox}{width=.8\textwidth}
\begin{tikzpicture}
    
    \filldraw[color=black, fill=white,  thick](0,0) rectangle (4,2.5);
	\filldraw[color=black, fill=white,  thick](5,0) rectangle (9,2.5);
	\filldraw[color=black, fill=white,  thick](10,0) rectangle (14,2.5);
    \filldraw[color=black, fill=white,  thick, rounded corners=5mm](0,-3.5)  rectangle (4,-1);
	\filldraw[color=black, fill=white,  thick, rounded corners=5mm](5,-3.5) rectangle (9,-1);
	\filldraw[color=black, fill=white,  thick, rounded corners=5mm](10,-3.5) rectangle (14,-1);
\filldraw[color=black, fill=white,  thick, rounded corners=5mm, densely dashed](0,-7)  rectangle (4,-4.5);
	\filldraw[color=black, fill=white,  thick, rounded corners=5mm, densely dashed](5,-7) rectangle (9,-4.5);
	\filldraw[color=black, fill=white,  thick, rounded corners=5mm, densely dashed](10,-7) rectangle (14,-4.5);
\path 
(2,2) node(fock) {\textbf{$H_m$-space}}
	      (2,1) node(f) {$L \simeq L^2(\mathbb{R}^3,\frac{d^3p}{\omega_p})$}
	      (7,2) node(segal) {\textbf{Segal}}
	      (7,1) node(s) {$H=\overline{\mathcal{S}(\mathbb{R}^4,\mathbb{R})/\mathcal{S}_0}$}
	      (7,.5)  node(s2) {with c.s. $\beta$}
	      (12,2) node(weyl) {\textbf{Weyl}}
	      (12,1) node(w) {$K=\overline{\left\{\right. \text{even functions}}$}
	      (12,.5) node(w2) {$\overline{\text{of } x^0   \left.\right\}}\subset H=K \oplus \beta K$};
     \path 
      (2,-1.5) node(q1)    	  {\textbf{$H_m$ CCR algebra}}
	(2,-2.5) node(a2){ $[a(f),a(g)^*]=i(f,g)_L$}
	(7,-1.5) node(q2) {\textbf{Segal CCR algebra}}
	(7,-2.5) node(c) {$W(h_1)W(h_2)=$}
	(7,-3) node(a2){ $W(h_1+h_2)e^{\frac{i}{2}(h_1,\beta h_2)_H}$}
	(12,-1.5) node(q2) {\textbf{Weyl CCR algebra}}
	(12,-2.5) node(c) {$V(g)U(f)=$}
	(12,-3) node(a2){ $U(f)V(g)e^{i(f, g)_K}$} ;
    \path 
      (2,-5) node(q1)    	  {\textbf{Fock vacuum rep.}}
	(2,-6) node(a2){ $a(f)$ and $a(f)^*$ }
	(2,-6.5) node(a2){on the Fock of $L$}
	(7,-5) node(q2) {\textbf{Segal vacuum rep.}}
	(7,-6) node(c) {$W_F(h)=e^{\frac{i}{\sqrt{2}} (a^{\ast}(h)+a(h))}$}
	(7,-6.5) node(a2){on the Fock of $L$}
	(12,-5) node(q2) {\textbf{Weyl vacuum rep.}}
	(12,-6) node(c) {$U_F(f)=W_F(f)=:e^{i\varphi(f)}$}
	(12,-6.5) node(a2){ $V_F(g)=W_F(\beta g)=:e^{i \pi(g)}$} 
	(9.5,-2) node(iso) {iso};
	\draw[thick] (2,0)  --  (2,-1) ;
	\draw[thick] (7,0) --  (7,-1);
	\draw[thick] (12,0) -- (12,-1);
	\draw[thick, <->] (9,-2.25) -- (10,-2.25);
	\draw[thick, densely dashed, -> ] (2,-3.5) -- (2,-4.5);
	\draw[thick, densely dashed, ->] (7,-3.5) -- (7,-4.5);
	\draw[thick, densely dashed, -> ] (12,-3.5) -- (12,-4.5);
\end{tikzpicture}
\end{adjustbox}
\end{center}
\caption{A diagram showing that from the different one-particle Hilbert  spaces ($L, H, K$) we can construct algebras showcasing different versions of the CCR. They will provide alternative formulations of the so-called second quantization. The positive mass hyperboloid CCR algebra generated by 1, $a(f)$ and $a(f)^*$ is a unital $*-$algebra. The Segal and Weyl CCR algebras are isomorphic $C^*-$algebras. The last row indicates an additional step, where one introduces a state (the vacuum state here) and by the GNS construction a representation on the Fock space of $L$ is obtained \cite{araki1963lattice}. }
\label{secondquantization}
\end{figure}
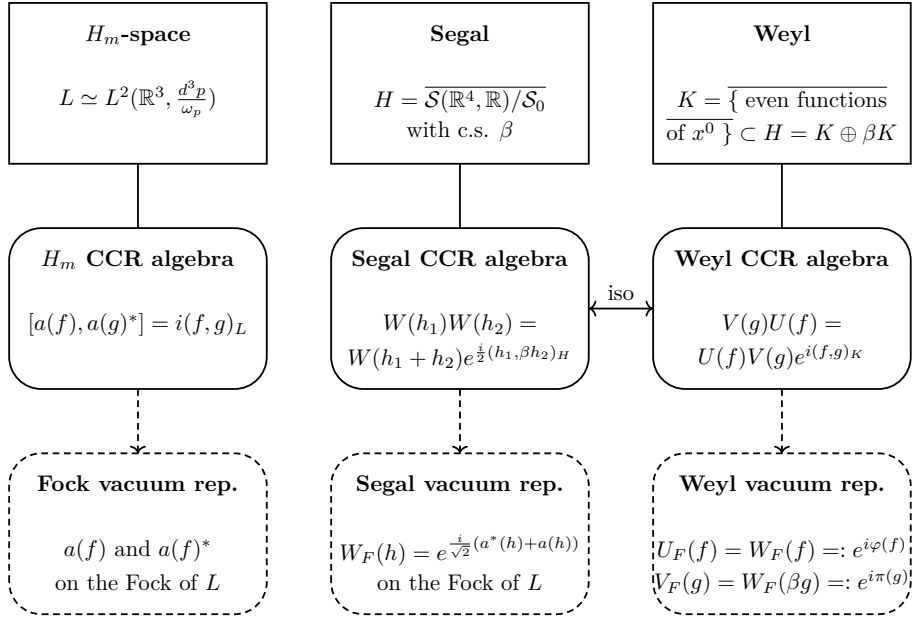

\subsection{Representation of the CCR-algebras}\label{weylrepresentation}

In short, Segal representation $W_{F}$ takes an element $h$ in the real Hilbert space $H$ and gives a unitary operator $W_{F}(h)$ acting on the complex Hilbert space $\mathfrak{H}_{T}(L)$, satisfying the relations \eqref{ccrsegal}. In order to construct such a representation we define the field $\chi(h)$ for each $h\in H$ as (the closure of) a combination of creation and annihilation operators introduced in the previous subsection\footnote{This is usually understood as the smeared scalar field operator $\int_{\mathbb{R}^4} d^4x h(x) \Phi(x)$.}  
\[ \chi(h)=\frac{1}{\sqrt{2}} \overline{(a^*(h)+a(h))},\] 
but this time as $a$ and $a^{\ast}$ have their  arguments in $H$ instead of $L$ (recall $H=L$ as real vector spaces), equation \eqref{conmut} must be replaced by   
\begin{equation}\label{conmutH}
	[a(h_{1}),a^{\ast}(h_{2})](\alpha)=\left((h_{1},h_{2})_{H}-i(h_{1},\beta h_{2})_{H}\right)\alpha
\end{equation} 
for all $\alpha \in \cup_{N=0}^{\infty}\bigoplus_{n=0}^{N} L^{\odot n}$, in agreement with \eqref{siete}. Then we define $W_{F}(h)=e^{i\chi(h)}$ and using this version of the commutation between creation and annihilation operators one obtains 
\begin{equation}\label{segalrep}
	W_{F}(h_{1})W_{F}(h_{2})=W_{F}(h_{1}+h_{2})e^{\frac{i}{2}(h_{1},\beta h_{2})_{H}}=W_{F}(h_{1}+h_{2})e^{-\frac{i}{2}\sigma(h_{1}, h_{2})}
\end{equation} 
straightforwardly. We call $\pi_{S}$ the representation map from the Segal's CCR- algebra given by:
\[ \pi_{S}: W(h)\longmapsto W_{F}(h).\] 
We also have $W_{F}(0)=1$ and $W_{F}(h)^*=W_{F}(-h)$. This representation gives a concrete unital $C^*-$algebra. An important fact is that even though the operators $\chi(h)$ are unbounded, they are self-adjoint and the operators $W_{F}(h)=e^{i\chi(h)}$ are unitary (see \cite{BR1}). 

There is another way to find a complex Hilbert space $\mathcal{H}$ on which we can represent the CCR-algebra in Segal form (or Weyl), this is the Gelfand–Naimark–Segal (GNS) construction. This procedure starts with a state\footnote{Given a unital $\ast$-algebra $R$, a state $\mu$ over $R$ is a $\mathbb{C}$-linear map $\mu: R\to \mathbb{C}$ which is positive (i.e.$\mu(a^{\ast}a)\geq 0 $ for all $a\in R$) and normalized (i.e. $\mu(1_{R})=1$).} $\mu$ over our $\ast$-algebra, and finally gives a representation $\pi_{\mu}$ for the $\ast$-algebra. If the state $\mu$ is given by the expression $\omega(W(h))=e^{-\frac{1}{4}(h,h)_{H}}$, it is possible to see that the complex Hilbert space so constructed $\mathcal{H}$ is unitarily equivalent to $\mathfrak{H}_{T}(L)$, so there exists a surjective isometric map $U:\mathcal{H}\to \mathfrak{H}_{T}(L)$ such that $U\pi_{\omega}(a)U^{-1}=\pi_{S}(a)$ for all $a$ in the CCR-algebra. This is why for our purposes it is enough to just construct $\mathfrak{H}_{T}(L)$. Note that for reasons already mentioned in the previous subsection $\mu$ is the vacuum state. For details on the GNS construction the reader can see \cite{araki1963representations}.  \\

The Weyl CCR-algebra has also a representation as bounded operators on $\mathfrak{H}_{T}(L)$. Obviously we are looking for operators $U_{F}(f)$, $V_{F}(g)$ for $f,g\in K$ satisfying the relations \eqref{weylrep}. First we define unbounded operators $\varphi_0(f)$  and $\pi_0(g)$ for $f,g\in K$ as linear combinations of the creation and annihilation operators\footnote{Again, these operators are usually formally understood as the scalar field operator and its conjugate momentum operator: $\int_{\mathbb{R}^4} d^4xf(x) \varphi(x) $ and $\int_{\mathbb{R}^4} d^4x g(x) \pi(x) $, respectively. In the following section we will see that they are equivalently  labelled by functions at fixed time. } on $\mathfrak{H}_{T}(L)$. These linear combinations are closable and their closures are self-adjoint operators\footnote{See for example \cite{reedsimonvol2}, Theorem X.41.\label{footnoteselfadjoint}} (i.e. $\varphi_{0}(f)$ and $\pi_{0}(g)$ are \textit{essentially selfadjoint}). Explicitly, working with $K\subseteq H=L$ (again as real vector spaces)
\begin{align*}
	\varphi(f)&=\frac{1}{\sqrt{2}}\overline{(a^{\ast}(f)+a(f))} \\
	\pi(g)&=\frac{i}{\sqrt{2}}\overline{(a^{\ast}(g)-a(g))}.
\end{align*}
with $f,g\in K$\footnote{It should be noted also that these creation and annihilation operators, which are fed by functions in $K$, are related to those in the CCR Fock representation which are fed by functions in $L$ (the 1-particle space). Note that the latter are  the (linear and antilinear) extensions of the former when $L$ is considered the complexification of $K$, $K+i K$.}. Again it is straightforward to see that the commutation relation for the creation and annihilation operators \eqref{conmut} implies 
\begin{equation}\label{golazo}
	\big[\varphi(h_{1}),\pi(h_{2})\big](\alpha)=i(h_{1},h_{2})_{K}\alpha \text{ con }h_{1},h_{2}\in K\subseteq H=L,
\end{equation}
for all $\alpha \in \cup_{N=0}^{\infty}\bigoplus_{n=0}^{N} L^{\odot n}$. Defining $U_{F}(f)=e^{i\varphi(f)}$ and $V_{F}(g)=e^{i\pi(g)}$ and using the commutation relations for $\varphi(f)$ and $\pi(g)$ it is direct to see that the relations of equations \eqref{weylrep} also holds for $U_{F}$ and $V_{F}$. As in Segal's case, the operators $U_{F}(f)$ and $V_{F}(g)$ are unitary.

In addition, before we had a  $\ast-$isomorphism $\kappa$, but at the level of representations we have an equality: the representations of the CCR-algebras in Segal and Weyl form are related by
\begin{equation}\label{weylsegal}
	W_{F}(h)=e^{\frac{i}{2}(f,g)_{K}}U_{F}(f)V_{F}(g),
\end{equation}
where $h=f+\beta g$. Here it is evident the role of  $\beta$ as an assistant to $W_{F}$ for distinguishing between the test functions of  $U_{F}$ and those of $V_{F}$ (basically, distinguishing between position and momentum inside $H$). 
If we write  \eqref{weylsegal} exhibiting the operators $\chi$, $\varphi$ and $\pi$, and the BCH formula is used, we get 
\begin{equation}\label{chi}
	\chi(f+\beta g)=\varphi(f)+\pi(g).
\end{equation}

\section{Explicit construction of 1-particle Hilbert spaces}

\subsection{Segal's real Hilbert space of spacetime test functions}

We are interested in spacetime real functions. The precise function space that we shall use at this point is  $\mathcal{S}(\mathbb{R}^{4},\mathbb{R})$ (modulo $\mathcal{S}_0$) together with a complex structure $\beta$, which we introduced in a somewhat abstract way in  the previous section. After implementing the construction of the real Hilbert space $H$ that we have already described, it remains to give a more hands-on presentation of $\beta$, which we do now.

\subsubsection{Definition of \texorpdfstring{$\beta$}{} operator}\label{beta}

We need to define an operator on $H$ such that satisfies \eqref{betadef}. These conditions mean that $\beta$ is a complex structure compatible with the inner product. Such inner product is determined by the value of the Fourier transform evaluated on the mass hyperboloid $H_{m}$.This motivates to define $\beta$ as multiplication by i on the space of square integrable functions on $H_{m}$. More precisely, if $h\in \mathcal{S}(\mathbb{R}^{4},\mathbb{R})$

\begin{equation}\label{betafourier}
	\mathscr{F}(\beta h)(p)=i\mathscr{F}(h)(p) \hspace{1cm} \text{with } p\in H_{m}.
\end{equation}
It should be noticed that the equality is only demanded on the hyperboloid.  

In order to find an explicit form of $\beta h$ for $h \in \mathcal{S}(\mathbb{R}^{4},\mathbb{R})$, Araki suggests  (on footnote 6 of \cite{araki1964neumann}) to introduce an arbitrary odd function $\eta(p^0)$ in $C^{\infty}(\mathbb{R})$ that is $1$ for $p^0\geq m$. Then for any $m>0$ we can take for instance,
\begin{align*}
 \eta(p^0)&=-\theta(-p^0)\theta(m+p^{0})[e^{1/p^0}/(e^{1/p^0}+e^{1/(-p^0-m)})]-\theta(-p^0-m) \\
   &+\theta(p^0)\theta(m-p^0)[e^{-1/p^0}/(e^{-1/p^0}+e^{1/(p^0-m)})]+\theta(p^0-m).
 \end{align*}
Such a function is useful because of the following. Let us consider a function $g$ defined by  
\begin{equation*}
    \mathscr{F}(g)(p)=i\mathscr{F}(h)(p)\eta(p^0), 
\end{equation*}
with $h\in \mathcal{S}(\mathbb{R}^{4},\mathbb{R})$. Then $g$  coincides with $\beta h$,  since when the previous expression is evaluated on $H_{m}$,  $\eta(p^0)=1$ and (\ref{betafourier}) is recovered. By the inverse Fourier transform we get,
\begin{equation}\label{defg}
	g(x)=i\mathscr{F}^{-1}(\mathscr{F}(h)(p)\eta(p^0))(x).
\end{equation}

It can be seen that  $g \in \mathcal{S}(\mathbb{R}^4,\mathbb{R})$. In order to do this, first note that  $g\in \mathcal{S}(\mathbb{R}^4,\mathbb{C})$, since$\mathscr{F}: \mathcal{S}(\mathbb{R}^4,\mathbb{R})\subseteq\mathcal{S}(\mathbb{R}^4,\mathbb{C})\to \mathcal{S}(\mathbb{R}^4,\mathbb{C})$ (idem for $\mathscr{F}^{-1}$) and the fact that a bounded function as is $\eta$ multplied by another function in $\mathcal{S}(\mathbb{R}^4,\mathbb{C})$ is still in  $\mathcal{S}(\mathbb{R}^4,\mathbb{C})$. So to see that indeed $g$ takes values in $\mathbb{R}$ we need to use that $\eta$ is odd: if we write \eqref{defg} explicitly,
\begin{align*}
	g(x)&=\frac{i}{(2\pi)^3}\int_{\mathbb{R}^4} h(x') \int_{\mathbb{R}^4} \eta(p^0)e^{ip(x'-x)}dp dx' \\
	  	&=-2\int_{0}^{\infty} \eta(p^0)\left[ \int_{-\infty}^{\infty} h(x^{0\prime},\vec{x})\sin(p^0(x^{0\prime}-x^{0}))dx^{0\prime} \right] dp^{0} 
\end{align*}
which is real if  $h\in \mathcal{S}(\mathbb{R}^4,\mathbb{R})$, so we have shown that $\beta|_{\mathcal{S}(\mathbb{R}^4,\mathbb{R})}:\mathcal{S}(\mathbb{R}^4,\mathbb{R})\to \mathcal{S}(\mathbb{R}^4,\mathbb{R})$. An explicit example of the action of $\beta$ on a specific function can be found in Appendix \ref{ejemplo}. To see how $\beta$ is extended to $H$ see how we extended the multiplication by $i$ in Section \ref{oneparticlespacesss}.

\subsubsection{It's a good definition}\label{good definition}

Although is may seem that the freedom in the choice of $\eta$ to compute $\beta h$ spoils the definition of $\beta$, it is not so. $\beta h$ should be thought as an element in $H$ not just in $\mathcal{S}(\mathbb{R}^4,\mathbb{R})$. Recall that $H$ is the completion of $\mathcal{S}(\mathbb{R}^4,\mathbb{R})/\mathcal{S}_0$ with real inner product \eqref{ddd}. Let us show that if given $h\in \mathcal{S}(\mathbb{R}^4,\mathbb{R})$ and two functions $h'$ and $h''$ in $\mathcal{S}(\mathbb{R}^4,\mathbb{R})$ satisfying  \eqref{betafourier}, then they are in the same  class in $H$. So we first assume
\begin{equation}\label{condicion}
	\mathscr{F}(h')(p)=i\mathscr{F}(h)(p) \hspace{0.5cm} \text{with } p\in H_{m}\hspace{1.6cm} 	\mathscr{F}(h'')(p)=i\mathscr{F}(h)(p) \hspace{0.5cm} \text{with } p\in H_{m}.
\end{equation}
Then the distance between them is
\begin{align*}
	||h'-h''||_{\mathcal{S}}^{2}&=(h'-h'',h'-h'')_{\mathcal{S}}=\int_{\mathbb{R}^{3}} |\mathscr{F}(h')(\omega_p,\vec{p})-\mathscr{F}(h'')(\omega_p,\vec{p})|^{2}\frac{d^{3}p}{\omega_p} \\
	  &=\int_{\mathbb{R}^{3}} |i\mathscr{F}(h)(\omega_p,\vec{p})-i\mathscr{F}(h)(\omega_p,\vec{p})|^{2}\frac{d^{3}p}{\omega_p}=0,
\end{align*}
where we have used \eqref{prodintmomento} for the inner product. This means that $\beta:H\to H$ is well defined and, given  $h\in \mathcal{S}(\mathbb{R}^4,\mathbb{R})$ and $h'$ y $h''$ satisfying \eqref{condicion}, then $h'-h''\in \mathcal{S}_{0}$.

\subsection{Weyl's real Hilbert space and initial conditions}

\subsubsection{The spaces of initial conditions  \texorpdfstring{$\mathfrak{F}_{\varphi}$}{} and \texorpdfstring{$\mathfrak{F}_{\pi}$}{}}

At this point we are ready to define the function spaces of initial conditions  $\mathfrak{F}_{\varphi}$ and $\mathfrak{F}_{\pi}$, which are essential both to the study of the solutions of the Klein-Gordon equation as well as to the first quantization in Weyl form. Let us start from  $\mathcal{S}(\mathbb{R}^3,\mathbb{R})$, since we are exploring the space of functions at fixed time. This is dense in $L^{2}(\mathbb{R}^3,\mathbb{R})$, with the standard inner product. We define the (three-dimensional) Fourier transform\footnote{We will use both the three-dimensional and four-dimensional Fourier transforms, and with abuse of notation are both noted by $\mathscr{F}$. We hope it is clear from the context which one we are using.}
\begin{equation*}
		\mathscr{F}(f)(p)=\frac{1}{(2\pi)^{\frac{3}{2}}}\int_{\mathbb{R}^{3}} f(\vec{x})e^{-i\vec{p}.\vec{x}}d^{3}x,
\end{equation*}
and the operators $\omega^{\alpha}$ as multiplication by $\omega^{\alpha}_{p}=(\sqrt{|\vec{p}|^{2}+m^2})^{\alpha}$ (in momentum space),
\begin{align*}
	\mathscr{F}(\omega^{\alpha} f)(\vec{p})&=\omega^{\alpha}_{p}\mathscr{F}(f)(\vec{p}), 	
\end{align*}
with $\alpha\in \mathbb{Q}$.

It is actually possible and desirable to consider different inner products on $\mathcal{S}(\mathbb{R}^3,\mathbb{R})$, as will become clear when we make the connection with $K$ and $\beta K$. So let us define,
\begin{align*}
	(f_{1},f_{2})_{\varphi}&:=(\omega^{-\frac{1}{2}} f_{1},\omega^{-\frac{1}{2}}f_{2})_{L^{2}} \\
	(g_{1},g_{2})_{\pi}&:=(\omega^{\frac{1}{2}} g_{1},\omega^{\frac{1}{2}} g_{2})_{L^{2}}, 
\end{align*}
The completions of $\mathcal{S}(\mathbb{R}^3,\mathbb{R})$ under these inner products are denoted by  $\mathfrak{F}_{\varphi}$ and $\mathfrak{F}_{\pi}$ respectively. Roughly, it can be stated that $\mathfrak{F}_{\varphi}$ is the space of functions whose Fourier transform divided by $\omega^{\frac{1}{2}}$ are in  $L^{2}$ with the standard Lebesgue measure, while  $\mathfrak{F}_{\pi}$ is the space of functions whose Fourier transform multiplied by $\omega^{\frac{1}{2}}$ are in $L^{2}$ again with the stander Lebesgue measure; namely their decay needs to be faster than what is necessary for them to be in $L^{2}$. We can be more precise: 
\[ \mathfrak{F}_{\pi} \subset L^2(\mathbb{R}^3,\mathbb{R}) \subset \mathfrak{F}_{\varphi}\]
with the inclusions\footnote{These maps are continuous, since $||f||_{\pi}^{2}=(\omega^{\frac{1}{2}} f,\omega^{\frac{1}{2}}f)_{L^{2}}=(\omega^{\frac{1}{2}}\mathscr{F}(f), \omega^{\frac{1}{2}}\mathscr{F}(f) )_{L^{2}} 
\geq m ||\mathscr{F}(f)||^{2}_{L^{2}}=m||f||^{2}_{L^{2}}$. Thus, $||j_{1}||\leq m^{-\frac{1}{2}}$. In a similar fashion it can be seen that $||j_{2}||\leq m^{-\frac{1}{2}}$.}
\begin{equation*}
    j_1:\mathfrak{F}_\pi \rightarrow L^2(\mathbb{R}^3,\mathbb{R}),\quad j_2:L^2(\mathbb{R}^3,\mathbb{R}) \rightarrow \mathfrak{F}_\varphi
    \label{jinclusions}
\end{equation*} 

The map $\omega^{\frac{1}{2}}:\mathcal{S}(\mathbb{R}^3,\mathbb{R})\to \mathcal{S}(\mathbb{R}^3,\mathbb{R})$ can be extended to $\omega^{\frac{1}{2}}:\mathfrak{F}_{\pi}\to L^{2}$ and to $\omega^{\frac{1}{2}}:L^{2}\to \mathfrak{F}_{\varphi}$, both extensions being onto\footnote{Let us show first that $\omega^{1/2}:L^2\rightarrow \mathfrak{F}_\varphi$ is onto. For any $f\in\mathfrak{F}_\varphi$, $f=\omega^{1/2}(\omega^{-1/2}f)$, so we need to show that $\omega^{-1/2}f \in L^2$. This is straightforward: $||\omega^{-1/2}f||_{L^2}^2=(\omega^{-1/2}f,\omega^{-1/2}f)_{L^2}=(f,f)_{\varphi}=||f||_\varphi^2<\infty$. In an identical way one can show that  $\omega^{1/2}:\mathfrak{F}_\pi\rightarrow L^2 $ is onto.}. 

How do we relate the spaces $K$ and $\beta K$ of even and odd functions in $x^0$ with the spaces of initial conditions just defined? The connection between these spaces comes from isometric isomorphisms relating $K$ with $\mathfrak{F}_{\varphi}$, and $\beta K$ with $\mathfrak{F}_{\pi}$. We denote these maps as $\delta_0$ and $\delta_1$, 
\[\delta_0:K\rightarrow \mathfrak{F}_\varphi,\quad \delta_1:\beta K \rightarrow \mathfrak{F}_\pi.\]
They are defined by (for $f\in K\cap \mathcal{S}(\mathbb{R}^{4},\mathbb{R})$ and $g\in \beta K\cap \mathcal{S}(\mathbb{R}^{4},\mathbb{R})$),
\begin{equation*}
	\mathscr{F}(\delta_{0}f)(\vec{p})=\mathscr{F}(f)(\omega_{p},\vec{p}) \hspace{1cm} 	\mathscr{F}(\delta_{1}g)(\vec{p})=(i\omega_{p})^{-1}\mathscr{F}(g)(\omega_{p},\vec{p}).
	\label{deltas}
\end{equation*}
These maps extend to $K$ and $\beta K$ by continuity (we will see they preserve the norm for elements in $\mathcal{S}(\mathbb{R}^{4},\mathbb{R})$). A natural question is how the operator $\beta$ induces maps between $\mathfrak{F}_\varphi$ and $\mathfrak{F}_\pi$. We can resort to the following commuting diagrams that define  $\bar{\beta}_{\varphi,\pi}$ and $\bar{\beta}_{\pi,\varphi}$.

\begin{center}
	\begin{minipage}{5cm}
		\[\scalemath{1.2}{
		\xymatrix{
			K \ar[r]^{\beta} \ar[d]_{\delta_{0}}  &   \beta K \ar[d]^{\delta_{1}} \\
			\mathfrak{F}_{\varphi} \ar@{-->}[r]_{\bar{\beta}_{\varphi,\pi}} & \mathfrak{F}_{\pi} 
		}}\]
	\end{minipage}
	\begin{minipage}{5cm}
			\[\scalemath{1.2}{
			\xymatrix{
				\beta K \ar[r]^{\beta} \ar[d]_{\delta_{1}}  &   K \ar[d]^{\delta_{0}} \\
				\mathfrak{F}_{\pi} \ar@{-->}[r]_{\bar{\beta}_{\pi,\varphi}} & \mathfrak{F}_{\varphi} 
			}}\]
	\end{minipage}
\end{center}
Namely,  $\bar{\beta}_{\varphi,\pi}=\delta_{1}\circ \beta \circ \delta_{0}^{-1}$ and $\bar{\beta}_{\pi,\varphi}=\delta_{0}\circ \beta \circ \delta_{1}^{-1}$.	

In order to see that $\delta_{0}$ and $\delta_{1}$ are isometries, let us start by noticing a simple form that adopts the norm  $||h||^{2}_{H}$ for $h\in H\cap\mathcal{S}(\mathbb{R}^{4},\mathbb{R})$. From the definition $\Delta_1=2 \Im{\Delta^{(+)}}$ and (\ref{Hinner}) it is straightforward to see that 
\begin{equation*}
    ||h||_H^2=\int_{\mathbb{R}^{3}} |\mathscr{F}(h)(\omega_{p},\vec{p})|^{2}\frac{d^{3}p}{\omega_p}.
    \label{Hnorm}
\end{equation*}
On the other hand, we can compute $||\delta_{0}f||^{2}_{\varphi}$ for $f\in K\cap \mathcal{S}(\mathbb{R}^{4},\mathbb{R})$ and  $||\delta_{1}g||^{2}_{\pi}$ for $g\in \beta K\cap \mathcal{S}(\mathbb{R}^{4},\mathbb{R})$, and check that these coincide with the norms of $f$ and $g$ in $H$ respectively. Indeed, 
\begin{align*}
	&||\delta_{0}f||^{2}_{\varphi}=(\delta_{0}f,\delta_{0}f)_{\varphi}=(\omega^{-\frac{1}{2}}\delta_{0}f,\omega^{-\frac{1}{2}}\delta_{0}f)_{L^{2}}=(\mathscr{F}(\omega^{-\frac{1}{2}}\delta_{0}f),\mathscr{F}(\omega^{-\frac{1}{2}}\delta_{0}f))_{L^{2}} \\
	&=(\omega^{-\frac{1}{2}}_{p}\mathscr{F}(\delta_{0}f)(\vec{p}),\omega^{-\frac{1}{2}}_{p}\mathscr{F}(\delta_{0}f)(\vec{p}))_{L^{2}}=(\omega^{-\frac{1}{2}}_{p}\mathscr{F}(f)(\omega_{p},\vec{p}),\omega^{-\frac{1}{2}}(\vec{p})\mathscr{F}(f)(\omega_{p},\vec{p}))_{L^{2}}   \\
	&=\int_{\mathbb{R}^{3}} |\mathscr{F}(f)(\omega_{p},\vec{p})|^{2}\frac{d^{3}p}{\omega_{p}}=||f||^{2}_{K}.
\end{align*}
Similarly for $\delta_{1}$, 
\begin{align*}
    &||\delta_{1}g||^{2}_{\pi}=(\delta_{1}g,\delta_{1}g)_{\pi}=(\omega^{\frac{1}{2}}\delta_{1}g,\omega^{\frac{1}{2}}\delta_{1}g)_{L^{2}}=(\mathscr{F}(\omega^{\frac{1}{2}}\delta_{1}g),\mathscr{F}(\omega^{\frac{1}{2}}\delta_{1}g))_{L^{2}} \\
 	&=(\omega^{\frac{1}{2}}_{p}\mathscr{F}(\delta_{1}g)(\vec{p}),\omega^{\frac{1}{2}}_{p}\mathscr{F}(\delta_{1}g)(\vec{p}))_{L^{2}}=(\omega^{\frac{1}{2}}_{p}\frac{1}{i\omega_{p}}\mathscr{F}(g)(\omega_{p},\vec{p}),\omega^{\frac{1}{2}}_{p}\frac{1}{i\omega_{p}}\mathscr{F}(g)(\omega_{p},\vec{p}))_{L^{2}}   \\
 	&=(\omega^{-\frac{1}{2}}_{p}\mathscr{F}(g)(\omega_{p},\vec{p}),\omega^{-\frac{1}{2}}(\vec{p})\mathscr{F}(g)(\omega_{p},\vec{p}))_{L^{2}} =\int_{\mathbb{R}^{3}} |\mathscr{F}(g)(\omega_{p},\vec{p})|^{2}\frac{d^{3}p}{\omega_{p}}=||g||^{2}_{\beta K}.
\end{align*}
Then, $\delta_{0}|_{K\cap \mathcal{S}(\mathbb{R}^{4},\mathbb{R})}:K\cap \mathcal{S}(\mathbb{R}^{4},\mathbb{R})\to \mathfrak{F}_{\varphi}\cap \mathcal{S}(\mathbb{R}^{3})$ and $\delta_{1}|_{\beta K\cap \mathcal{S}(\mathbb{R}^{4},\mathbb{R})}:\beta K\cap \mathcal{S}(\mathbb{R}^{4},\mathbb{R})\to \mathfrak{F}_{\pi}\cap \mathcal{S}(\mathbb{R}^{3})$ are isometries and because of this they are injective maps. 

It can be also shown that $\delta_{0}|_{K\cap C_{0}^{\infty}(\mathbb{R}^{4})}:K\cap C_{0}^{\infty}(\mathbb{R}^{4})\to \mathfrak{F}_{\varphi}\cap C_{0}^{\infty}(\mathbb{R}^{3})$ and $\delta_{1}|_{\beta K\cap C_{0}^{\infty}(\mathbb{R}^{4})}:\beta K\cap C_{0}^{\infty}(\mathbb{R}^{4})\to \mathfrak{F}_{\pi}\cap C_{0}^{\infty}(\mathbb{R}^{3})$ are surjective maps. We sketch the proof here and leave the technical details to Appendix \ref{ap3}. Let us consider two functions $f\in \mathfrak{F}_{\pi}\cap C^{\infty}_{0}(\mathbb{R}^{3})$ and $g\in \mathfrak{F}_{\varphi}\cap C^{\infty}_{0}(\mathbb{R}^{3})$, then the unique $C^{\infty}(\mathbb{R}^4)$ solution to the Klein-Gordon initial value problem (see equation \eqref{suave} and corollary 1.1 in \cite{dimock1980algebras}) is given by:
\begin{equation*}
 	F_{E(h)}(x)=\int_{\mathbb{R}^{4}} \Delta(x-y)h(y)dy=\frac{-i}{(2\pi)^{\frac{3}{2}}}\int_{\mathbb{R}^{4}} e^{-ipx}\delta(p^2-m^2)\sgn(p^{0})\mathscr{F}(h)(p) d^{4}p,
\end{equation*} 
where $\Delta=2\Re{\Delta^{(+)}}$ is the propagator and $h\in C_{0}^{\infty}(\mathbb{R}^{4})$ (as we will show and use in Section \ref{contra}). As we show in appendix \ref{ap3} one can easily arrive to the identities 
\begin{equation}\label{initial conditions}
 F_{E(h)}(0,\vec{x})=(\delta_{1}h_{-})(\vec{x}) \hspace{1cm} -\frac{\partial F_{E(h)}}{\partial x^{0}}(0,\vec{x})=(\delta_{0}h_{+})(\vec{x}), 
\end{equation}
where  $h_{\pm}$ are the even and odd parts of $h$ in $x^0$, as previously defined.  But on the other hand (see appendix \ref{ap3} again) 
\[F_{E(h)}(0,\vec{x})=f(\vec{x}),\quad 
		\frac{\partial F_{E(h)}}{\partial x^{0}}(0,\vec{x})=g(\vec{x}),\]
then $\delta_{0}(-h_{+})=g$ and $\delta_{1}(h_{-})=f$. 

Moreover, let us analyze $\delta_{0}:K\to \mathfrak{F}_{\varphi}$ deeper, if we take an arbitrary $g \in \mathfrak{F}_{\varphi}$ there is a sequence $\{g_{n}\}_{n\in \mathbb{N}}$ in $\mathfrak{F}_{\varphi}\cap C^{\infty}_{0}(\mathbb{R}^{3})$ such that $g_{n}\to g$ in the topology of $\mathfrak{F}_{\varphi}$, but as we already saw for each $g_{n}$ there is a $h_{n}\in K\cap C_{0}^{\infty}(\mathbb{R}^{4})$ such that $\delta_{0}(h_{n})=g_{n}$ so this $\delta_{0}(h_{n})$ is a Cauchy sequence and $h_{n}$ too, so it has a limit $h\in H$. Then $\delta_{0}|_{K\cap C_{0}^{\infty}(\mathbb{R}^{4})}:K\cap C_{0}^{\infty}(\mathbb{R}^{4})\to \mathfrak{F}_{\varphi}\cap C_{0}^{\infty}(\mathbb{R}^{3})$ can be extended by continuity to a map $\delta_{0}:K\to \mathfrak{F}_{\varphi}$ such that $\delta_{0}(h)=g$. The same argument can be translate to $\delta_{1}:\beta K\to \mathfrak{F}_{\pi}$ obtaining its surjectivity. Then both maps $\delta_{0}:K\to \mathfrak{F}_{\varphi}$ and $\delta_{1}:\beta K\to \mathfrak{F}_{\pi}$ are bijective isometries.

\subsubsection{Maps between \texorpdfstring{$\mathfrak{F}_\varphi$}{} and \texorpdfstring{$\mathfrak{F}_\pi$}{}}

We would like now to introduce how to translate the map $\beta$ to the spaces of the initial conditions $\mathfrak{F}_{\varphi}$ and $\mathfrak{F}_{\pi}$. Based on the previous subsection, we identify $K$ and $\beta K$ with $\mathfrak{F}_\varphi$ and $\mathfrak{F}_\pi$ respectively. In particular we can think of the one-particle space as $\mathfrak{F}_{\varphi}\oplus \mathfrak{F}_{\pi}$, and $\beta$ should be labeled  $\bar{\beta}_{\varphi,\pi}$ or $\bar{\beta}_{\pi,\varphi}$ depending on where lives the element to which it is applied. More precisely, we can write the elements of $\mathfrak{F}_{\varphi}\oplus \mathfrak{F}_{\pi}$ as column vectors and consider $\bar{\beta}:\mathfrak{F}_{\varphi}\oplus \mathfrak{F}_{\pi}\to \mathfrak{F}_{\varphi}\oplus \mathfrak{F}_{\pi}$ given by the matrix 
\begin{equation*}
    \bar{\beta}=\begin{pmatrix}
        0 & \bar{\beta}_{\pi,\varphi} \\
        \bar{\beta}_{\varphi,\pi} & 0
    \end{pmatrix}.
\end{equation*}
Then a straightforward computation shows that $\bar{\beta}^{2}=-\begin{pmatrix}
    1_{\mathfrak{F}_{\varphi}} & 0 \\
    0 & 1_{\mathfrak{F}_{\pi}}
\end{pmatrix}.$ Indeed,
\begin{align*}
    \bar{\beta}_{\pi,\varphi}\circ \bar{\beta}_{\varphi,\pi}&=(\delta_{0}\circ \beta \circ \delta_{1}^{-1})\circ(\delta_{1}\circ \beta \circ \delta_{0}^{-1}) \\
    &=\delta_{0}\circ \beta \circ 1_{\beta K}\circ \beta \circ \delta_{0}^{-1}=\delta_{0}\circ \beta^{2} \circ \delta_{0}^{-1} \\
    &=\delta_{0}\circ (-1_{K}) \circ \delta_{0}^{-1}=-1_{\mathfrak{F}_{\varphi}}.
\end{align*}
Similarly one can see that $\bar{\beta}_{\varphi,\pi}\circ \bar{\beta}_{\pi,\varphi}=-1_{\mathfrak{F}_{\pi}}$.

In order to show that $\bar{\beta}^{\ast}=-\bar{\beta}$ it is sufficient to see that $\left(\bar{\beta}_{\varphi,\pi}\right)^{\ast}=-\bar{\beta}_{\pi,\varphi}$ and  $\left(\bar{\beta}_{\pi,\varphi}\right)^{\ast}=-\bar{\beta}_{\varphi,\pi}$. Let us show one of these relations by considering again commuting diagrams 
\begin{center}
	\begin{minipage}{5cm}
		\[\scalemath{1.2}{
		\xymatrix{
			K \ar[r]^{\beta} \ar[d]_{\delta_{0}}  &   \beta K \ar[d]^{\delta_{1}} \\
			\mathfrak{F}_{\varphi} \ar@{-->}[r]_{\bar{\beta}_{\varphi,\pi}} & \mathfrak{F}_{\pi} 
		}}\]
	\end{minipage}
	\begin{minipage}{5cm}
		\[\scalemath{1.2}{
		\xymatrix{
			\beta K \ar[r]^{\beta} \ar[d]_{\delta_{1}}  &   K \ar[d]^{\delta_{0}} \\
			\mathfrak{F}_{\pi} \ar@{-->}[r]_{\bar{\beta}_{\pi,\varphi}} & \mathfrak{F}_{\varphi} 
		}}\]
	\end{minipage}
\end{center}
and  for $g \in \mathfrak{F}_{\varphi}$ and $f\in \mathfrak{F}_{\pi}$ we have
\begin{align*}
	(\bar{\beta}_{\varphi,\pi}g,f)_{\pi}&=(\delta_{1}^{-1}\bar{\beta}_{\varphi,\pi}g,\delta_{1}^{-1}f)_{H}=(\beta \delta_{0}^{-1}g,\delta_{1}^{-1}f)_{H} \\
	&=(\delta_{0}^{-1}g,\beta^{\ast}\delta_{1}^{-1}f)_{H}=(\delta_{0}^{-1}g,-\beta \delta_{1}^{-1}f)_{H} \\
	&=(\delta_{0}^{-1}g,-\delta_{0}^{-1}\bar{\beta}_{\pi,\varphi}f)_{H}=(g,-\bar{\beta}_{\pi,\varphi}f)_{\varphi}. 
\end{align*}
Notice that we have used the fact that both $\delta$'s are isometries and that  $\beta^{\ast}=-\beta$. Similarly $\left(\bar{\beta}_{\pi,\varphi}\right)^{\ast}=-\bar{\beta}_{\varphi,\pi}$, proving that \eqref{betadef} holds for $\bar{\beta}$. Then,  in matrix notation we have
\begin{align*}
    \left(\bar{\beta}\begin{pmatrix}
        g_{1} \\ f_{1}
    \end{pmatrix},\begin{pmatrix}
        g_{2} \\ f_{2}
    \end{pmatrix} \right)_{\mathfrak{F}_{\varphi}\oplus \mathfrak{F}_{\pi}}&= \left(\begin{pmatrix}
        0 & \bar{\beta}_{\pi,\varphi} \\
        \bar{\beta}_{\varphi,\pi} & 0
    \end{pmatrix}\begin{pmatrix}
        g_{1} \\ f_{1}
    \end{pmatrix},\begin{pmatrix}
        g_{2} \\ f_{2}
    \end{pmatrix} \right)_{\mathfrak{F}_{\varphi}\oplus \mathfrak{F}_{\pi}} =\left(
    \begin{pmatrix}
        \bar{\beta}_{\pi,\varphi}f_{1} \\ \bar{\beta}_{\varphi,\pi} g_{1}
    \end{pmatrix},
    \begin{pmatrix}
        g_{2} \\ f_{2}
    \end{pmatrix} \right)_{\mathfrak{F}_{\varphi}\oplus \mathfrak{F}_{\pi}} \\
    &=\left(\bar{\beta}_{\pi,\varphi}  f_{1}, g_{2} \right)_{\varphi}+\left(\bar{\beta}_{\varphi,\pi} g_{1}, f_{2}\right)_{\pi}=\left( f_{1},-\bar{\beta}_{\varphi,\pi}  g_{2} \right)_{\pi}+\left( g_{1}, -\bar{\beta}_{\pi,\varphi} f_{2}\right)_{\varphi} \\
    &=\left(\begin{pmatrix}
        g_{1} \\ f_{1}
    \end{pmatrix},\begin{pmatrix}
        0 & -\bar{\beta}_{\pi,\varphi} \\
        -\bar{\beta}_{\varphi,\pi} & 0
    \end{pmatrix}\begin{pmatrix}
        g_{2} \\ f_{2}
    \end{pmatrix} \right)_{\mathfrak{F}_{\varphi}\oplus \mathfrak{F}_{\pi}}=\left(\begin{pmatrix}
        g_{1} \\ f_{1}
    \end{pmatrix},-\bar{\beta}\begin{pmatrix}
        g_{2} \\ f_{2}
    \end{pmatrix} \right)_{\mathfrak{F}_{\varphi}\oplus \mathfrak{F}_{\pi}}
\end{align*} 

Later, we will need to see how $\bar{\beta}$ acts on an element in $\mathfrak{F}_{\varphi}$ or $\mathfrak{F}_{\pi}$. Let us once again consider the commuting diagram 
	\[\scalemath{1.2}{
\xymatrix{
	K \ar[r]^{\beta} \ar[d]_{\delta_{0}}  &   \beta K \ar[d]^{\delta_{1}} \\
	\mathfrak{F}_{\varphi} \ar@{-->}[r]_{\bar{\beta}_{\varphi,\pi}} & \mathfrak{F}_{\pi} 
}}\]
Let $f\in K\cap \mathcal{S}(\mathbb{R}^{4},\mathbb{R})$ and recall that $\beta f$ in momentum space is given by $\mathscr{F}(\beta f)(p)=i\mathscr{F}(f)(p)\eta(p^{0})$ (see the comments after \eqref{betafourier}). By using also $\delta_{0}$ and $\delta_{1}$ in momentum space \eqref{deltas}, we obtain an expression for $\bar{\beta}_{\varphi,\pi}$ in momentum space, mapping $\mathscr{F}(\delta_{0}f)(\vec{p})$ to $\mathscr{F}\left(\delta_{1}\beta f\right)(\vec{p})$. Namely,
\begin{align*}
	\bar{\beta}_{\varphi,\pi}\left(\mathscr{F}(\delta_{0}f)(\vec{p})\right)&=\mathscr{F}(\delta_{1}\beta f)(\vec{p})=(i\omega_p)^{-1}\mathscr{F}(\beta f)(\omega_p,\vec{p}) \\
\Longrightarrow	\bar{\beta}_{\varphi,\pi}\left(\mathscr{F}(f)(\omega_p,\vec{p})\right)&=\omega^{-1}_{p}\mathscr{F}(f)(\omega_{p},\vec{p}), 
\end{align*}
This means that $\bar{\beta}_{\varphi,\pi}(h)=\omega_p^{-1}h$ with $h\in \mathfrak{F}_\varphi\cap \mathcal{S}(\mathbb{R}^{3})$. Also, applying  $\bar{\beta}_{\pi,\varphi}$ to this equality one obtains $-h=\bar{\beta}_{\pi,\varphi}\left(\omega_p^{-1}h\right)$, or $\bar{\beta}_{\pi,\varphi}(f)=-\omega_p f$ with $f\in \mathfrak{F}_{\pi}\cap \mathcal{S}(\mathbb{R}^{3})$. Then we extend this by continuity to  $\mathfrak{F}_\varphi$ and $\mathfrak{F}_\pi$.

\section{First quantization and duality}

As we mentioned, AQFT is based on a net of algebras over spacetime regions. This means that to each open region of spacetime an algebra of observables is assigned. The construction of such net is sometimes presented as a two-steps process, the so-called first quantization and second quantization. In the first quantization, to each region of spacetime a closed real vector subspace is assigned (contained in some 1-particle space). We will denote this map by $\mathsf{S}$, such that $\mathsf{S}(\mathcal{O})\subset H$, with $\mathcal{O}$ an open region of spacetime. The second quantization is the construction of a map that assigns to each subspace $\mathsf{S}(\mathcal{O})$ a subalgebra $R(\mathcal{O})$ of the representation  of a CCR-algebra as bounded operators as in Section \eqref{weylrepresentation}, and depends on the formulation we are dealing with. This will be made manifest in Section \ref{secondquantizationsection}.

Despite these comments, in the Weyl formulation   the map $\mathsf{S}$ does not assign closed real subspaces to open regions of \textit{spacetime}, but assigns two closed real subspaces to $B$, a \textit{spacelike} region. The relation between both approaches will be stated once we define the first quantization maps for Segal and Weyl, which we do next.

We shall now define the first quantization map $\mathsf{S}$ for the Weyl formulation. Let $B\subseteq \mathbb{R}^3$ be a measurable region\footnote{At this point on page 243 of \cite{horuzhy2012introduction} it is demanded that the boundary of $B$ is piecewise smooth. We do not see the need for such requirement and indeed it is not required in \cite{araki1964neumann}. \label{footnote:boundary}} of a fixed-time slice of $\mathbb{R}^{4}$. Then we define
\begin{equation}\label{mapauno}
\mathsf{S}_{R}(B):=\overline{j_{2}L^{2}(B)}^{||-||_{\varphi}},\qquad \mathsf{S}_{I}(B):=\bar{\beta}j_{1}^{-1}\left(L^{2}(B)\right),
\end{equation}
where $R$ and $I$ denote ``real'' and ``imaginary'' parts respectively. Here $j_{1}:\mathfrak{F}_{\pi} \to L^{2}$ y  $j_{2}:L^{2} \to \mathfrak{F}_{\varphi}$ are the continuous inclusions of \eqref{jinclusions} and $L^{2}(B)\subseteq L^{2}(\mathbb{R}^{3})$ is the closed subspace of $L^2$ functions that vanish outside $B$. Then, as  $j_{1}$ is  continuous  $j_{1}^{-1}\left(L^{2}(B)\right)$ is closed in  $\mathfrak{F}_{\pi}$, and since  $\bar{\beta}$ is unitary  $\mathsf{S}_{I}(B)$ is also closed in $\mathfrak{F}_{\varphi}$. Note that both $\mathsf{S}_{R}$ and $\mathsf{S}_{I}$ belong to $\mathfrak{F}_{\varphi}$ and then we have a map $B \mapsto \mathsf{S}_{R}(B)\oplus_{\mathbb{R}}\mathsf{S}_{I}(B) \subset \mathfrak{F}_{\varphi} \oplus_{\mathbb{R}} \mathfrak{F}_{\varphi}\simeq K \oplus_{\mathbb{R}} K$.

On the other hand, the first quantization map in the Segal presentation is defined by 
\begin{equation}
    \mathsf{S}_{S}(\mathcal{O}):=\overline{\{ [f]\in H|\quad f\in C_{0}^{\infty}(\mathbb{R}^{4},\mathbb{R}),\quad  \text{supp}(f)\subseteq \mathcal{O} \}}^{\|-\|_{H}}\subseteq H,
    \label{mapados}
\end{equation}
for each open region $\mathcal{O}\subseteq \mathbb{R}^{4}.$

In what follows 
\begin{equation}\label{prima}
	\mathcal{O}'=\left\{x\in \mathbb{R}^{4}|\quad (x-y)^2<0\text{ for }y\in \mathcal{O}\right\},
\end{equation}
denotes the causal complement of the region $\mathcal{O}$ and $B^c$ denotes the set complement of  the closure of $B$ in the fixed-time slice. Let us define the causal envelope of $B$ as
\begin{equation}
    C(B)=\{ x\in \mathbb{R}^{4} | (x-y)^2<0  \hspace{0.2cm} \forall \hspace{0.15cm} y\in B^{c} \}.
    \label{Cdef}
\end{equation}
We immediately  have that $C(B)'=C(B^c)$. Now we can relate both first quantization maps \eqref{mapauno} and \eqref{mapados}, which follows from the biyective isomorphisms $\delta_0$ and $\delta_1$ described in detailed in the previous section:
\begin{equation}
    \mathsf{S}_{S}(C(B))\simeq\mathsf{S}_{R}(B)\oplus_{\mathbb{R}}\bar{\beta} \mathsf{S}_{I}(B)
    \label{firstmapsrelation}
\end{equation}

We are concerned in this section with the proof that the first quantization maps in Weyl and Segal form satisfy a duality property. This property, in the Weyl formulation, is the starting point of the Haag duality proof of \cite{eckmann1973application}. We start with the description of this duality in the Weyl formulation and its proof, and then do the same for the Segal formulation, showing that it holds in one formulation if and only if it holds in the other thanks to \eqref{firstmapsrelation}.

\subsection{Duality in  Weyl form}\label{first}

Some important facts can be stated for the maps \eqref{mapauno}, for example \textit{isotony} says that if $B_{1}\subseteq B_{2}$, then $\mathsf{S}_{R}(B_{1})\subseteq \mathsf{S}_{R}(B_{2})$ and $\mathsf{S}_{I}(B_{1})\subseteq \mathsf{S}_{I}(B_{2})$. Another important property is the so-called \textit{additivity} $\overline{\bigoplus\mathsf{S}_{R}(B_{\lambda})}^{\|-\|_{\varphi}}=\mathsf{S}_{R}(\bigcup_{\lambda\in \Lambda}B_{\lambda})$ (see lemma 2 in \cite{araki1964neumann}). But here we are mainly interested in the \textit{duality property} which reads,
\begin{equation}
    \boxed{\mathsf{S}_{I}(B)^{\perp_{\varphi}}=\mathsf{S}_{R}(B^{c}),\quad \mathsf{S}_{R}(B)^{\perp_{\varphi}}=\mathsf{S}_{I}(B^{c})}
\label{dualidadweyl}
\end{equation}

\begin{theorem}\label{haagteoremauno}
	The assignment of subspaces to   regions $B\subseteq\mathbb{R}^{3}$ at fixed time given by \eqref{mapauno} satisfies Haag duality in the Weyl presentation, namely  \[\mathsf{S}_{I}(B)^{\perp_{\varphi}}=\mathsf{S}_{R}(B^{c}),\quad \mathsf{S}_{R}(B)^{\perp_{\varphi}}=\mathsf{S}_{I}(B^{c})\]
\end{theorem} 
\begin{proof}
 Let us first notice that  $L^{2}(\mathbb{R}^{3})=L^{2}(B)\oplus L^{2}(B^{c})$, since $L^{2}(B)=L^{2}(B^c)^{\perp_{L^2}}$. From now on $\perp$ means $\perp_\varphi$.  Then,
\begin{align*}
	j_{1}^{-1}\left(L^{2}(B)\right)&=\left\{ g\in \mathfrak{F}_{\pi} |\quad  (f,g)_{L^{2}}=0 \hspace{0.2cm}\forall f\in L^{2}(B^{c}) \right\} \\
	 &=\left\{ g\in \mathfrak{F}_{\pi} |\quad (j_2 f,\omega g)_{\varphi}=0 \hspace{0.2cm}\forall f\in L^{2}(B^{c}) \right\} \\
	 &=\left\{ g\in \mathfrak{F}_{\pi} |\quad (j_2 f,\bar{\beta}_{\pi,\varphi}g)_{\varphi}=0 \hspace{0.2cm}\forall f\in L^{2}(B^{c}) \right\}. 
\end{align*}
This means that if $g\in j_{1}^{-1}\left(L^{2}(B)\right)$ then $\bar{\beta}_{\pi,\varphi}g\in\left[\overline{j_{2}L^{2}(B^{c})}^{||-||_{\varphi}}\right]^\perp$, where we have completed $j_2L^{2}(B^{c})$ using the continuity of the inner product. In other words, we have $\bar{\beta}_{\pi,\varphi}j_{1}^{-1}\left(L^{2}(B)\right) \subseteq  \left[\overline{j_{2}L^{2}(B^{c})}^{||-||_{\varphi}}\right]^\perp$. On the other hand, if $g\in\mathfrak{F}_\varphi$ such that $g\in \left[{j_{2}L^{2}(B^{c})}\right]^\perp$, then by definition $(j_2 f,g)_\varphi=0$ for all $f\in L^{2}(B^{c})$. This condition can be stated as $(j_2 f,\bar{\beta}_{\pi,\varphi}\bar{\beta}_{\varphi,\pi}g)_{\varphi}=0$   for all $f\in L^{2}(B^{c})$, and by the above relations we have that $\bar{\beta}_{\varphi,\pi}g \in j_{1}^{-1}\left(L_{2}(B)\right)$, and then $ \left[\overline{j_{2}L^{2}(B^{c})}^{||-||_{\varphi}}\right]^\perp  \subseteq \bar{\beta}_{\pi,\varphi}j_{1}^{-1}\left(L^{2}(B)\right)  $. We have then shown that $\left[\overline{j_{2}L^{2}(B^{c})}^{||-||_{\varphi}}\right]^\perp  = \bar{\beta}_{\pi,\varphi}j_{1}^{-1}\left(L^{2}(B)\right)  $. Recalling the definitions (\ref{mapauno}) and exchanging $B$ with $B^c$, we get
\[\mathsf{S}_{R}(B)^{\perp}=\mathsf{S}_{I}(B^c)\]
This is one of the relations we wanted to prove. The remaining one is obtained by taking orthogonal complement to this one (noticing that $({\mathsf{S}_R}^\perp)^\perp=\mathsf{S}_R$ since $\mathsf{S}_R$ is closed) and exchanging again $B$ with $B^c$.
\end{proof}
The above theorem states that Haag duality holds in the first quantization context. It is shown in the Weyl formulation, but it can be translated, as we will explain in the next subsection, to the Segal formulation.

\subsection{Duality in Segal form}

Given a real closed subspace $\mathsf{S}$ of $H$, let us define its symplectic complement $\mathsf{S}'\subset H$ that is also a real subspace of $H$ as one can see using the $\mathbb{R}$-linearity of the inner product of $H$ in the first argument:
\begin{equation}
     \mathsf{S}'=\left\{h \in H |\quad  \sigma(h,f)=0 , \quad \forall f \in \mathsf{S}\right\}.
\end{equation}
If $g\in \mathsf{S}'\subseteq H$, then by  \eqref{sigma},  $(g,\beta f)_{H}=0$ which implies $(\beta g,f)_{H}=0$. Namely, $\beta g\in \mathsf{S}^{\perp_{H}}$. Reversing the line of reasoning, 
\begin{equation}\label{relacionfs}
	\mathsf{S}^{\perp_{H}}=\beta\mathsf{S}'.
\end{equation}
In addition, due to the fact that the inner product in $K$ is the restriction of the one in $H$, we have for $\mathsf{S}\subset K $ 
\begin{equation}\label{super}
	\mathsf{S}^{\perp_{K}}=\beta\mathsf{S}'\cap K,
\end{equation}

Within the Segal presentation, the duality for the first map is stated as
\begin{equation*}
 \boxed{\mathsf{S}_{S}'(\mathcal{O})=\mathsf{S}_{S}(\mathcal{O'})}
\end{equation*}
for any $\mathcal{O}=C(B)$ with $B$ an open subset of $
\mathbb{R}^3$. 

Thanks to  \eqref{relacionfs} we can rewrite it as 
\begin{equation}\label{haagsegal}
	\beta \left(\mathsf{S}_{S}(\mathcal{O})^{\perp_{H}}\right)=\mathsf{S}_{S}(\mathcal{O'})
\end{equation}
\begin{lemma}
    The duality property in the context of first quantization holds in Segal form if and only if it holds in Weyl form. 
\end{lemma}
\begin{proof}
Let us first assume that \eqref{haagsegal} holds and then we have to show that \eqref{dualidadweyl} holds. In order to approach this, it is important to note that a subspace $\mathsf{S}\subseteq H=K\oplus_{\mathbb{R}} \beta K$ defines two subspaces in $K$, $P_{K}(\mathsf{S})$ and $\beta P_{\beta K}(\mathsf{S})$ (where $P$'s are projectors with self-explained notation). In this way the excision generates a map $\varepsilon$ from subspaces of $H$ to a pair of subspaces of $K \simeq \mathfrak{F}_\varphi $,

\[\varepsilon(\mathsf{S}_{S}):=(P_{K}(\mathsf{S}_{S}),\beta P_{\beta K}(\mathsf{S}_{S}))\simeq(\mathsf{S}_{R},\mathsf{S}_{I})\]
Now, the duality for subspaces of $K$ in the Weyl presentation comes from applying the map $\varepsilon$ to \eqref{haagsegal} and taking into account \eqref{firstmapsrelation}. Indeed, for some $\mathcal{O}=C(B)$,
\begin{align*}
	\varepsilon(\beta \mathsf{S}_{S}(\mathcal{O})^{\perp_{H}})&=\varepsilon(\mathsf{S}_{S}(\mathcal{O'})) \\
	(P_K \beta \mathsf{S}_{S}(\mathcal{O})^{\perp_{H}},\beta P_{\beta K}\beta \mathsf{S}_{S}(\mathcal{O})^{\perp_{H}})&\simeq(\mathsf{S}_{R}(B^c),\mathsf{S}_{I}(B^c)) \\
	(P_K \beta \mathsf{S}_{S}(\mathcal{O})^{\perp_{H}}, P_{K} \mathsf{S}_{S}(\mathcal{O})^{\perp_{H}})&\simeq(\mathsf{S}_{R}(B^c),\mathsf{S}_{I}(B^c)) \\
		(\mathsf{S}_{I}(B)^{\perp_{K}},\mathsf{S}_{R}(B)^{\perp_{K}})&=(\mathsf{S}_{R}(B^c),\mathsf{S}_{I}(B^c)). 	
\end{align*}
We have used the fact\footnote{This is so since given $h=\alpha_{R}+\beta \alpha_{I} \in S$ with $\alpha's \in K$, applying $\beta$: $\beta h=\beta \alpha_{R}+\beta^{2} \alpha_{I}=\beta \alpha_{R}-\alpha_{I}$. In this way $P_{K}(\beta h)=-\alpha_{I}$ as well as  $\beta P_{\beta K}(h)=\beta \beta \alpha_{I}=-\alpha_{I}$.} that $P_{K}(\beta \mathsf{S})=\beta P_{\beta K}(\mathsf{S})$ (i.e. $(\beta \mathsf{S})_{R}=(\mathsf{S})_{I}$) in the third line, and  both $P_K \mathsf{S}^{\perp_H}=(P_K \mathsf{S})^{\perp_K}$ and $(\beta \mathsf{S}^{\perp_H})^{\perp_H}=\beta \mathsf{S}$ in the last line.  

Reciprocally using  \eqref{firstmapsrelation} we can obtain the Segal's version of the duality from Weyl's version in the following way\footnote{We identify $K$ and $\mathfrak{F}_{\varphi}$ , $H$ with $\mathfrak{F}_{\varphi} \oplus_{\mathbb{R}}\mathfrak{F}_{\pi}$,  and $\beta$ with $\bar{\beta}$ in order to make the expressions more readable.}
\begin{align*}
    S_{S}(C(B))'&=\beta S_{S}(C(B))^{\perp_{H}}\simeq \beta \left[S_{R}(B)+\beta S_{I}(B)\right]^{\perp_{H}} \\
    &=\beta \left[S_{R}(B)^{\perp_{H}}\cap (\beta S_{I}(B))^{\perp_{H}}\right]=\beta \left[S_{R}(B)^{\perp_{H}}\cap \beta S_{I}(B)^{\perp_{H}}\right]  \\
    &=S_{I}(B)^{\perp_{H}}\cap\beta S_{R}(B)^{\perp_{H}}=\left[S_{I}(B)^{\perp_{K}}\oplus \beta K\right]\cap \beta \left[S_{R}(B)^{\perp_{K}}\oplus \beta K\right] \\
    &=\left[S_{I}(B)^{\perp_{K}}\oplus \beta K\right]\cap  \left[K \oplus \beta S_{R}(B)^{\perp_{K}}\right] =(S_{I}(B)^{\perp_{K}}\cap K)\oplus \beta(K\cap S_{R}(B)^{\perp_{K}}) \\
    &=S_{I}(B)^{\perp_{K}}\oplus \beta S_{R}(B)^{\perp_{K}}=S_{R}(B^{c})^{\perp_{K}}\oplus \beta S_{I}(B^{c})^{\perp_{K}} \\
    &\simeq S_{S}(C(B^{c}))=S_{S}(C(B)'),
\end{align*}
were we used the fact that $(\beta \mathsf{S})^{\perp_{H}}=\beta \mathsf{S}^{\perp_{H}}$ and the duality in Weyl form.
\end{proof}

\begin{remark}
We have shown in Theorem \eqref{haagteoremauno} that the duality holds in the Weyl formalism, so by the previous lemma it also holds in Segal form.
\end{remark}

\section{Second quantization and Haag duality}\label{secondquantizationsection}

Having defined in the previous Section the first quantization map $\mathsf{S}$ that assigns real vector subspaces to regions of spacetime, it is now time to introduce the second quantization. This means to define a map $R$ that assigns to each vector subspace a von Neumann algebra. Instead of following  \cite{araki1964neumann} in this second stage of quantization, we make a $\pi/2$ turn and start following \cite{eckmann1973application}. The reason for this choice is that the latter makes use of the power of Tomita-Takesaki modular theory, and this results in a more direct proof of Haag duality. As we will explain, at the end we modify the proof in \cite{eckmann1973application}  in order to make it even simpler. 

In \cite{eckmann1973application}, the authors work exclusively in the Weyl context, and we shall do the same most of the time. However, it is possible without much effort to translate the results of the following subsections to the Segal context. We will comment on this along the way. 

\subsection{Second quantization}\label{secondq}

\subsubsection{Weyl form}

Let us recall that, according to the previous Section, the $\mathsf{S}$ map in the Weyl presentation actually gives two subspaces of $\mathfrak{F}_\varphi\simeq K$ for each spacelike region $B$. Then, the second quantization map $R_{W}$ should take two subspaces $K_1$ and $K_2$ of $K$ and give a von Neumann subalgebra $R_{W}(K_1,K_2)$ of the representation of Weyl's  CCR-algebra given in section \ref{weylrepresentation}. To be more precise  $R_{W}(K_1,K_2)$ is generated by $U_{F}(f)=e^{i\varphi(f)}$ and $V_{F}(g)=e^{i\pi(g)}$, with $f$ and $g$ in $K_{1}$ and $K_{2}$ respectively and obeying the relations \eqref{weylrep}. As can be guessed, given $K_1$ and $K_2$, let us define
\begin{equation}\label{segundomapa}
	R_{W}(K_{1},K_{2})=\{ e^{i\varphi(f)}e^{i\pi(g)}|\quad f\in K_{1}, g\in K_{2} \}''.
\end{equation}
The double commutant ensures that $R_{W}(K_1,K_2)$ is a von Neumann algebra, since the bicommutant of a selfadjoint set is automatically a von Neumann algebra. We recall that concrete von Neumann algebras are both weakly and strongly closed as subalgebras of the linear bounded operators in a Hilbert space. 

Let us describe more accurately how $\varphi(f)$ and $\pi(g)$ are introduced. The real part of the complex Fock space $\mathfrak{H}_{T}(L)$ (denoted by $\mathcal{F}_{r}$) consists of real linear combinations of symmetric products among elements such that their second component of $L$ (thought as a complexification of $K$) is zero. That is
$$\mathcal{F}_{r}=\bigoplus_{n=0}^{\infty}\mathcal{F}_{r}^{(n)},$$
where  $\mathcal{F}_{r}^{(0)}=\mathbb{R}$ and 
$$\mathcal{F}_{r}^{(n)}:=\bigodot_{i=1}^{n} K$$ for $n \geq 1$. Here $\odot$ means the symmetrized tensor product. Then, we can recover our complex Fock Hilbert space of subsection \ref{Fockrep} as the complexification of $\mathcal{F}_{r}$, $\mathfrak{H}_{T}(L)=\mathcal{F}_{r}+i \mathcal{F}_{r}$. We denote $\mathcal{F}^{(n)}:=\mathcal{F}_{r}^{(n)}+i\mathcal{F}_{r}^{(n)}$, so for example $\mathcal{F}^{(1)}\simeq L$. Then, as usual, one introduces creation and annihilation operators as explained in subsection \ref{weylrepresentation}.

Now, let us denote $\Omega$ the $1\in \mathcal{F}^{(0)}$, namely the vacuum vector.  Then  for $f, g\in K=\mathcal{F}^{(1)}_{r}$ we have

\begin{equation}\label{operadores}
	\varphi(f)\Omega=\frac{1}{\sqrt{2}}f  \hspace{1cm } i\pi(g)\Omega=-\frac{1}{\sqrt{2}}g.
\end{equation}

In order to construct the second quantization map we need to exponentiate $\varphi(f)$ and $\pi(g)$. However, some care must be taken since their domains are only a dense subset in $\mathfrak{H}_{T}(L)$. By using functional analysis of unbounded self-adjoint operators, we can claim that the operator $U_{F}(t)=e^{it\varphi(f)}$ is bounded, with $t\in\mathbb{R}$, since $\varphi(f)$ is self-adjoint (see footnote \ref{footnoteselfadjoint}). Even more $U_{F}(t)$ is a strongly  unitary  group acting on the complex Hilbert space $\mathfrak{H}_{T}(L)$ and by Stone's theorem there is a unique self-adjoint operator which generates such group. This is precisely $\varphi(f)$. The same holds for $e^{i\pi(g)}$. In short, we have the concrete realization \eqref{segundomapa} as an algebra of bounded operators on the complex Fock space $\mathfrak{H}_{T}(L)$.

\subsubsection{Segal form}
An analogous construction can be done using the  representation $W_F$ of the Segal CCR-algebra as in Section \ref{weylrepresentation}. The second quantization map is
\begin{equation}\label{segalsecmap}
	R_{S}(Y)=\left\{ e^{i \chi(h)} |\quad h\in Y\right\}'',
\end{equation}
where $Y$ is a closed subspace of $H$ and $\chi$ is related to $\varphi$ and $\pi$ as in \eqref{chi}. As the first quantization map \eqref{mapados} this also satisfies  isotony: $R_{S}(Y_1)\subseteq R_{S}(Y_2)$ if   $Y_{1}\subseteq Y_{2}$,  and additivity: $R_{S}(\overline{\bigoplus_{\lambda\in \Lambda}Y_{\lambda}})=(\bigcup_{\lambda\in \lambda}R_{S}(Y_{\lambda}))''$ (see theorem 1 in \cite{araki1963lattice}). Note that $(\bigcup_{\lambda\in \lambda}R_{S}(Y_{\lambda}))''$ is the smallest von Neumann algebra containing each $R_{S}(Y_{\lambda})$. These properties can also be stated in terms of the represented Weyl  CCR-subalgebras $R_{W}(K_{1},K_{2})$ previously defined. 

It is straightforward to see from \eqref{weylsegal} and the definitions \eqref{segundomapa} and \eqref{segalsecmap} that 
\begin{equation}
    R_W(K_1,K_2)=R_S(K_1 \oplus_{\mathbb{R}} \beta K_2), \qquad K_1,K_2 \subset K.
    \label{Requality}
\end{equation}

\subsection{Duality on causal diamonds and wedges}\label{conos}

Let us consider $B \subset \mathbb{R}^3$ and its complement $B^c\subset \mathbb{R}^3$, we would like to show  that
\begin{equation}
    (R_{W}\circ \mathsf{S}_{W})(B^c)=(R_{W}\circ \mathsf{S}_{W})(B)'\qquad \text{(Haag duality in Weyl form)}
\end{equation} Namely,  $R_{W}(\mathsf{S}_{R}(B^c),\mathsf{S}_{I}(B^c))=R_{W}(\mathsf{S}_{R}(B),\mathsf{S}_{I}(B))'$, which by Theorem \ref{haagteoremauno} is equivalent to  \begin{equation}
    R_{W}(\mathsf{S}_{I}(B)^{\perp_{K}},\mathsf{S}_{R}(B)^{\perp_{K}})=R_{W}(\mathsf{S}_{R}(B),\mathsf{S}_{I}(B))'.
\end{equation} 
Therefore our goal is to show that  $R_{W}(K_{2}^{\perp_{K}},K_{1}^{\perp_{K}})=R_{W}(K_{1},K_{2})'$, for $K_{1},K_{2}\subseteq K$ coming from a first quantization map. Actually, from the definition of the second map \eqref{segundomapa} it is straightforward using the CCR  \eqref{weylrep} to see that $R_{W}(K_{2}^{\perp_{K}},K_{1}^{\perp_{K}})\subseteq R_{W}(K_{1},K_{2})'$, so called the \textit{locality property}. We shall give the proof of the opposite inclusion $R_{W}(K_{1},K_{2})'\subseteq   R_{W}(K_{2}^{\perp_{K}},K_{1}^{\perp_{K}})$ in the following  subsection by using the power of  Tomita-Takesaki modular theory. Note that the analogous expression in Segal from is,
\begin{equation}\label{seconddualitysegal}
    R_S(H')=R_S(H)'.
\end{equation}

Before proving $R_{W}(K_{1},K_{2})'\subseteq   R_{W}(K_{2}^{\perp_{K}},K_{1}^{\perp_{K}})$ we introduce the  nets of algebras in Weyl and Segal form. Let  $N_{W}$ be the map that sends any relative open set $B\subseteq \{x\in \mathbb{R}^4 | x^0=0\}$ to a  von Neumann algebra. This is the composition  $N_{W}=R_{W}\circ \mathsf{S}_{W}$, which defines the net of algebras
\begin{equation*}
	N_{W}(B)=\{ e^{i\varphi(f)}e^{i\pi(g)}|\quad f\in \mathsf{S}_{R}(B), g\in \mathsf{S}_{I}(B) \}'',
\end{equation*}
where the subindex  $W$ stands for subalgebras of the Weyl CCR-algebra representation, and $\mathsf{S}_{R}(B)$ and $\mathsf{S}_{I}(B)$ are as in \eqref{mapauno}. Analogously, for the Segal formulation we have 
\begin{equation*}
	N_{S}(\mathcal{O})=\{ e^{i\chi(h)}|\quad h\in \mathsf{S}_{S}(\mathcal{O})\}'',
\end{equation*}
At the moment, we have defined both nets, in Weyl form and Segal form. The former takes regions at fixed time while the latter takes spacetime regions. We want to relate both nets. Just as in the previous section, we consider a region $B$ at fixed time (say  $x^{0}=0$) and its causal envelope $C(B)$ as in \eqref{Cdef}\footnote{If $B$ is bounded, then $C(B)$ is the causal diamond of the ball circumscribed in $B$, while if $B$ is unbounded $C(B)$ is a wedge limited by two light-like hyperplanes \cite{horuzhy2012introduction}.}. The following lemma is a key  result, where for just once some differentiablity condition is required for the boundary of $B$, 

\begin{lemma}[Proposition 3.3.2 in \cite{horuzhy2012introduction}]\label{lema}
    Let $B\subseteq \{x\in \mathbb{R}^4 | x^0=0\}$ be a measurable region with piecewise smooth boundary such that $int(\overline{B})=B$. Then $N_{W}(B)=N_{S}(C(B)).$  
\end{lemma}

Now assuming for  a moment that  $R_{W}(K_{1},K_{2})'= R_{W}(K_{2}^{\perp_{K}},K_{1}^{\perp_{K}})$ is valid we can conclude that  Haag duality holds for regions $C(B)$ satisfying the hypothesis  of the previous lemma, since
\begin{align*}
	N_{S}(C(B))'&=N_{W}(B)'=R_{W}(\mathsf{S}_{R}(B),\mathsf{S}_{I}(B))' \\
	&=R_{W}(\mathsf{S}_{I}(B)^{\perp_{\varphi}},\mathsf{S}_{R}(B)^{\perp_{\varphi}})=R_{W}(\mathsf{S}_{R}(B^{c}_{0}),\mathsf{S}_{I}(B^{c}_{0})) \\
	&=N_{W}(B^{c}_{0})=N_{S}(C(B^{c}_{0}))=N_{S}(C(B)').
\end{align*}
In the second line the duality for the second map in the Weyl representation was used, and in the following step  Theorem  \ref{haagteoremauno} was used. 
\begin{remark}
In Lemma \ref{lema}, it is required that the boundary of $B$ is piece-wise smooth. However, from  \eqref{firstmapsrelation} and \eqref{Requality} it follows directly that $N_S(C(B))=N_W(B)$. (See also footnote \eqref{footnote:boundary}.) This means that it is not really needed the aforementioned requirement in order to show Haag duality in the Segal form (once it is proven in Weyl form)  and this can be considered an  improvement.
\end{remark}

\begin{remark}
Haag duality also holds in more general cases. First, let us point put that it holds in any other region $g\cdot C(B) $, the transformed of $C(B)$ by an element of the Poincaré group $g=(a,\Lambda)$. This is because of Poincaré covariance of the net of algebras, namely $N(g\cdot \mathcal{O})=U_g N(\mathcal{O})U_g^{-1}$, with $U_g$ the unitary  representation on the Fock space by exponentiation of the 1-particle unitary irreducible representation $U(a,\Lambda)$, given in footnote \ref{footnoteirrep}. Indeed, calling $\mathcal{O}=C(B)$, we have
\begin{align*}
    N(g\cdot\mathcal{O})&=U_g N(\mathcal{O})U_g^{-1}=U_g N(\mathcal{O}')'U_g^{-1}\\
    &=\left(U_g N(\mathcal{O}')U_g^{-1}\right)'= \left( N(g \cdot \mathcal{O}')\right)'\\
    &=N((g\cdot \mathcal{O})')'  
\end{align*}
Second, it also holds for any non-empty relative open set $B$ inside a Cauchy surface $\Sigma$ of Minkowski spacetime. This is because there is nothing particularly exceptional about the Cauchy surface $t=0$, and the Klein-Gordon operator is a Green-hyperbolic operator and therefore there is a unique solution associated to any pair of initial conditions (with the regularity properties we already described) defined on an arbitrary Cauchy surface $\Sigma$. Many of the previous steps just change by evaluation at $\Sigma$ instead of at $t=0$ and the derivative with respect to $t$ changes to a normal derivative orthogonal to $\Sigma$. Also, one can replace the previous definition \eqref{Cdef} of $C(B)$ by the causal completion of $B$,  $C(B):=B''$, with no reference to the $x^0=0$ Cauchy surface and $B$ some non-empty relative open set in $\Sigma$. For further details see \cite{camassa2007relative}.  

\end{remark}

In general there is no reason for Haag duality to hold in other type of regions. A counter example found in \cite{araki1964neumann} is explained in detail in section \ref{contra}.

\subsection{Proof of Haag duality}\label{proof}

As we already mentioned, instead of proving the inclusion $R_{W}(K_{1},K_{2})'\subseteq   R_{W}(K_{2}^\perp,K_{1}^\perp)$ as in \cite{araki1964neumann}, we can turn our attention to the approach of \cite{eckmann1973application}. The idea is that the commutant of some von Neumann algebra $R$ acting on the Hilbert space $\mathfrak{H}_{T}(L)$ (actually any Hilbert space) can be obtained directly by modular theory.

If there is a cyclic and separating vector\footnote{A vector $\Omega$ is \textit{cyclic} for $R$ a von Neumann algebra on a Hilbert space $\mathcal{H}$,  if $R\Omega$ is dense in $\mathcal{H}$. Also, $\Omega$ is \textit{separating} for $R$ if for any $A\in R$ it holds that $A \Omega=0$ implies $A=0$.} $\Omega$ for $R$, then it is also cyclic and separating for $R'$ \cite{BR1}. Let us consider the following anti-linear operator
\begin{align}\label{tomitaop}
    S_{0} A \Omega &= A^*\Omega, \quad A \in R,
\end{align}
which is defined on a dense set by the cyclicity of $\Omega$. Moreover $S_{0}$ is closable (because $\Omega$ is separating for $R$) and its closure $S$ is called the Tomita operator \footnote{We hope there is no confusion between the Tomita operator $S$ and the first quantization maps studied in the previous sections.}. The domain of $S$ is given by 
\begin{equation*}
    D=\{x\in\mathfrak{H}_{T}(L) : \exists \{x_{n}\}_{n\in \mathbb{N}}\subseteq R\Omega \text{ such that } \lim_{n\to \infty}x_{n}=x \text{ and } S_{0}x_{n} \text{ has a limit}\}.
\end{equation*}
The Tomita operator $S$ is invertible  and then has a unique polar decomposition
\begin{equation}
    S=J \Delta^{1/2}
\end{equation}
where $J$ is anti-unitary and $\Delta=S^*S$ is positive and self-adjoint. They are called modular conjugation and modular operator respectively. Some basic properties have to be remarked. First, taking $A=1_{R}$ in equation \eqref{tomitaop} one has $S\Omega=\Omega$. Second from $S^{2}=1_{R}$ we have $J \Delta^{1/2}J \Delta^{1/2}=1_{R}$  or $\Delta^{-1/2}=J \Delta^{1/2}J$, hence 
\begin{equation*}
    J^{2}(J^{-1} \Delta^{1/2} J)=\Delta^{-1/2}=1_{R}\Delta^{-1/2},
\end{equation*}
then because of the positivity of $J^{-1} \Delta^{1/2} J$ and the uniqueness of the polar decomposition we arrive to $J^2=1_R$. The domain of $S^{\ast}$ is 
\begin{equation*}
    D'=\{x\in\mathfrak{H}_{T}(L) : \exists \{x_{n}\}_{n\in \mathbb{N}}\subseteq R'\Omega \text{ such that } \lim_{n\to \infty}x_{n}=x \text{ and } F_{0}x_{n} \text{ has a limit}\}
\end{equation*}
being $F_{0}$ the closable operator defined by (see \cite{sunder2012invitation}, Proposition 2.3.1.)
\begin{align}
    F_{0} A \Omega &= A^*\Omega, \quad A \in R'.
\end{align}

We are now ready to state an important result of Tomita and Takesaki we will use.

\begin{theorem}[Tomita-Takesaki Theorem, Th. 10.1 of \cite{Takesaki}, Th. 2.5.14 of \cite{BR1}]\label{tomita}
	Let $R$ be a von Neumann algebra on a Hilbert space $\mathcal{F}$ admitting a cyclic and separating vector $\Omega$. Then, the following two relations hold 
	\[ R'=JRJ, \]
	\[\Delta^{it}R\Delta^{-i t}=R,\]
	for all $t \in \mathbb{R}$.
\end{theorem}

The vacuum vector $\Omega$ is cyclic for the algebras $R_{W}(\mathsf{S}_{R}(B),\mathsf{S}_{I}(B))$. This is because by the Reeh-Schlieder Theorem   $\Omega$ is cyclic for the Segal algebra $N_S(\mathcal{O})$ \cite{Reeh:1961ujh}, and in our case we are interested in $\mathcal{O}$ being a double cone $C(B)$ over $B$ in which case by Lemma \ref{lema} we have $N_{S}(C(B))=N_{W}(B)$ and then $\Omega$ is cyclic for $R_{W}(\mathsf{S}_{R}(B),\mathsf{S}_{I}(B))$. Moreover it is well known that if a vector is cyclic for an algebra then is separating for its commutant\footnote{ In our case $\Omega$ is of course also cyclic for $R_{W}(K_2^\perp,K_1^\perp)\subseteq R_{W}(K_1,K_2)'$ and then we can take any $A \in R_{W}(K_1,K_2)$ such that $A\Omega=0$, any $B \in R_{W}(K_2^\perp,K_1^\perp)$, and write 
\[A (B \Omega)= B A \Omega=0, \]
which implies that $A=0$ since it annihilates a dense set in $\mathfrak{H}_{T}(L)$.}. In short, $\Omega$ is cyclic and separating for any $R_{W}(K_1,K_2)$  with the pair $(K_1, K_2)$ coming from Weyl's first quantization map (from now on we always assume this). Therefore Tomita-Takesaki Theorem \eqref{tomita} applies to our situation and in particular we have
$$R_{W}(K_{1},K_{2})'=JR_{W}(K_{1},K_{2})J.$$

Let us anticipate how the strategy of \cite{eckmann1973application} follows. Let  $\hat{J}=J|_{\mathcal{F}_{r}^{(1)}}=J|_{K}$ be the restriction of $J$ to the real one-particle space, then one important step is to show that $\hat{J}K_{1}\subseteq K_{2}^\perp$ and $\hat{J}K_{2} \subseteq K_{1}^\perp$. This in turn will imply:
\begin{align}
	Je^{i\varphi(f)}J&=e^{i\varphi(\hat{J}f)} \nonumber \\
	Je^{i\pi(g)}J&=e^{i\pi(\hat{J}g)},
\label{almosthaagduality}
\end{align} 
from which we can deduce $R_{W}(K_{1},K_{2})'\subseteq R_{W}(K_{2}^\perp,K_{1}^\perp)$ as follows. As already mentioned, an element $M\in R_{W}(K_{1},K_{2})'$ can be written as $M=JAJ$, where $A\in R_{W}(K_{1},K_{2})$, and then as the von Neumann algebras are weakly closed $A=w-\lim_{n}A_{n}$ and $M=w-\lim_{n}JA_{n}J$, where $A_{n}\in R_{W}(K_{1},K_{2})$. But any of these $A_{n}$ is a finite sum of  elements $e^{i\varphi(f_{k_n})}e^{i\pi(g_{k_n})}$, meaning there are elements $f_{k_n}\in K_{1}$ and $g_{k_n}\in K_{2}$ such that 
\begin{align*}
    A_{n}&=\sum_{k=0}^{N_{n}}a_{k_n}e^{i\varphi(f_{k_n})}e^{i\pi(g_{k_n})}, \\ 
    JA_{n}J&=\sum_{k=0}^{N_{n}}a_{k_n}Je^{i\varphi(f_{k_n})}JJe^{i\pi(g_{k_n})}J \\
    &=\sum_{k=0}^{N_{n}}a_{k_n}e^{i\varphi(\hat{J}f_{k_n})}e^{i\pi(\hat{J}g_{k_n})},
\end{align*}
that clearly belong to $R_{W}(K_{2}^\perp,K_{1}^\perp)$. Moreover as the von Neumann algebras are weakly closed, so $M=JAJ\in R_{W}(K_{2}^\perp,K_{1}^\perp)$. This shows that \eqref{almosthaagduality} implies Haag duality. 

However, in order to get an explicit formula for $\hat{J}$, in \cite{eckmann1973application} several auxiliary operators are introduced (see equation 5 there). Instead, we shall follow the main plot of that reference, but in order to construct $J$ we shall only deal with $S$ and $S^*$, since one has
\begin{equation}
    J=S(S^*S)^{-1/2}
\end{equation}

Let us start with the following Lemma, which is an extension of Lemma 4 in \cite{eckmann1973application}.
\begin{lemma}
	Let $D$ and $D'$ be the (dense) domains of $S$ and $S^*$ respectively, and  $E^{(k)}$ the orthogonal projections on $\mathcal{F}^{(k)}$, $k \in \mathbb{N}_0$. Then $J$, $S$ and $S^{\ast}$ commute with the projections,
	\begin{align*}
		E^{(k)}J&= JE^{(k)} \\
		E^{(k)}S&\subseteq SE^{(k)} \\
	    E^{(k)}S^{\ast}&\subseteq S^{\ast}E^{(k)} \\ 
	\end{align*}
Even more, the restriction of $S$ to $E^{(k)}D$ is a close densely defined operator in  $\mathcal{F}^{(k)}$, with domain $E^{(k)}D=D\cap \mathcal{F}^{(k)}$ and \textbf{leaving invariant $D\cap \mathcal{F}_{r}^{(k)}$} (which is \textbf{also left invariant by $J$}). Analogously, the restriction of $S^{\ast}$ to $E^{(k)}D'$ is a closed densely defined operator in  $\mathcal{F}^{(k)}$, with domain $E^{(k)}D'=D'\cap \mathcal{F}^{(k)}$ and \textbf{leaves invariant  $D'\cap \mathcal{F}_{r}^{(k)}$}.
\label{commutations}
\end{lemma}
\begin{proof}
Taking into account that $E^{(k)}$ and $J$ are bounded operators, the domains of the operators $E^{(k)}J$ and $JE^{(k)}$
 are the entire Fock space $\mathfrak{H}_{T}(L)$ whereas the domain of $E^{(k)}S$ is $D$. However, the domain of $SE^{(k)}$ are those vectors in Fock space such that their $k$ component is in $D$ which is a bigger set\footnote{It is possible to have a vector $x\in \mathfrak{H}_{T}(L)$ such that $x\notin D$ but $E^{(k)}x\in D$, then we can apply $S$ to it. That is why we have contentions and not equalities in the lemma.} and similarly for $S^{\ast}$ but with $D'$.

The part of this Lemma referring to  $S$ and $J$ is proven in \cite{eckmann1973application}. In contrast, the claims regarding $S^{\ast}$ are new and deserve an explanation. $S^*$, being the adjoint of  $S$,  satisfies
$$\overline{(S\Psi,\Phi)}_{\mathfrak{H}_{T}(L)}=(\Psi,S^{\ast}\Phi)_{\mathfrak{H}_{T}(L)}$$ 
for all $\Phi \in D'$ and $\Psi\in D$. Following this notation, it is possible to show that $E^{(k)}\Phi \in D'$ and $E^{(k)}S^\ast = S^\ast E^{(k)}|_{D'}$, 
\begin{align*}
    (S\Psi,E^{(k)}\Phi)_{\mathfrak{H}_{T}(L)}&=(E^{(k)}S\Psi,\Phi)_{\mathfrak{H}_{T}(L)}\\
    &=(SE^{(k)}\Psi,\Phi)_{\mathfrak{H}_{T}(L)}\\
    &=(S^\ast\Phi,E^{(k)}\Psi)_{\mathfrak{H}_{T}(L)}\\
    &=(E^{(k)}S^\ast\Phi,\Psi)_{\mathfrak{H}_{T}(L)}.
\end{align*}
We have used the orthogonality of $E^{(k)}$ and that $E^{(k)}S \subseteq SE^{(k)}$.

Let us introduce some notation, if $\Phi\in \mathcal{F}^{(1)}$, then there are $\Phi_R,\Phi_I\in \mathcal{F}_{r}$ such that $\Phi=\Phi_{R}+i\Phi_{I}$ and we note $\bar{\Phi}:=\Phi_{R}-i\Phi_{I}$ and $\Im{\Phi}:=\frac{1}{2i}(\Phi-\bar{\Phi})$, note that a vector $\Phi\in \mathcal{F}^{(1)}$ is in $\mathcal{F}_{r}$ if and only if $\Im{\Phi}=0$. In order to see that $S^*$ leaves $\mathcal{F}^{(1)}_r$  invariant, namely that $\Im{ S^*\Phi}$=0 for $\Phi \in \mathcal{F}^{(1)}_r\cap D'$,  we can proceed as follows. By taking into account that $\overline{(\Psi,\Phi)}_{\mathfrak{H}_{T}(L)}=(\overline{\Psi},\overline{\Phi})_{\mathfrak{H}_{T}(L)}$, and the definition of $S^*$ the following two equations hold,
\begin{align*}
    \overline{(S\Psi,\Phi)}_{\mathfrak{H}_{T}(L)}&=(\Psi,S^{\ast}\Phi)_{\mathfrak{H}_{T}(L)},\\
(S\overline{\Psi},\Phi)_{\mathfrak{H}_{T}(L)}&=(\Psi,\overline{S^{\ast}\Phi})_{\mathfrak{H}_{T}(L)}, 
\end{align*}
for $\Psi \in \mathcal{F}^{(1)}\cap D$ and $\Phi \in \mathcal{F}^{(1)}\cap D'$. From now on we further assume  $\Phi \in \mathcal{F}^{(1)}_r\cap D'$, that is $\overline{\Phi}=\Phi$. By subtracting the second one to the first one we get
\begin{equation}\label{proofofinvariance}
    \overline{(S\Psi,\Phi)}_{\mathfrak{H}_{T}(L)}-(S\overline{\Psi},\Phi)_{\mathfrak{H}_{T}(L)}=2i(\Psi,\Im{S^{\ast}\Phi})_{\mathfrak{H}_{T}(L)}.
\end{equation}
We consider two cases separately. First let us assume that $\Psi\in \mathcal{F}^{(1)}_r\cap D$, i.e. it is purely real: $\Psi=\overline{\Psi}$.  Then the LHS of \eqref{proofofinvariance} reads $(2i\Im{S\Psi},\Phi)$ which is zero since
we know that $S$ leaves $\mathcal{F}^{(1)}_r$ invariant by Lemma 4 of \cite{eckmann1973application}, implying that $\Im{S^{\ast}\Phi}\perp  (\mathcal{F}^{(1)}_r\cap D)$. Now let us assume that  $\Psi=i\chi$, with $\chi \in \mathcal{F}^{(1)}_r\cap D$. Then the LHS of \eqref{proofofinvariance} reads $i(S \chi,\Phi-\overline{\Phi})=0$, where we have used that $S$ is anti-linear and that $\overline{\Phi}=\Phi$. This implies that   $\Im{S^{\ast}\Phi}\perp  (i\mathcal{F}^{(1)}_r\cap D)$, and taking into account the first case we conclude that $\Im{S^{\ast}\Phi}$ is orthogonal to the dense set $\mathcal{F}^{(1)}\cap D $ and thus vanishes, which is what we wanted to prove. We will neither use nor prove the case $k>1$, but it follows from similar arguments as those used for $k=1$.
\end{proof}

We can examine further the domain of $S$ restricted to the one-particle real space $\mathcal{F}^{(1)}_{r}\simeq K$. The main point we want to stress is that $S$ can act on one-particle states, although they are not \textit{a priori} of the form $A\Omega$, with $A\in R_{W}(K_1,K_2)$.
\begin{lemma}
   If $K_1$ and $K_2$ are in generic position (i.e. $K_1\cap K_2=K_1\cap K_2^\perp=K_2\cap K_1^\perp=K_1^\perp\cap K_2^\perp=\varnothing$) then $K_1+K_2 \subseteq D \cap \mathcal{F}^{(1)}_r$ so $S$ can act on these vectors.
\end{lemma}  
\begin{proof}
 In order to see this, let us take a generic one-particle state 
\begin{align*}
    \Psi&=\frac{1}{\sqrt{2}}\left( f_1+if_2+ig_1-g_2\right),\quad f_i \in K_1, \quad g_i\in K_2 \\
    &=\left(\varphi(f_1)+i\varphi(f_2)+\pi(g_1)+i\pi(g_2)\right)\Omega
\end{align*} 
It can be written as the limit (in the norm topology of $\mathfrak{H}_{T}(L)$, see \cite{hall2013quantum}, Proposition 10.14)
$$\lim_{\lambda\rightarrow 0}\lambda^{-1}\left[ -i(e^{\lambda i (\varphi(f_1)+\pi(g_1))}-1)+(e^{\lambda i (\varphi(f_2)+\pi(g_2))}-1) \right]\Omega  $$
of states in the domain of $S$. Also, the limit
$$\lim_{\lambda\rightarrow 0}\lambda^{-1}S\left[ -i(e^{\lambda i (\varphi(f_1)+\pi(g_1))}-1)+(e^{\lambda i (\varphi(f_2)+\pi(g_2))}-1) \right]\Omega  $$
exists and gives $\left(\varphi(f_1)-i\varphi(f_2)+\pi(g_1)-i\pi(g_2)\right)\Omega=\frac{1}{\sqrt{2}}\left( f_1-if_2+ig_1+g_2\right)$. Then, since $S$ is closed, we get
\begin{equation}\label{Soneparticle}
    S (f_1+if_2+ig_1-g_2)= f_1-if_2+ig_1+g_2.
\end{equation}
This shows explicitly how $S$ acts on the one-particle states and also that $K_1+K_2 \subseteq D \cap \mathcal{F}^{(1)}_r$. Even more, since $(K_1+K_2)^\perp=K_1^\perp\cap K_2^\perp=0$ (because of the generic position property, see below), we have that $K_1+K_2$ is actually dense in $\mathcal{F}^{(1)}_r$ and so it is dense in the domain of $S$ restricted to the one particle space.  In \cite{eckmann1973application} the authors claim that $K_1+K_2=D\cap \mathcal{F}_r^{(1)}$, however we do not see how to prove $D\cap \mathcal{F}_r^{(1)} \subseteq K_1+K_2$. In any case it will be sufficient with what we have claimed above.
\end{proof}

The fact that $K_1$ and $K_2$ are in generic position was assumed in the proof of Lemma 4 in \cite{eckmann1973application}, and here we show in Appendix \ref{ap4}, following \cite{araki1964neumann}, that $K_1$ and $K_2$ are indeed in generic position if they are given by the first quantization maps\footnote{The fact that $K_1$ and $K_2$ are in generic position immediately implies that $S=K_1+\beta K_2$ is separating, namely $S \cap \beta S=\varnothing$. Also, it can be shown that $S=K_1+\beta K_2$ is cyclic ($S+\beta S$ is dense in $H$) by means of the duality in the Weyl formulation and the generic position property. We suspect this can be shown even without invoking the duality in the Weyl formulation. A closed real subspace $S$ that is cyclic and separating is called \textit{standard}.} $S_{R}(B)$ and $S_{I}(B)$.

The following theorem, which holds for real or complex Hilbert spaces, sometimes referred to as Halmos’ two projections theorem,  is a key result  in order to have more control over $J$ later on, by first writing the restrictions  to $\mathcal{F}_{r}^{(1)}\simeq K$ of $S$ and $S^{\ast}$ in ``matrix form''. 

\begin{theorem}[Stated as in Lemma 5 of \cite{eckmann1973application}]\label{joyita}
	Let $K$ be a Hilbert space and $K_{1}$ and $K_{2}$ two closed subspaces of $K$ in generic position. Then, there exists another Hilbert space $\mathcal{K}_{\ast}$ and a positive contraction $T$ on it, with  $\text{ker}\, T=\text{ker}\, (1-T)=0$, such that the pair $\{K_{1},K_{2}\}$ is unitarily equivalent to the pair $\{\Gamma(T),\Gamma(-T)\}$, this means there is a unitary (or orthogonal) map $U:K\to \mathcal{K}_{\ast} \oplus_{(\mathbb{R})} \mathcal{K}_{\ast}$ which carries $K_{1}$ to the graph of $T$ and $K_{2}$ to the graph of $-T$. Here $\Gamma(T)$ denotes the graph of $T$. 
\end{theorem}
For the proof of this theorem see \cite{halmos1969two} (in particular Theorem 1 therein) or \cite{rieffel1977bounded} theorem 2.4, where it can be seen that it is possible to identify $K$ with $\mathcal{K}_{\ast}\oplus_{\mathbb{R}}\mathcal{K}_{\ast}$ and that $T$ can be chosen to be self-adjoint. A contraction is a linear operator with norm less or equal to 1. 

Let us apply Halmos' theorem to the case at hand. Calling $U$ to the orthogonal map sending $K_{1}$ to $ \Gamma(T)$ and $K_2$ to $\Gamma(-T)$, then given $f\in K_{1}$ and $g\in K_2$, there exist  $h,h'\in\mathcal{K}_{_\ast}$ such that $$U(f)=\begin{pmatrix}
	h \\ Th
\end{pmatrix},\qquad U(g)=\begin{pmatrix}
	h' \\ -Th'
\end{pmatrix}.$$ 
The converse also holds (since $U$ is invertible), namely for any element of the form as in the RHSs above, there are unique $f\in K_1$ and $g \in K_2$ such that the these equations hold. Even more, we can see that $K_1^\perp$ and $K_2^\perp$ are also identified with elements of $\mathcal{K}_{_*}\oplus_{\mathbb{R}}\mathcal{K}_{_*}$. Let us take $x,y\in \mathcal{K}_{\ast}$, so we have
\begin{align*}
    \begin{pmatrix}
    x \\
    y
    \end{pmatrix} \in \Gamma(T)^{\perp} &\Leftrightarrow
    \left( \begin{pmatrix}
    x \\
    y
    \end{pmatrix}, \begin{pmatrix}
    h \\
    T(h)
    \end{pmatrix}\right)_{\mathcal{K}_{_*}\oplus_{\mathbb{R}}\mathcal{K}_{_*}}=0,\quad \forall h\in \mathcal{K}_{\ast} \\
     &\Leftrightarrow \left(x,h\right)_{\mathcal{K}_{\ast}}=(-y,T(h))_{\mathcal{K}_{\ast}}, \quad\forall  h\in \mathcal{K}_{\ast} \\
     &\Leftrightarrow \begin{pmatrix}
         -y \\ x
     \end{pmatrix}\in \Gamma(T^{\ast})=\Gamma(T),
\end{align*}
where in the last line we use the fact that $T$ can be chosen self-adjoint, and this implies that $x=-T(y)$. Then $\begin{pmatrix}
    x \\
    y
    \end{pmatrix}\in \Gamma(T)^{\perp}$  if and only if $ \begin{pmatrix}
    x \\
    y
    \end{pmatrix}= \begin{pmatrix}
    -T(y) \\
    y
    \end{pmatrix}$
    so we have the following identifications
\begin{equation}
    (U K_1)^\perp =\left\{ \begin{pmatrix}
	-T y \\ y
\end{pmatrix}, y\in \mathcal{K}_{\ast}\right\},\qquad (U K_2)^\perp=\left\{\begin{pmatrix}
	T y \\ y
\end{pmatrix}, y\in \mathcal{K}_{\ast} \right\}
\label{Uorthogonal}
\end{equation}

Let us call $\hat{S}$ the restriction of $S$ to $D\cap \mathcal{F}_{r}^{(1)}$. Considering $f\in K_{1}$ and $g\in K_{2}$, using \eqref{Soneparticle} we can write
\begin{align*}
	\hat{S}(f)=f,\quad \hat{S}(g)=-g.
\end{align*}
These expressions can be used in order to define $\tilde{S}$, which is induced on $\Gamma(T)$ by $S$ as 
$\tilde{S}=U\hat{S}U^{-1}$.
Therefore, for any $h\in \mathcal{K}_{\ast}$,
\begin{align*}
	\tilde{S}\begin{pmatrix}
		h \\ Th
	\end{pmatrix}&=\tilde{S}U(f)=U\hat{S}(f)=U(f)=\begin{pmatrix}
	h \\ Th
\end{pmatrix}  \\ 
	\tilde{S}\begin{pmatrix}
	h \\ -Th
\end{pmatrix}&=\tilde{S}U(g)=U\hat{S}(g)=U(-g)=-\begin{pmatrix}
	h \\ -Th
\end{pmatrix}, 
\end{align*}
from which we deduce by linear combinations,
\begin{equation*}
	\tilde{S}\begin{pmatrix}
		h \\ 0
	\end{pmatrix}=\begin{pmatrix}
	0 \\ Th
\end{pmatrix}, \hspace*{1cm} 	\tilde{S}\begin{pmatrix}
0 \\ Th
\end{pmatrix}=\begin{pmatrix}
h \\ 0
\end{pmatrix}.
\end{equation*} 
This permits to identify
\begin{equation*}
	\tilde{S}=\begin{pmatrix}
		0 & T^{-1}\\ T & 0
	\end{pmatrix} \hspace*{0.5cm}\text{and then }\quad 	\tilde{S}^{\ast}=\begin{pmatrix}
		0 & T \\ T^{-1} & 0
	\end{pmatrix}.
	\label{SandSstar}
\end{equation*} 
Now we consider the polar decomposition of $S$, $S=J(S^{\ast}S)^{\frac{1}{2}}$, and restrict it to  $K_1+K_2$. Namely $\tilde{S}=\tilde{J}(\tilde{S}^{\ast}\tilde{S})^{\frac{1}{2}}$, where $\tilde{J}=U\hat{J}U^{-1}$. Explicitly,  by using \eqref{SandSstar}
\begin{align*}
   \begin{pmatrix}
   	0 & T^{-1}\\ T & 0
   \end{pmatrix}=\tilde{J}\left[\begin{pmatrix}
   0 & T \\ T^{-1} & 0
\end{pmatrix} \begin{pmatrix}
   0 & T^{-1}\\ T & 0
\end{pmatrix} \right]^{\frac{1}{2}} \Longrightarrow \tilde{J}=\begin{pmatrix}
	0 & 1\\ 1 & 0
\end{pmatrix}
\end{align*}
which coincides with the expression in  \cite{eckmann1973application} (there is no distinction between $\tilde{J}$ and $\hat{J}$ in that reference).  Then, 
\begin{equation}\label{jota}
	\tilde{J}\begin{pmatrix}
		h \\ Th
	\end{pmatrix}=\begin{pmatrix}
		Th \\ h
	\end{pmatrix} \hspace{0.5cm}\text{ y }\hspace{0.5cm} \tilde{J}\begin{pmatrix}
		h \\ -Th
	\end{pmatrix}=\begin{pmatrix}
		-Th \\ h
	\end{pmatrix}
\end{equation}
From here and \eqref{Uorthogonal}  we recognize that $\hat{J}$ acts on $K_1$ by sending it to $K_2^\perp$:
\begin{equation*}
	\hat{J}(f)=U^{-1}\tilde{J}U(f)=U^{-1}\tilde{J}\begin{pmatrix}
		h \\ Th
	\end{pmatrix}=U^{-1}\begin{pmatrix}
	Th \\ h 
\end{pmatrix}\in K_2^\perp. 
\end{equation*}
In analogous way we can show that $\hat{J}$ sends $K_2$ to $K_1^\perp$. In summary, we have just proved the following 
\begin{lemma}
    If $K_{1},K_{2}\subseteq K$ are in generic position then 
\begin{equation}\label{JKs}
    \hat{J}K_1 = K_2^\perp,\qquad \hat{J}K_2 = K_1^\perp
\end{equation}
\end{lemma}
We have succeeded in accomplishing the first step of the anticipated strategy. We had departed from the technical computations of \cite{eckmann1973application} in order to avoid introducing extra unbounded operators and managed to arrive to the same conclusion. From now on we follow closely that reference, clarifying a few small but technically relevant steps.

We want to prove the following intermediate lemma that claims a stronger fact than  \eqref{JKs}, since it roughly says that it holds in general for the fields $\varphi(f)$ and $\pi(g)$, not only for $J$ acting on the one-particle vectors.

\begin{lemma}\label{ultimo}
	Let  $f\in K_{1}$ and $g\in K_{2}$. Then
	\begin{equation*}
	J\varphi(f)J=\varphi(\hat{J}f) \hspace{1cm} 	J\pi(g)J=\pi(\hat{J}g),
	\end{equation*} 
with $\hat{J}f\in K_{2}^{\perp}$ and $\hat{J}g\in K_{1}^{\perp}.$ 
\end{lemma}
Before proving this lemma let us state and prove the following important corollary which by the discussion after \eqref{almosthaagduality} is sufficient to claim that Haag duality holds.
\begin{corollary}
With the conditions of the above lemma 
\begin{align*}
    Je^{i\varphi(f)}J&=e^{i\varphi(\hat{J}f)} \\
    Je^{i\pi(g)}J&=e^{i\pi(\hat{J}g)}
\end{align*}
\end{corollary}
\begin{proof}
Let us assume for a moment the validity of lemma \ref{ultimo} to prove the corollary. Let $A\in R_{W}(K_{1},K_{2})$, then by  Tomita-Takesaki Theorem, for  all $f\in K_{1}$ we have, 
\begin{equation*}
	Je^{i\varphi(f)}JA\Omega=AJe^{i\varphi(f)}J\Omega,
\end{equation*}
and then,
\begin{align*}
	Je^{i\varphi(f)}JA\Omega&=AJ\sum_{n=0}^{\infty}\frac{\left(i\varphi(f)\right)^{n}}{n!}J\Omega=A\sum_{n=0}^{\infty}\frac{\left(iJ\varphi(f)J\right)^{n}}{n!}\Omega \\
	&=A\sum_{n=0}^{\infty}\frac{\left(i\varphi(\hat{J}f)\right)^{n}}{n!}\Omega=Ae^{i\varphi(\hat{J}f)}\Omega=e^{i\varphi(\hat{J}f)}A\Omega,
\end{align*}
where in the last equality we used that $e^{i\varphi(\hat{J}f)}\in R_{W}(K_{2}^{\perp},K_{1}^{\perp})\subseteq R_{W}(K_{1},K_{2})'$ by locality. Thus $Je^{i\varphi(f)}J=e^{i\varphi(\hat{J}f)}$ on a dense set, and by continuity in all $\mathfrak{H}_{T}(L)$. Similarly the analogous claim is shown for  $\pi(g)$.
\end{proof}
As we have already mentioned above the Haag duality follows from this corollary. Let us now prove the lemma \ref{ultimo}.
\begin{proof}
From the fact that $S\Omega=\Omega$ and $J^2=1$  we have $\Delta^{1/2}\Omega=J \Omega$. As $J$ commutes with $E^{(0)}$ by Lemma \ref{commutations} and noticing that $\Delta$ is positive,  it follows\footnote{Another way to see this is the following: it is a fact that $S^*$ is the Tomita operator of $R_{W}(K_1,K_2)'$ (see \cite{BR1}), which is well defined since $\Omega$ is a cyclic and separating vector for $R_{W}(K_1, K_2)$ and this implies it is also cyclic and separating for its commutant $R_{W}(K_1,K_2)'$. Then, $S^*\Omega=\Omega$ and $J\Omega=S(S^* S)^{-1/2}\Omega=\Omega$.} that $J\Omega=\Omega$. Then, if $f\in K_{1}$ and $g\in K_{2}$
\begin{align}\label{phi}
	J\varphi(f)J\Omega&=\frac{1}{\sqrt{2}}\hat{J}f=\varphi(\hat{J}f)\Omega \\
	J\pi(g)J\Omega&=-i\frac{1}{\sqrt{2}}\hat{J}g=\pi(\hat{J}g)\Omega	
\end{align}
where  $\hat{J}f\in K_{2}^{\perp}$ and $\hat{J}g\in K_{1}^{\perp}$ as we already observed in \eqref{JKs}.

Let us show now that actually  $J\varphi(f)J=\varphi(\hat{J}f)$.  Let $f \in K_1$ and $U(t)=e^{it\varphi(f)}$ be a one-parameter group in $R_{W}(K_1,K_2)$, then for any $B\in R_{W}(K_{1},K_{2})'$ it holds  $BU(t)=U(t)B$ and for all $\Psi\in \text{Dom}(\varphi(f))$ we have $$\varphi(f)\Psi=\lim_{t\to 0}\dfrac{1}{i}\dfrac{U(t)\Psi-\Psi}{t}$$ 
in the norm topology of $\mathfrak{H}_{T}(L)$ \cite{hall2013quantum}. Then we can consider $BU(t)\Psi=U(t)B\Psi$ and take the derivative at $t=0$ and get  $B\varphi(f)\Psi=\varphi(f)B\Psi$ for all $\Psi\in \text{Dom}(\varphi(f))$, i.e. $B\varphi(f)\subseteq\varphi(f)B$. Let us be more clear about these last steps:
\begin{align*}
    B \varphi(f) \Psi &=B \lim_{t\to 0}\frac{1}{i}\dfrac{U(t)\Psi-\Psi}{t}\\
    &=\lim_{t\to 0}\frac{1}{i}\dfrac{BU(t)\Psi-B\Psi}{t}\\
    &=\lim_{t\to 0}\frac{1}{i}\dfrac{U(t)B\Psi-B\Psi}{t}\\
    &=\lim_{t\to 0}\frac{1}{i}\dfrac{U(t)-1}{t}B\Psi\\
    &=\varphi(f)B \Psi
\end{align*}
where in the second equality we used that $B$ is continuous, in the third line we used that $B\in R_{W}(K_1,K_2)'$, and in the fifth line we used that since the limit in the fourth line exists (the LHS) then $B\Psi$ must be in the domain of $\varphi(f)$ and more over the limit in the fourth line equals $\varphi(f)B\Psi$ (Proposition 10.14 of \cite{hall2013quantum}). We have shown that $B\text{Dom}(\varphi(f))\subseteq \text{Dom}(\varphi(f))$ and $\varphi(f)B =  B\varphi(f)$ on Dom$(\varphi(f))$. When this happens for any $B\in R_{W}(K_1,K_2)'$ it is said the the unbounded operator $\varphi(f)$ is \textit{affiliated with} $R_{W}(K_1,K_2)$ \cite{BR1}.

We can take $B=JAJ$, which by Tomita-Takesaki Theorem belongs to $R_{W}(K_{1},K_{2})'$, and then we obtain
\begin{equation}\label{dos}
	JAJ\varphi(f)\subseteq \varphi(f)JAJ,
\end{equation}
from which 
\begin{align}
    J\varphi(f)JA\Omega&=J\varphi(f)JAJ^{2}\Omega\nonumber\\
    &=J\varphi(f)(JAJ)\Omega\nonumber\\
    &=J(JAJ)\varphi(f)\Omega\nonumber\\
    &=A J \varphi(f)J \Omega \nonumber\\
    &=A \varphi(\hat{J}f)\Omega
    \label{JphiJ}
\end{align} 
where we have used repeatedly $J^2=1$, $J\Omega=\Omega$, and in the last line we used \eqref{phi}. Note that $\hat{J}f$ belongs to $K_2^\perp$ by \eqref{JKs}, and that $A \in R_{W}(K_1,K_2) \subseteq R_{W}(K_2^\perp,K_1^\perp)'$ by locality. Repeating the same arguments as above we can say that $\varphi(\hat{J}f)$ is affiliated with $R_{W}(K_2^\perp,K_1^\perp)$, and in particular $A \varphi(\hat{J}f) \subset \varphi(\hat{J}f) A$. Then, \eqref{JphiJ} states that $J\varphi(f)JA\Omega =\varphi(\hat{J}f) A\Omega$ which is equivalent to 
\begin{equation}\label{restrictedeq}
    J\varphi(f)J|_{R_{W}(K_1,K_2)\Omega} =\varphi(\hat{J}f)|_{R_{W}(K_1,K_2)\Omega}
\end{equation}
We need to show now that one can remove the restrictions on both sides to reach complete equality and end  the proof. The idea to do so is to take closure to this equation, but this is tricky. In \cite{eckmann1973application} the authors do not give much detail and we believe this is a crucial step, so we now give the necessary details of how we approach it.

We start with the RHS. The aim is to show that $R_{W}(K_1,K_2)\Omega$ is a core for $\varphi(\hat{J}f)$. This is equivalent to show that $\varphi(\hat{J}f)|_{R_{W}(K_1,K_2)\Omega}$ is essentially self-adjoint, since its closure would then be a self-adjoint extension, which by uniqueness of self-adjoint extensions has to be $\varphi(\hat{J}f)$.  Because of this we want to show that $R_{W}(K_1,K_2)\Omega$ is a (total) set of analytic vectors\footnote{An analytic vector for some operator $A$ is a vector $\Psi$  such that $\sum_{n\geq 0} ||A^n \Psi|| \frac{t^n}{n!} <\infty$ for some $t>0$. We shall use that $\Omega$ is an analytic vector for $\varphi(f)$, and this is so because actually any finite-particle vector is analytic, as shown in Theorem X.41 of \cite{reedsimonvol2}.} which by Nelson's theorem (Theorem X.39 in \cite{reedsimonvol2}) implies that  $\varphi(\hat{J}f)|_{R_{W}(K_1,K_2)\Omega}$ is essentially self-adjoint. So let us consider $U(t \hat{J}f) R_{W}(K_1,K_2) \Omega$, with $t\in\mathbb{R}$, and see where we arrive,
\begin{align*}
    U(t\hat{J}f) R_{W}(K_1,K_2) \Omega&=R_{W}(K_1,K_2)U(t\hat{J}f)\Omega\\
    &=R_{W}(K_1,K_2)\sum_{n\geq 0} \frac{(i t \varphi(\hat{J}f))^n}{n!}\Omega\\
    &=\sum_{n\geq 0} \frac{(i t )^n}{n!}\varphi(\hat{J}f)R_{W}(K_1,K_2)\varphi(\hat{J}f)^{n-1}\Omega\\
    &=\sum_{n\geq 0} \frac{(i t \varphi(\hat{J}f))^n}{n!}R_{W}(K_1,K_2)\Omega\\
\end{align*}
where in the first line we used that $U(t\hat{J}f)$ belongs to $R_{W}(K_1,K_2)'$, in the second line that $\Omega$ is an analytic vector for $\varphi(\hat{J}f)$, in the third line that $\varphi(\hat{J}f)$ is affiliated with $R_{W}(K_2^\perp,K_1^\perp)$ and in the last line we just repeated the previous step. This result says first that it makes sense to apply $n$ times $\varphi(\hat{J}F)$ to any element of  $R_{W}(K_1,K_2)$ and second and most important that the series converges for any $t\in\mathbb{R}$. Namely, if $\Psi\in R_{W}(K_1,K_2)\Omega$, 
\begin{equation}
    \sum_{n\geq 0} \frac{(i \varphi(\hat{J}f))^n \Psi}{n!}t^n =U(t\hat{J}f)\Psi\in \mathcal{F}, \qquad t \in \mathbb{R}
\end{equation}
The convergence of this power series for any $t$ implies  that the radius of  convergence is $\infty$, and then it actually converges uniformly and absolutely\footnote{See Theorem 2 in Section 2.4 of Chapter 2 in \cite{Ahlfors1966}. We are just adapting the proof there to our case, by taking the Hilbert norm instead.}, which means that $\Psi$ is an analytic vector for $\varphi(\hat{J}f)$.  This completes the argument to show that the closure of $\varphi(\hat{J}f)|_{R_{W}(K_1,K_2)\Omega}$ in the RHS in \eqref{restrictedeq} is precisely $\varphi(\hat{J}f)$.

With the previous analysis completed, it is easier to show what the closure of the LHS of \eqref{restrictedeq} is. First note that since $\varphi(f)$ is self-adjoint and so is $J$, then $J\varphi(f)J$ is also self-adjoint and is obviously an extension of the LHS. But we now know that the closure of   $J\varphi(f)J|_{R_{W}(K_1,K_2)\Omega}$ is $\varphi(\hat{J}f)$ which is self-adjoint. Then, $\varphi(\hat{J}f)$ is also a self-adjoint extension of $\varphi(\hat{J}f)|_{R_{W}(K_1,K_2)\Omega}$, and by uniqueness of self-adjoint extensions, it must be that $J\varphi(f)J=\varphi(\hat{J}f)$. The same reasoning can be applied to $\pi(g)$.

\end{proof}

\subsection{A counterexample of the duality}\label{contra}

The counterexample we present for the Haag duality must be necessarily related to a region that  is not the causal enveolope $C(B)$ of a region $B$ in a Cauchy surface. It was originally constructed in \cite{araki1964neumann}.

First we establish some facts that will be useful in what follows. The first one is that the second quantization map is an injective map. To fix ideas consider it in the Segal form as in \eqref{segalsecmap}. If $H_{1}\subseteq H_{2}$ then it is clear that $R_{S}(H_{1})\subseteq R_{S}(H_{2})$ (isotony). Now we prove the reciprocal of isotony. It will be useful to write the CCR in Segal's form \eqref{segalrep} as
\begin{align}
	\left[W(\eta),W(\mu)\right]=W(\eta)W(\mu)\left(1-e^{i(\mu,\beta \eta)_H}\right).\label{CCRreloaded}
\end{align}     
In order to prove that $H_{1}\cancel{\subseteq} H_{2}$ implies $R_{S}(H_{1})\cancel{\subseteq} R_{S}(H_{2})$, suppose $H_{1}\cancel{\subseteq} H_{2}$. Then $H_{2}^{\perp}\cancel{\subseteq} H_{1}^{\perp}$ and there exists an element $h_{2}\in H_{2}^{\perp}$ with $h_{2}\notin H_{1}^{\perp}$. Then by taking an $h_{1}\in H_{1}$ such that $(h_{2},h_{1})_{H}\neq 0$ and using equation \eqref{CCRreloaded} with $\eta=\beta h_{2}$ and $\mu=\lambda h_1$ it is possible to find a $\lambda\in \mathbb{R}$ such that $\left[W(\beta h_{2}),W(\lambda h_{1})\right]\neq 0$. Hence $W(\beta h_{2})\notin R_{S}(H_{1})'$. 
On the other hand by applying again equation \eqref{CCRreloaded} taking $\eta=\beta h_{2}$ and $\mu\in H_{2}$ we obtain from $(\mu,h_{2})=0$ that $W(\beta h_{2})\in R_{S}(H_{2})'$. Now suppose that actually $R_S(H_1)\subseteq R_S(H_2)$, which implies $R_S(H_2)'\subseteq R_S(H_1)'$. But this is in contradiction with what we just showed: $W(\beta h_{2})\in R_{S}(H_{2})'$ and $W(\beta h_{2})\notin R_{S}(H_{1})'$.  
Therefore two subspaces $H_1,H_2$ of $H$ satisfy  $R_{S}(H_{1})=R_{S}(H_{2})$ if and only if $H_{1}=H_{2}$.

The following formula will be useful too. In order to derive it,  recall the inner product in $H$ \eqref{prodintespaciales} and the definition of $\beta$ \eqref{betadefH}.  Given $h_{1},h_{2}\in H$
\begin{align}
	(h_{1},\beta h_{2})_H&=-\Im{(h_{1},h_{2})_{L}}=-\Im{2i\int_{\mathbb{R}^{4}\times \mathbb{R}^{4}} h_{1}(x)\Delta^{(+)}(x-y)h_{2}(y)dx dy}\notag \\
	&=-\int_{\mathbb{R}^{4}\times \mathbb{R}^{4}}h_{1}(x)2\Re{\Delta^{(+)}(x-y)}h_{2}(y) dy dx \notag \\
	&=-\int_{\mathbb{R}^{4}\times \mathbb{R}^{4}}h_{1}(x)\Delta(x-y)h_{2}(y) dy dx=-\int_{\mathbb{R}^{4}}h_{1}(x)F_{h_{2}}(x)dx,
	\label{producto}
\end{align}
where we used the causal propagator $\Delta=2\Re{\Delta^{(+)}}$ and equation \eqref{efe}. 

As we mentioned, the counterexample we present for the Haag duality must be necessarily related to a region that  is not a causal diamond or a wedge. The proposed region is 
$$B=C(T_{1})\cup C(T_{2})$$ 
where we call $T_{1}$ to the open interval in the time axis $(t_{1},t_{2})$ where $t_{2}>t_{1}>0$, $T_{2}:=-T_{1}$, and $C(T_{i})=T_{i}''$ (the causal complement of the causal complement, namely the causal completion, as defined in \eqref{prima}). Then region $B$ is the union of two timelike-separated diamonds and its causal complement $B'$ is depicted in Figure \ref{be}.
\begin{figure}[ht]
	\begin{center}
		\includegraphics[scale=0.6]{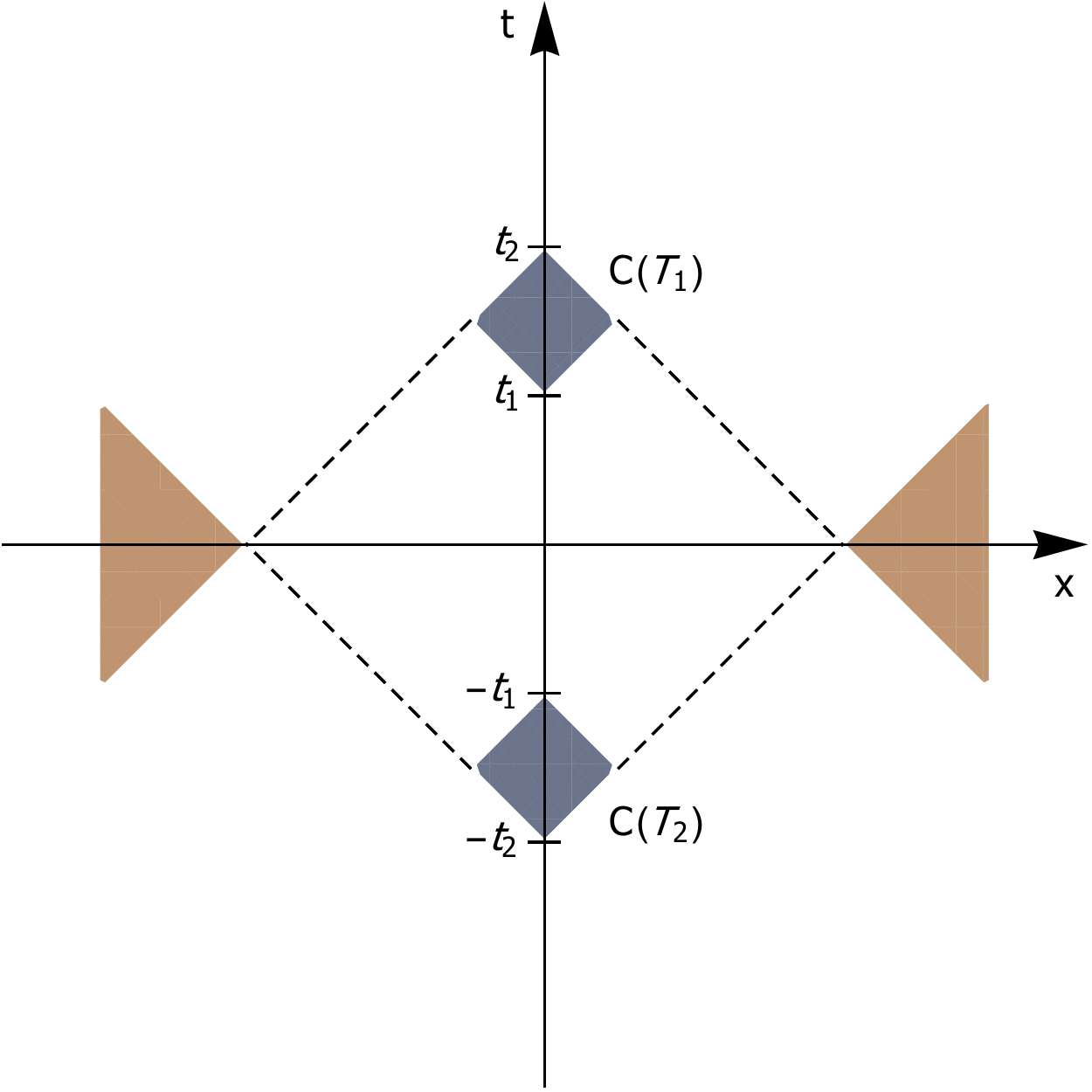}
		\caption{Diagram of  region $B=C(T_{1})\cup C(T_{2})$ and its causal complement $B'$ in brown.}
		\label{be}
	\end{center}
\end{figure}

We define $C(T)$ as before but with $T=(-t_{2},t_{2})$, again an interval on the time axis. It is easy to see that $B'=C(T)'$, see Figure \ref{commutants}. As $C(T)$ is a causal diamond we know Haag duality holds and $N_{S}(C(T))=N_{S}(C(T)')'=N_{S}(B')'$. So, in order to prove that Haag duality fails for region $B$, it is  enough to show $N_{S}(C(T))\neq N_{S}(B)$.  But from the previous observation about the injectivity property of the second qunatization map, it is equivalent to $\mathsf{S}_{S}(C(T))\neq \mathsf{S}_{S}(B)$. To show this we construct a function that does not belong to $\mathsf{S}_{S}(C(T))^{\perp}$ but belongs to $\mathsf{S}_{S}(B)^{\perp}=\mathsf{S}_{S}(C(T_{1})\cup C(T_{2}))^{\perp}$.

\begin{figure}[!h]
\centering
\begin{tabu}to \textwidth {X[c]X[c]}
  \includegraphics[width=0.45\textwidth]{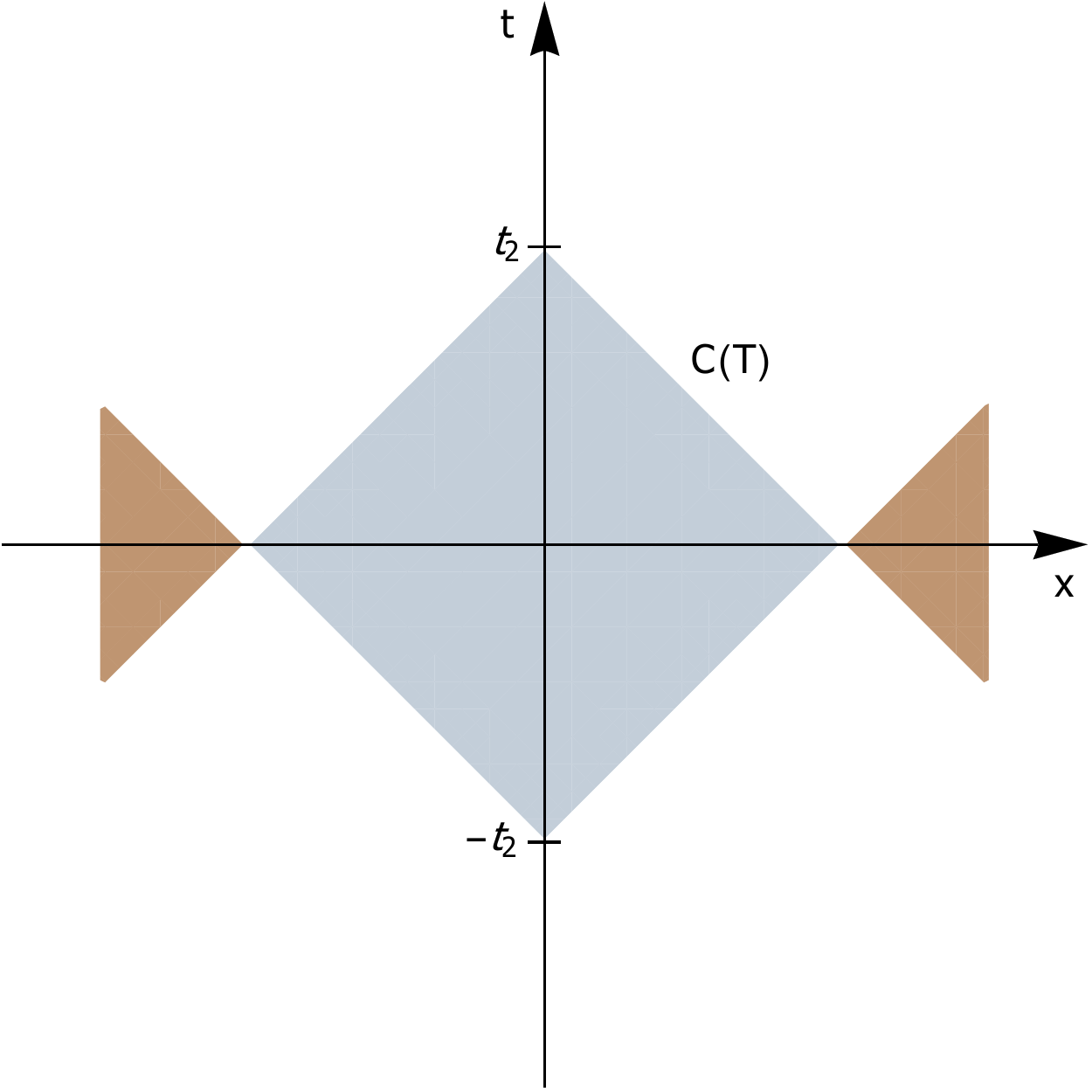} &\includegraphics[width=0.45\textwidth]{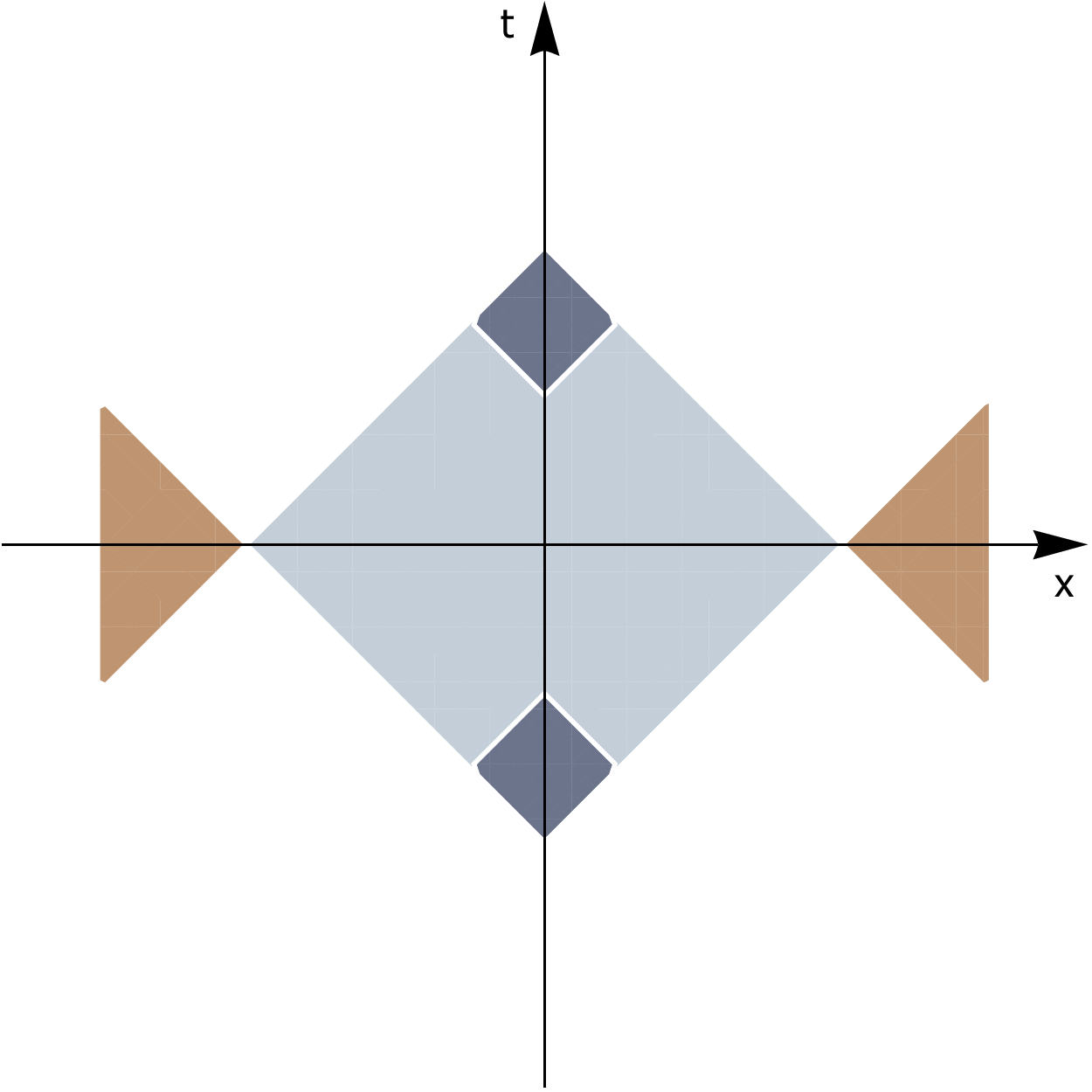}
  \end{tabu}
\caption{On the left, a diagram showing region $C(T)$, the causal diamond of $(-t_2,t_2)$, and its causal complement in brown. On the right, $C(T_1)$ and $C(T_2)$ are included, and the causal complement is the same for $B$ and $C(T)$.  }
\label{commutants}
\end{figure}

The first step is to look at the two-dimensional equation 
\begin{equation}\label{kgimaginary}
	\frac{\partial^2 \phi}{\partial x^2}-\frac{\partial^2 \phi}{\partial t^2}-m^{2}\phi=0
\end{equation}
which can be thought as a Klein-Gordon equation where the roles of $x$ and $t$ are inverted and the mass is purely imaginary. This type of equation was studied in detail by Robinett in \cite{robinett1978tachyons} where he showed that causality still holds. That is, for a given initial data $\phi(t,0)$ and $\frac{\partial \phi}{\partial t}(t,0)$ infinitely differentiable functions with compact supports, there exists a unique solution $\phi(t,x)$ also infinity differentiable  that in addition satisfies $\phi(t,x)=0$ for any  $(t,x)$ such that $x^{2}<(t-t')^2$ for all $t'\in \text{supp}(\phi(t,0))\cup \text{supp}(\frac{\partial \phi}{\partial t}(t,0))$. Then taking differentiable initial conditions whose supports are  contained in the interval $(-t_{1},t_{1})$ of the time axis, we get a differentiable solution $\phi_{0}(t,x)$ whose support does not intersect $C(T_{1})$ nor $C(T_{2})$, as Figure  \ref{sup} schematically shows
\begin{figure}[!h]
	\begin{center}
	\includegraphics[scale=0.6]{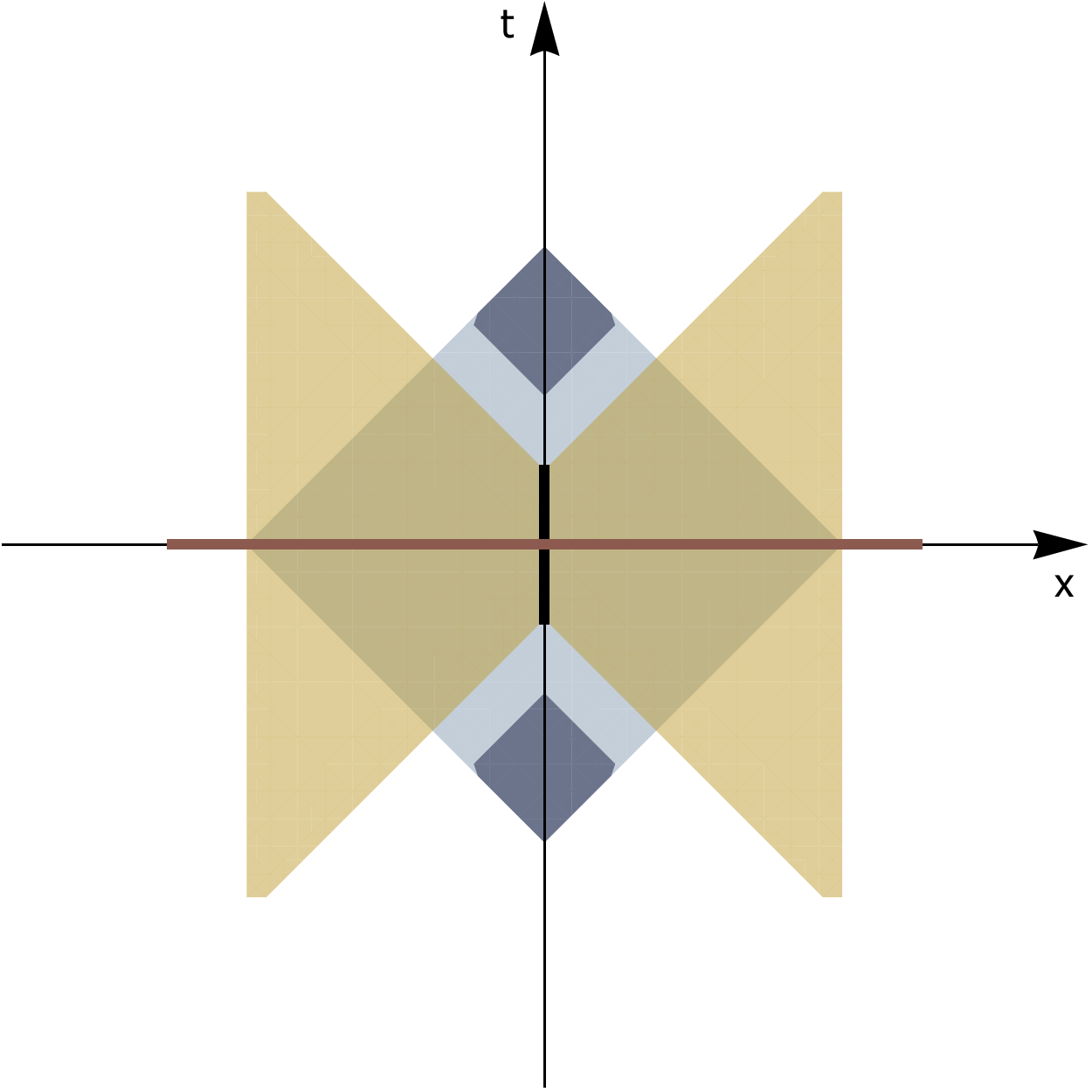}
		\caption{The initial conditions of $\phi_{0}$ are supported in the thick black line in the time axis. Its support is depicted as the beige region and does not intersect $B$. The horizontal thick clay-colored 
		line represents the compact support of the initial conditions of $G$.}
		\label{sup}
	\end{center}
\end{figure}

Now we define initial conditions to the usual Klein-Gordon problem in four dimensions 
\begin{equation}
    	g_{1}(\vec{x})=\phi_{0}(0,x_{1})h_{1}(\vec{x}),\qquad
	g_{2}(\vec{x})=\frac{\partial \phi_{0}}{\partial t}(0,x_{1})h_{1}(\vec{x}),
\end{equation}
where $h_{1}(\vec{x})$ is an infinitely differentiable function with compact support such that $h_{1}(\vec{x})=1$ if $|\vec{x}|\leq t_{2}$. Then we consider the initial-value problem 
	\begin{equation}
	    (\square+m^2)G(x)=0,\qquad G(0,\vec{x})=g_{1}(\vec{x}),\qquad 
		\frac{\partial G}{\partial x^{0}}(0,\vec{x})=g_{2}(\vec{x}).
		\label{kg2}
	\end{equation}
By uniqueness of the solution on $C(T)$ and as $\phi_{0}(x^{0},x^{1})$ satisfies \eqref{kg2}, then $G(x)=\phi_{0}(x^{0},x^{1})$ on $C(T)$.  In addition, $G$ comes from applying the causal propagator $\Delta$ to a function $g$ in $H$, which we show next. 

We know from appendix \ref{ap3}, we slight variation of   equation \eqref{candidato} changing time zero by an arbitrary time (just following the same steps),  that for any $s\in \mathbb{R}$,
\begin{align}
    G(x)&=\int_{\mathbb{R}^4}\Delta(x-y)\phi_{s}(y)dy \\
    \text{where }\phi_{s}(y)&=-\delta(y^0-s)\frac{\partial G}{\partial y^0}(s,\vec{y})-\delta'(y^0-s)G(s,\vec{y}).\label{hs}
\end{align}
Then by taking a $\chi\in C^{\infty}_{0}(\mathbb{R})$ satisfying $\int_{\mathbb{R}} \chi(s) ds=1$,  we can rewrite $G$ in a useful way
\begin{align*}
    G(x)&=\int_{\mathbb{R}} \chi(s)G(x)ds=\int_{\mathbb{R}} \chi(s) \left( \int_{\mathbb{R}^4}\Delta(x-y)\phi_{s}(y)dy \right) ds \\
    &=\int_{\mathbb{R}^4}\Delta(x-y)\left(\int_{\mathbb{R}} \chi(s)\phi_{s}(y)ds \right) dy=F_{g}(x),
\end{align*}
where 
\begin{equation}\label{g}
    g(y):=\int_{\mathbb{R}} \chi(s)\phi_{s}(y)ds.
\end{equation}
We claim that $g$ is in $H$, even more $g\in C^{\infty}_{0}(\mathbb{R}^{4})$. This follows from the fact that, since $g_1$ and $g_2$ have compact support, our solution $G$ is spacelike-compact (meaning that given $s\in \mathbb{R}$,  the supports of $G(s,\vec{y})$ and $\frac{\partial G}{\partial y^0}(s,\vec{y})$ are compact sets of $\mathbb{R}^{3}$, see \cite{Brunetti:2015vmh} for more details.). Let us see this explicitly,
\begin{align*}
    g(y)&=\int_{\mathbb{R}} \chi(s)\phi_{s}(y)ds \\
    &=-\int_{\mathbb{R}} \chi(s) \delta(y^0-s)\frac{\partial G}{\partial y^0}(s,\vec{y})+\chi(s)\delta'(y^0-s)G(s,\vec{y}) ds \\
    &=\chi'(y^{0})G(y)
\end{align*}
where in the last line we used the weak derivative. Written in this way, it is obvious that $g$ is a smooth function and also that has compact support since $\chi$ is compactly supported and $G$ is spacelike-compact.

 Previously we mentioned we are looking for a function such that does not belong to $\mathsf{S}_{S}(C(T))^{\perp}$ but belongs to $\mathsf{S}_{S}(B)^{\perp}=\mathsf{S}_{S}(C(T_{1})\cup C(T_{2}))^{\perp}$. We assert that this function is $\beta g$ because, given $f\in \mathsf{S}_{S}(C(T_{1})\cup C(T_{2}))$
\begin{equation*}
	(f,\beta g)_{H}=\int_{\mathbb{R}^{4}}f(x)F_{g}(x)dx=\int_{\mathbb{R}^{4}}f(x)G(x)dx=0,
\end{equation*}
where we used  equation \eqref{producto} and the fact that $G(x)$ vanishes on $C(T_{1})\cup C(T_{2})$. On the other hand, by a similar computation $\beta g\notin \mathsf{S}_{S}(C(T))^{\perp}$.

\section*{Acknowledgements}

We would like to thank the members of the HEPTH group at Universtiy of Buenos Aires for valuable discussions and specially an anonymous referee for a careful reading and providing constructive criticism. This work was partially supported by grants PIP and PICT from CONICET and ANPCyT. The work of G.P. is supported by the University of Buenos Aires.

\newpage

\appendix

\section{Creation and annihilation operators: a reminder} \label{ap1}

The inner product in the bosonic Fock space  $\mathfrak{H}_{T}(L)=\bigoplus_{n=0}^{\infty}\bigodot_{i=0}^{n}L$ is defined from the inner product in the one particle space $L$. If $\alpha=\odot_{i=1}^{n}\alpha_{i}$ and $\beta=\odot_{j=1}^{n}\beta_{j}$, then:
\begin{equation*}
	(\alpha,\beta)_{\mathfrak{H}_{T}(L)}=\frac{1}{n!}\sum_{\sigma\in \mathbb{S}_{n}}\prod_{i=1}^{n}(\alpha_{i},\beta_{\sigma(i)})_{L},
\end{equation*}
 and the inner product for elements of different degree vanish. Here $( , )_{L}$ is the inner product in $L$. By using this definition and the definitions \eqref{creayaniq}, we can easily see that the operators are mutually adjoint (over the dense set $\mathcal{D}=\cup_{N=0}^{\infty}\bigoplus_{n=0}^{N}\mathcal{F}^{(n)}$) 

\begin{align*}
	\big(\odot_{i=1}^{n}h_{i},a(f)(\odot_{j=1}^{n+1}g_{j})\big)_{\mathfrak{H}_{T}(L)}&=\big(\odot_{i=1}^{n}h_{i},\frac{1}{\sqrt{n+1}}\sum_{j=1}^{n+1}(f,g_{j})_{L} \odot_{r=1, r\neq j}^{n+1} g_{r} \big)_{\mathfrak{H}_{T}(L)} \\
	&=\frac{1}{\sqrt{n+1}}\sum_{j=1}^{n+1}(f,g_{j})_{L}\big(\odot_{i=1}^{n}h_{i},\odot_{r=1, r\neq j}^{n+1} g_{r} \big)_{\mathfrak{H}_{T}(L)} \\
	&=\frac{1}{n!\sqrt{n+1}}\sum_{j=1}^{n+1}(f,g_{j})_{L}\sum_{\sigma\in \mathbb{S}_{n}}\prod_{i=1}^{n}(h_{i},g_{\sigma(i)})_{L},
\end{align*}
where we think $S_{n}$ as bijective maps from $\{1,\cdots, n\}$ to $\{1, \cdots ,\hat{j}, \cdots n+1\}$ in each term. Calling $l_{1}=f$ y $l_{i+1}=h_{i}$ we can join all together in one sum
\begin{align*}
	 	\big(\odot_{i=1}^{n}h_{i},a(f)(\odot_{j=1}^{n+1}g_{j})\big)_{\mathfrak{H}_{T}(L)}&=\frac{1}{n!\sqrt{n+1}}\sum_{\sigma\in \mathbb{S}_{n+1}}\prod_{i=1}^{n+1}(l_{i},g_{\sigma(i)})_{L}. 
\end{align*}
Additionally,
\begin{align*}
		\big(a^{\ast}(f)(\odot_{i=1}^{n}h_{i}),\odot_{j=1}^{n+1}g_{j}\big)_{\mathfrak{H}_{T}(L)}&=(\sqrt{n+1}f\odot_{i=1}^{n}h_{i},\odot_{j=1}^{n+1}g_{j}\big)_{\mathfrak{H}_{T}(L)} \\
		&=\frac{\sqrt{n+1}}{(n+1)!}\sum_{\sigma \in \mathbb{S}_{n+1}}\prod_{i=1}^{n+1}(l_{i},g_{\sigma(i)})_{L}, 
\end{align*}
where we used the same notation as above. And as $\frac{\sqrt{n+1}}{(n+1)!}=\frac{1}{n!\sqrt{n+1}}$ we conclude $a^{\ast}(f)$ is the adjoint of $a(f)$ over $\mathcal{D}=\cup_{N=0}^{\infty}\bigoplus_{n=0}^{N}\mathcal{F}^{(n)}$.

Now we prove \eqref{conmut}. Let us consider $\odot_{i=1}^{n}h_{i}$ arbitrary in $\bigotimes_{i=1}^{n}L$,

\begin{align*}
	a(f)a^{\ast}(g)(\odot_{i=1}^{n}h_{i})&=a(f)\left[ \sqrt{n+1}g\odot_{i=1}^{n}h_{i} \right] \\
	&=\sqrt{n+1}a(f)\left( \odot_{i=1}^{n}l_{i} \right) \\
	&=\cancel{\sqrt{n+1}}\frac{1}{\cancel{\sqrt{n+1}}}\sum_{j=1}^{n+1}(f,l_{j})_{L}(l_{1}\odot\cdots\odot \hat{l_{j}}\odot \cdots \odot l_{n+1}  ) \\
	&=\sum_{j=1}^{n+1}(f,l_{j})_{L}(l_{1}\odot\cdots\odot \hat{l_{j}}\odot \cdots \odot l_{n+1} ),
\end{align*}
again using a simplifying notation, $l_{1}=g$ y $l_{i+1}=h_{i}$. On the other hand

\begin{align*}
	a^{\ast}(g)a(f)(\odot_{i=1}^{n}h_{i})&=a^{\ast}(g)\left[ \frac{1}{\sqrt{n}}\sum_{k=1}^{n}(f,h_{j})_{L} (h_{1}\odot\cdots\odot \hat{h_{j}}\odot \cdots \odot h_{n+1} ) \right] \\
        &=\sqrt{n}\sum_{k=1}^{n}(f,h_{j})_{L} a^{\ast}(g)(h_{1}\odot\cdots\odot \hat{h_{j}}\odot \cdots \odot h_{n+1} )    \\
        &=\frac{1}{\cancel{\sqrt{n}}}\sum_{k=1}^{n}(f,h_{j})_{L} \cancel{\sqrt{n}}(g \odot h_{1}\odot\cdots\odot \hat{h_{j}}\odot \cdots \odot h_{n+1} ).
\end{align*}
Substracting the last two expressions we taht all except the $j=1$ disappear, obtaining $(f,g)_{L}(h_{1}\odot \cdots \odot h_{n})$, that is \eqref{conmut}.

\section{An example of how  \texorpdfstring{$\beta$ acts}{}}\label{ejemplo}

As explained in the body of the article, we define can define $\beta$ on $H$ by its action on $H\cap \mathcal{S}(\mathbb{R}^{4},\mathbb{R})$ by $g(x)=(\beta f)(x)$ with,
\begin{equation*}
    \mathscr{F}(g)(p)=i\mathscr{F}(f)(p)\eta(p^0), 
\end{equation*}
where $\eta \in {C}^\infty(\mathbb{R})$ is odd and such that $\eta(p^0)=1$ for $|p^0|\geq m$. Now let us show with a simple example what $g$ looks like. Let us consider a Gaussian function $f(x^0,\vec{x})=e^{-(x^{0})^{2}-\vec{x}^2}$, then 
\begin{align*}
	g(x)=(\beta f)(x^0,\vec{x})&=-2\int_{0}^{\infty}\eta(p^0)\left[ \int_{-\infty}^{\infty} f(x^{0\prime},\vec{x})\sin(p^0(x^{0\prime}-x^0))dx^{0\prime} \right] dp^{0}\\
	&=2\int_{0}^{\infty}\eta(p^0)\sqrt{\pi}e^{-\frac{(p^0)^2}{4}-\vec{x}^2}\sin(p^0 x^0) dp^{0} \\
	&=2\sqrt{\pi}e^{-\vec{x}^2}\int_{0}^{\infty}\eta(p^0)e^{-\frac{(p^0)^2}{4}}\sin(p^0 x^0) dp^{0} \\
	&=2\sqrt{\pi}e^{-\vec{x}^2}\bigg[\int_{0}^{m}\eta(p^0)e^{-\frac{(p^0)^2}{4}}\sin(p^0 x^0) dp^{0}+\int_{m}^{\infty}e^{-\frac{(p^0)^2}{4}}\sin(p^0 x^0) dp^{0}\bigg].
\end{align*} 
In general $\eta$ is complicated enough so that the first integral cannot be cast in any familiar form. We can instead consider a limiting procedure in order to take $\eta$ as   $\eta(p^0)=\theta(p^0-m)$, for $p^0\geq 0$ and in this way getting rid of the first integral above. Of course, such $\eta$ is not smooth. So to overcome this issue we take a succession of smooth bounded functions $\eta_{n}(p^0)$ such that they vanish in the interval  $[0,m-\frac{1}{n}]$ and they are equal to $1$ for $p^0\geq m$ (and they are extended to the negative real line by the odd property). Note that for one of these functions,
\begin{equation*}
	\int_{0}^{m}\eta_{n}(p^0)e^{-\frac{(p^0)^2}{4}}\sin(p^0 x^0) dp^{0}=\int_{m-\frac{1}{n}}^{m}\eta_{n}(p^0)e^{-\frac{(p^0)^2}{4}}\sin(p^0 x^0) dp^{0},
\end{equation*}
which tends to zero because the integrand in bounded. Then, substituting these $\eta_n$ functions in the expression for $(\beta f)(x^0,x)$ above and taking the limit,
\begin{align}\label{betaf}
	(\beta f)(x^0,x)&=2\sqrt{\pi}e^{-x^2}\int_{m}^{\infty}e^{-\frac{(p^0)^2}{4}}\sin(p^0 x^0) dp^{0} \\
	&= i \pi e^{-((x^{0})^2+x^2)}\big[ erf\big( \frac{1}{2}(m-2ix^0) \big)-erf\big( \frac{1}{2}(m+2ix^0) \big)  \big],
\end{align}
where $erf$ is the error function. Due to the property $erf(z^{\ast})=erf(z)^{\ast}$, the result is indeed real, as expected.

\begin{figure}[ht]
	\centering
	\begin{subfigure}[b]{0.45\linewidth}
		\includegraphics[width=0.7\linewidth]{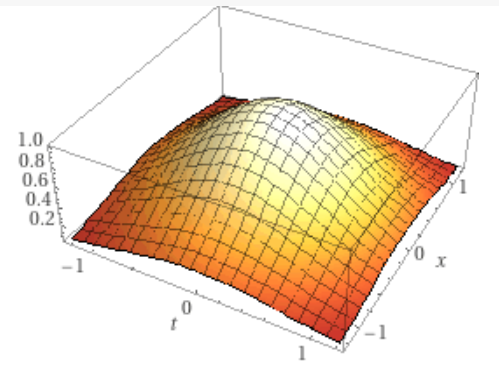}
	\end{subfigure}
	\begin{subfigure}[b]{0.45\linewidth}
		\includegraphics[width=0.8\linewidth]{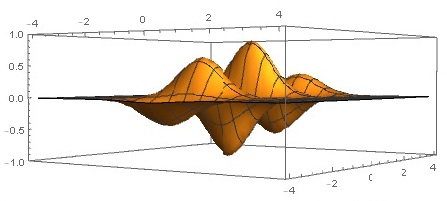}
		\end{subfigure}
	\caption{On the left, the Gaussian function $f(t,x)$. On the right, $g=\beta f$. }
	\label{f}
\end{figure}

In this example it is explicitly seen that starting from $f$ which is an even function of $x^0$, we obtain  $\beta f$ which is odd (see Figures \ref{f} and \ref{faraway}). This is actually a general feature as explained in Subsection \ref{ka}. Just to give more qualitative intuition, we show below in Figure \ref{faraway} a ``far away'' view of $\beta f$ where the rapid decay is evident consistent with $\mathcal{S}(\mathbb{R}^2,\mathbb{R})$, and a lateral view where the odd behavior in $x^0=t$ is also evident.

\begin{figure}[ht]
	\centering
	\begin{subfigure}[b]{0.45\linewidth}
		\includegraphics[width=0.75\linewidth]{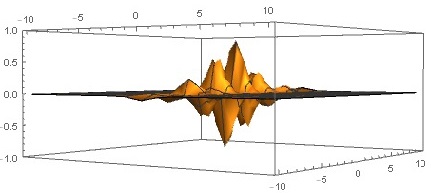}
				\end{subfigure}
	\begin{subfigure}[b]{0.45\linewidth}
		\includegraphics[width=0.75\linewidth]{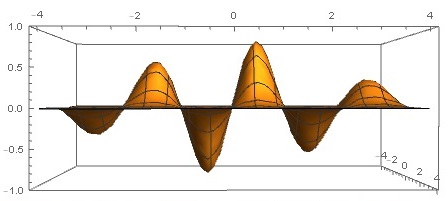}
		\label{impar}
	\end{subfigure}
\caption{On the left, $\beta f$ from  ``far away''. On the right, a  lateral view of $\beta f$ }
\label{faraway}
\end{figure}

\section{More on the initial conditions of the Klein-Gordon equation} \label{ap3}

Given  the real function $f\in C_{0}^{\infty}(\mathbb{R}^{3}) \subset \mathfrak{F}_{\pi}$ and $g\in  C_{0}^{\infty}(\mathbb{R}^{3}) \subset \mathfrak{F}_{\varphi}$, in what follows we are going to seek $h\in C^{\infty}(\mathbb{R}^4,\mathbb{R})$ such that $F_{E(h)}$ is  a smooth function that  solves the following system

\begin{equation}\label{kg}
 \text{KG initial-value problem:}\left\lbrace
	\begin{array}{ll}
		(\square+m^2)F_{E(h)}(x)=0\\
		F_{E(h)}(0,\vec{x})=f(\vec{x}) \\
		\frac{\partial F_{E(h)}}{\partial x^{0}}(0,\vec{x})=g(\vec{x})
	\end{array}
	\right.
\end{equation}
 
We follow  \cite{dimock1980algebras}, where in Corollary 1.1 is stated the relation between initial conditions and the solution. Suspecting that $E$ in that reference can only differ in a global sign with the operator $\Delta$ we have been using, we believe $h(x)$ should be $-\rho'_{0}(f)(x)+\rho'_{1}(g)(x)$. Before showing this is the solution, let us briefly describe the maps $\rho_{0}$ and $\rho_{1}$ and their corresponding pull-backs  $\rho'_{0}$ and $\rho'_{1}$.

Let $\rho_{0}:\mathcal{S}(\mathbb{R}^{4})\to \mathcal{S}(\mathbb{R}^{3})$ be the restriction to $x^{0}=0$ and $\rho_{1}:\mathcal{S}(\mathbb{R}^{4})\to \mathcal{S}(\mathbb{R}^{3})$ the derivative with respect to $x^{0}$ evaluated on the $x^{0}=0$ hypersurface. The pull-backs of these maps act by precomposition in the space of tempered distributions $\rho'_{0}:\mathcal{S}(\mathbb{R}^{3})'\to \mathcal{S}(\mathbb{R}^{4})'$, namely $\rho'_{0}(v)=v\circ\rho_{0}$ for any $v \in \mathcal{S}(\mathbb{R}^{3})'$. Similarly for $\rho_{1}'$. We can view the initial conditions as elements of the dual space $\mathcal{S}(\mathbb{R}^{3})'$, since $\mathcal{S}(\mathbb{R}^{3}) \subset \mathcal{S}(\mathbb{R}^{3})'$.

Given,  $h\in\mathcal{S}(\mathbb{R}^{4})$ and the initial condition $g$ in \eqref{kg}, 
\begin{equation*}
	(\rho_{0}'(g))(h)=(g\circ \rho_{0})(h)= g(h|_{x^{0}=0})=\int g(\vec{x}) h(0,\vec{x}) d^{3}x=\int h(x) \delta(x^{0})g(\vec{x})d^{4}x,
\end{equation*}
 from where we read  $\rho_{0}'(g)(x)=\delta(x^{0})g(\vec{x})$. Similarly, with the other initial condition and $\rho_{1}'$,
 \begin{align}
 	(\rho_{1}'(f))(h)=(f\circ \rho_{1})(h)= f\left(\frac{\partial h}{\partial x^{0}}(0,\vec{x})\right)=\int_{\mathbb{R}^{3}} f(\vec{x}) \frac{\partial h}{\partial x^{0}}(0,\vec{x}) d^{3}x,
 \end{align}
so we see that  $\rho_{1}'(f)=-\delta'(x^{0})f(\vec{x})$, where $'$ means weak derivative. Now we define the distribution 
\begin{equation}
	p(x)=-\delta(x^{0})g(\vec{x})-\delta'(x^{0})f(\vec{x}),
\end{equation}
and follow the steps of section \ref{contra}, constructing $p_{s}$ as in equation \eqref{hs} and defining an smooth function
\begin{equation}\label{candidato}
    h(x):=\int_{\mathbb{R}} \chi(s)p_{s}(x)ds,
\end{equation}
as in equation \eqref{g} in order to obtain our solution $F_{E(h)}(x)$, with $h$ a smooth function.

We can now show that $F_{E(h)}$ is the solution to \eqref{kg} given by \eqref{efe} with $h$ as in \eqref{candidato}. That it is a solution is immediate since we are propagating with the causal Klein-Gordon propagator $\Delta$. Let us then evaluate $F_{E(h)}$ at $x^{0}=0$,
\begin{align*}
	&F_{E(h)}(0,\vec{x})=\frac{-i}{(2\pi)^{\frac{3}{2}}}\int_{\mathbb{R}^{4}} e^{i\vec{p}.\vec{x}}\delta(p^2-m^2)\sgn(p^{0})\mathscr{F}(h)(p) d^{4}p \\
	&=\frac{-i}{(2\pi)^{\frac{3}{2}}}\int_{\mathbb{R}^{4}} e^{i\vec{p}.\vec{x}}\delta(p^2-m^2)\theta(p^{0})\mathscr{F}(h)(p)d^{4}p+\frac{i}{(2\pi)^{\frac{3}{2}}}\int_{\mathbb{R}^{4}} e^{i\vec{p}.\vec{x}}\delta(p^2-m^2)\theta(-p^{0})\mathscr{F}(h)(p)d^{4}p,\\
	&=\frac{-i}{(2\pi)^{\frac{3}{2}}}\int_{\mathbb{R}^{3}} e^{i\vec{p}.\vec{x}}\mathscr{F}(h)(\omega_p,\vec{p})\frac{d^{3}p}{2\omega_p}+\frac{i}{(2\pi)^{\frac{3}{2}}}\int_{\mathbb{R}^{3}} e^{i\vec{p}.\vec{x}}\mathscr{F}(h)(-\omega_p,\vec{p})\frac{d^{3}p}{2\omega_p} \\
	&=\frac{-i}{(2\pi)^{\frac{3}{2}}}\int_{\mathbb{R}^{3}} e^{i\vec{p}.\vec{x}}\frac{\mathscr{F}(h)(\omega_p,\vec{p})-\mathscr{F}(h)^{\ast}(\omega_p,-\vec{p})}{2}\frac{d^{3}p}{\omega_p} \\
	&=\frac{1}{(2\pi)^{\frac{3}{2}}}\int_{\mathbb{R}^{3}} e^{i\vec{p}.\vec{x}}\frac{1}{i\omega_p}\mathscr{F}(h_{-})(\omega_p,\vec{p})d^{3}p =\frac{1}{(2\pi)^{\frac{3}{2}}}\int_{\mathbb{R}^{3}} e^{i\vec{p}.\vec{x}}\mathscr{F}(\delta_{1}h_{-})(\vec{p})d^{3}p \\
	&=\mathscr{F}^{-1}\mathscr{F}(\delta_{1}h_{-})(\vec{x})=\delta_{1}h_{-}(\vec{x}). 
\end{align*}
And at the same time this should be equal to the initial condition $f(\vec{x})$. Just to check things work properly, let us show this explicitly,
\begin{align*}
	&F_{E(h)}(0,\vec{x})=\int_{\mathbb{R}^{4}} \frac{-1}{(2\pi)^{3}}\int_{\mathbb{R}^{3}} \sin(\omega_{p}(x^{0}-y^0)-\vec{p}.(\vec{x}-\vec{y}))\frac{d^3p}{\omega_{p}}h(y)d^4y\bigg|_{x^0=0} \\
	&=\frac{1}{(2\pi)^{3}}\int_{\mathbb{R}^{4}} \int_{\mathbb{R}^{3}} \sin(\omega_{p}y^{0})\cos(\vec{p}(\vec{x}-\vec{y}))\frac{h(y)}{\omega_p}d^{3}pd^4y\\
	&=\frac{1}{(2\pi)^{\frac{3}{2}}} \int_{\mathbb{R}^{4}} \int_{\mathbb{R}^{3}} \sin(\omega_{p}y^{0})\cos(\vec{p}(\vec{x}-\vec{y}))\frac{-\delta(y^{0})g(\vec{y})-\delta'(y^{0})f(\vec{y})}{\omega_p}d^{3}pd^4y \\
	&=\frac{-1}{(2\pi)^{3}} \int_{\mathbb{R}^{4}} \int_{\mathbb{R}^{3}} \sin(\omega_{p}y^{0})\cos(\vec{p}(\vec{x}-\vec{y}))\frac{\delta'(y^{0})f(\vec{y})}{\omega_p}d^{3}pd^4y \\
	&=\frac{1}{(2\pi)^{3}} \int_{\mathbb{R}^{4}} \int_{\mathbb{R}^{3}} \cancel{\omega_p}\cos(\omega_{p}y^{0})\cos(\vec{p}(\vec{x}-\vec{y}))\frac{\delta(y^{0})f(\vec{y})}{\cancel{\omega_p}}d^{3}pd^4y \\
	&=\frac{1}{(2\pi)^{3}} \int_{\mathbb{R}^{3}} \int_{\mathbb{R}^{3}} \cos(\vec{p}(\vec{x}-\vec{y}))f(\vec{y})d^{3}pd^{3}y \\
	&=\Re \left\{ \frac{1}{(2\pi)^{3}} \int_{\mathbb{R}^{3}} \int_{\mathbb{R}^{3}} e^{i\vec{p}(\vec{x}-\vec{y})}f(\vec{y})d^{3}pd^{3}y  \right\} \\
	&=\Re \left\{\mathscr{F}^{-1} \left[ \mathscr{F}(f) \right] (\vec{x})\right\}=\Re\left[f(\vec{x})\right]=f(\vec{x})
\end{align*}
In the second line we dropped one integral since it is odd in $\vec{p}$. This confirms that $\delta_1 h_-=f$. 
An identical computation shows that $\delta_0(- h_+)=g$.

\section{Generic position of subspaces}\label{ap4}

The aim of this Appendix is to show that the subspaces of $\mathfrak{F}_\varphi$ defined in the first quantization map \eqref{mapauno} are in generic position, namely that,
\begin{align*}
	 \mathsf{S}_{R}(B)\cap \mathsf{S}_{I}(B)&=\{0\}  \\
	 \mathsf{S}_{R}(B)\cap \mathsf{S}_{I}(B)^{\perp}&=\{0\}  \\
	 \mathsf{S}_{R}(B)^{\perp}\cap \mathsf{S}_{I}(B)&=\{0\}  \\ 
	 \mathsf{S}_{R}(B)^{\perp}\cap \mathsf{S}_{I}(B)^{\perp}&=\{0\}.
\end{align*}
In order to show the second and third lines, it is enough to prove 
$\mathsf{S}_{R}(B)^{\perp}\vee \mathsf{S}_{I}(B)=\mathfrak{F}_{\varphi}$ and $ \mathsf{S}_{R}(B)\vee \mathsf{S}_{I}(B)^{\perp}=\mathfrak{F}_{\varphi}$ (where $A\vee B=\overline{A+B}$). The former is deduced by first noticing that $\mathsf{S}_{R}(B)^{\perp}\vee \mathsf{S}_{I}(B)=\mathsf{S}_{I}(B^{c})\vee \mathsf{S}_{I}(B)$ thanks to  Theorem \ref{haagteoremauno},  and then by  definition of $\mathsf{S}_{I}$ one concludes $\mathsf{S}_{R}(B)^{\perp}\vee \mathsf{S}_{I}(B)=\bar{\beta}(j_{1}^{-1}(L^2(B^{c}))\vee j_{1}^{-1}(L^2(B)))=\bar{\beta}j_{1}^{-1}(L^2(\mathbb{R}^{3}))=\bar{\beta}\mathfrak{F}_{\pi}=\mathfrak{F}_{\varphi}$. As for $ \mathsf{S}_{R}(B)\vee \mathsf{S}_{I}(B)^{\perp}=\mathfrak{F}_{\varphi}$, it can be shown analogously by using the other relation in Theorem \ref{haagteoremauno}. 
Regarding the remaining two intersections above, we need the following result. 

\begin{lemma}[Lemma 5 in \cite{araki1964neumann}]\label{lemaap}
	If $f\in \mathsf{S}_{R}(B)$ in not zero, then $\bar{\beta}f$ never vanishes in the entire neighborhood of any point in $\left(supp(f)
	\right)^{c}$.
\end{lemma}

In order to prove this, we consider the function $\displaystyle{F(z)=\frac{1}{(2\pi)^{\frac{3}{2}}}\int_{\mathbb{R}^{3}} e^{i(\vec{p}\cdot\vec{z}-\omega_{p}z^{0})}\mathscr{F}(f)(\vec{p})\omega^{-1}_{p} d^{3}p}$. For  $x\in \mathbb{R}^{4}$ it is a solution of Klein-Gordon equation $(\square_{x}+m^{2})F(x)=0$ with initial conditions $F(0,\vec{x})=(\bar{\beta}f)(\vec{x})$ and $\displaystyle{\frac{\partial F}{\partial x^{0}}}(0,\vec{x})=-if(\vec{x})$. Moreover,  $F(z)$ is analytic for $\Im z$ in the past light cone, since the integrand is analytic and the integral converges. Now, let us consider a point $\vec{y}\in supp(f)^{c}$ and assume that $\bar{\beta}f$ vanishes in any real three-dimensional neighborhood of $\vec{y}$. Then, since $f$ also vanishes there, the solution to the Klein-Gordon equation $F(x)$ vanishes on the real four-dimensional nieghborhood of $(0,\vec{y})$. Then, by the identity theorem, $F(z)$ vanishes identically, implying $f(x)=0$.

We can now show that $\mathsf{S}_{R}(B)\cap \mathsf{S}_{I}(B)=\{0\}$. This is equivalent to $\bar{\beta}\mathsf{S}_{R}(B)\cap \bar{\beta}\mathsf{S}_{I}(B)=\{0\}$, since $\bar{\beta}$ is unitary. Let us consider then some $g\in \bar{\beta}\mathsf{S}_{R}(B)\cap \bar{\beta}\mathsf{S}_{I}(B)$. As $g\in \bar{\beta}\mathsf{S}_{I}(B)=j_1^{-1}L_{2}(B)$,  $g$ must vanish in all $B^{c}$, but at the same time, by the previous lemma and the fact that $g\in \bar{\beta}\mathsf{S}_{R}(B)$, $g$  cannot be zero in an entire neighborhood of any point in $B^{c}$ unless it vanishes identically. Therefore $\bar{\beta}\mathsf{S}_{R}(B)\cap \bar{\beta}\mathsf{S}_{I}(B)=\{0\}$. As for the remaining equality $ \mathsf{S}_{R}(B)^{\perp}\cap \mathsf{S}_{I}(B)^{\perp}=\{0\}$, we can apply Theorem \ref{haagteoremauno} in order to write it as  $\mathsf{S}_{I}(B^{c})\cap \mathsf{S}_{R}(B^{c})=\{0\}$ which is equivalent to the first one, which we  just proved, exchanging $B$ with $B^{c}$.

\bibliographystyle{toine}
\bibliography{biblio}{}

\end{document}